\documentclass[envcountsame,runningheads]{llncs}
\usepackage{amsmath,amssymb,amsfonts}
\usepackage{xcolor}
\usepackage{cite}
\usepackage{enumerate}
\usepackage{graphicx}
\usepackage{epstopdf}
\usepackage{bm}
\usepackage{array}
\usepackage{wrapfig}
\usepackage[textsize=footnotesize]{todonotes}
\usepackage[algo2e]{algorithm2e}

\usepackage{hyperref}
\hypersetup{        
    bookmarksnumbered=true,     
    bookmarksopen=true,         
    bookmarksopenlevel=1,       
    colorlinks=true,            
    pdfstartview=Fit,           
    pdfpagemode=UseOutlines,    
    linkcolor={red!80!black},
    citecolor={green!60!black},
    urlcolor={blue}
}

\usepackage[capitalize]{cleveref}
\crefname{equation}{}{}
\crefname{lemma}{Lemma}{Lemmas}
\crefname{example}{Ex.}{Ex.}
\crefname{section}{Sect.}{Sect.}
\crefname{appendix}{App.}{App.}
\crefname{definition}{Def.}{Def.}
\crefname{theorem}{Thm.}{Thm.}
\crefname{corollary}{Cor.}{Cor.}
\crefname{algorithm}{Alg.}{Alg.}

\makeatletter
\@ifundefined{theorem}{%
  \spnewtheorem{theorem}{Theorem}[section]{\bfseries}{\itshape}
}
\@ifundefined{corollary}{%
  \spnewtheorem{corollary}[theorem]{Corollary}{\bfseries}{\itshape}
}
\@ifundefined{lemma}{%
  \spnewtheorem{lemma}[theorem]{Lemma}{\bfseries}{\itshape}
}
\@ifundefined{definition}{%
  \spnewtheorem{definition}[theorem]{Definition}{\bfseries}{\normalfont}
}

\@ifundefined{example}{%
  \spnewtheorem{example}[theorem]{Example}{\itseries}{\normalfont}
}

\@ifundefined{remark}{%
  \spnewtheorem{remark}[theorem]{Remark}{\bfseries}{\itshape}
}

\@ifundefined{algorithm}{%
  \spnewtheorem{algorithm}[theorem]{Algorithm}{\bfseries}{\itshape}
}
\makeatother

\everymath{\displaystyle}

\usepackage{thmtools, thm-restate}

\newcommand{\oldcomment}[1]{}
\newcommand{\superfluous}[1]{}
\newcommand{\IZ}{\mathbb{Z}}
\newcommand{\IN}{\mathbb{N}}
\newcommand{\IR}{\mathbb{R}}
\newcommand{\IC}{\mathbb{C}}
\newcommand{\IE}{\mathbb{E}}
\newcommand{\II}{\mathbb{I}}
\newcommand{\OO}{\mathcal{O}}
\renewcommand{\epsilon}{\varepsilon}
\renewcommand{\emptyset}{\varnothing}
\newcommand{\0}{\mathfrak{0}}
\newcommand{\Cyl}{\mathit{Cyl}}
\newcommand{\IP}{\mathbb{P}}
\newcommand{\F}[1]{\mathfrak{#1}}
\newcommand{\C}[1]{\mathcal{#1}}
\newcommand{\dashed}[1]{#1^{\prime}}

\newcommand{\expecP}[2]{\IE^\PP_{#1}\left(#2\right)}
\newcommand{\expecM}[2]{\IE^{\C{M}}_{#1}\left(#2\right)}
\newcommand{\lfp}{\mathrm{lfp}}
\newcommand{\BE}{\mathbf{E}}
\newcommand{\BP}{\mathbf{P}}

\newcommand{\PP}{\mathcal{P}}

\newcommand{\run}[1]{\left\langle #1 \right\rangle}
\newcommand{\uni}{\mathit{rdw}}
\newcommand{\column}[2]{(#1, #2)}

\newcommand{\startvec}[1]{\vec{#1}_0}
\newcommand{\start}[1]{#1_0}
\newcommand{\unisup}{{\mathit{rdw}}}

\newcounter{def-univariate-program}
\newcounter{Univariate Transformation}

\usepackage{microtype}

\newcommand{\cameraready}[1]{#1}
\newcommand{\techrep}[1]{#1}

\renewcommand{\cameraready}[1]{}

\begin{document}	
\title{Computing
Expected Runtimes for Constant Probability Programs\thanks{Supported by the
  DFG Research 
  Training Group 2236 UnRAVeL and
  the London Mathematical Society (Grant 41662, Research in Pairs).}} 

\titlerunning{Computing
Expected Runtimes}

\author{J\"urgen Giesl\inst{1}
  \and Peter Giesl\inst{2}
  \and Marcel Hark\inst{1}
}

\authorrunning{J.\ Giesl, P.\ Giesl, and M.\ Hark}

\institute{LuFG Informatik 2, RWTH Aachen University,  Germany\\ \email{\{giesl,marcel.hark\}@cs.rwth-aachen.de} \and Department of Mathematics, University of Sussex, UK\\ \email{p.a.giesl@sussex.ac.uk} }
	\date{\today}

        \maketitle
        
        \begin{abstract}
 We introduce the class of
 \emph{constant probability (CP) programs} and show that 
 classical results from probability theory
 directly yield a\linebreak[4] simple decision procedure for 
(positive) almost sure termination of  programs in this
  class.
Moreover, asymptotically tight bounds on their expected runtime can always be
computed easily. 
Based on this,
 we present an algorithm to infer the \emph{exact}
 expected runtime of any CP program.
\keywords{Probabilistic Programs \and Expected Runtimes \and (Positive) Almost Sure Termination  \and Complexity \and Decidability}
\end{abstract}

\section{Introduction}\label{Introduction}

Probabilistic programs are used to describe randomized algorithms and probabil-\linebreak
distributions, with applications in many areas.  
For example, consider the well-\linebreak known program which
models the race between a tortoise and a hare (see, e.g.,\linebreak
\cite{Chakarov13,DBLP:conf/esop/KaminskiKMO16, DBLP:conf/pldi/NgoC018}). As long as  the
tortoise
 (variable $t$) is not behind the hare (variable $h$), 
 \begin{wrapfigure}[5]{r}{5.7cm}
   \centering
   \vspace*{-.75cm}
  \hspace*{-.3cm}\fbox{\begin{minipage}{5cm}\begin{tabbing}
  \= \hspace*{.2cm}\=\kill
  \>\texttt{while} $(h \leq t) \; \{$\\[-.05cm]
  \>\>$t=t+1$; \hspace*{1cm}\\[-.05cm]
  \>\>$\{h = h + \mathit{Unif}(0,10)\}\oplus_{\tfrac{1}{2}} \{h=h\};$\\[-0.2cm]
  \>$\}$
  \end{tabbing}\end{minipage}}
 \end{wrapfigure}
 it does one step in each iteration. With
probability $\tfrac{1}{2}$, the hare stays at its position and with probability $\tfrac{1}{2}$ it does
a random number of steps uniformly chosen between $0$ and $10$. The race ends
when\linebreak the hare is in front of the tortoise.
Here, the hare wins with probability 1  and the\linebreak technique of
\cite{DBLP:conf/pldi/NgoC018} infers
the upper bound
$\tfrac{2}{3}\cdot \max(t-h+9,0)$
on the
expected num-\linebreak ber of loop iterations. Thus, the program is
positively almost surely terminating.

\cref{ConnectionSection} recapitulates
 preliminaries on probabilistic programs and on the
 connection between their expected runtime and 
 their
corresponding recurrence equation.
 Then we show 
 in \cref{BoundsSection,sec:random_walk_programs}
that classical results on random walk theory directly yield a very simple decision
procedure for (positive) almost sure termination  of \emph{CP programs} like the
tortoise and hare
example.
In this way, we also obtain
 asymptotically
tight bounds on the expected runtime of any CP program. \pagebreak
Based on these bounds,  in \cref{sec:exact} we develop 
the first algorithm to compute  closed forms for 
the \emph{exact}
expected runtime of such programs. 
In \cref{Conclusion}, we present its implementation
in our tool \textsf{KoAT} \cite{KoAT}
and discuss
related and future work.
We refer to \techrep{the appendix}\cameraready{\cite{TECHREP}} 
for a collection of examples to illustrate the application of our algorithm
and for all proofs.

\section{Expected Runtimes of Probabilistic Programs}\label{ConnectionSection}
\begin{example}[Tortoise and Hare]\label{exmpl:tortoise_and_hare}
 \sl \hspace*{.08cm} The  \hspace*{.08cm} pro-
{\makeatletter
\let\par\@@par
\par\parshape0
\everypar{}
  \begin{wrapfigure}[9]{r}{4.9cm}
    \centering
    \vspace*{-1.3cm}
    \hspace*{-.1cm}\fbox{\begin{minipage}{5cm}\begin{tabbing}
      \= \hspace*{.3cm}\=\kill
      \>\mbox{\rm \texttt{while}} $\left(\column{1}{-1}\bullet \column{t}{h} > -1\right) \; \{$\\
      \>\>$\column{t}{h} = \column{t}{h} + \column{1}{0}$ \hspace*{.5cm}\=$[\tfrac{6}{11}]$;\\[0.1cm]
      \>\>$\column{t}{h} = \column{t}{h} + \column{1}{1}$\>$[\tfrac{1}{22}]$;\\[0.1cm]
      \>\>$\column{t}{h} = \column{t}{h} + \column{1}{2}$\>$[\tfrac{1}{22}]$;\\[0.1cm]
      \>\>$\column{t}{h} = \column{t}{h} + \column{1}{3}$\>$[\tfrac{1}{22}]$;\\[-0.1cm]
      \>\>\hspace*{.82cm} $\vdots$\\[-0.1cm]
      \>\>$\column{t}{h} = \column{t}{h} + \column{1}{10}$\>$[\tfrac{1}{22}]$;\\
      \>$\}$
    \end{tabbing}\end{minipage}}
  \end{wrapfigure}
  \noindent
  gram $\PP_{race}$ on the right formulates the  race of the tortoise and the hare 
  as a  
  \emph{CP program}.
In the loop guard, we use the  scalar product $\column{1}{-1}\bullet \column{t}{h}$ which
stands for $t - h$. 
  Exactly one of
  the instructions with numbers in brackets $[\ldots]$ is executed in each loop iteration and the
  number indicates the probability that the corresponding instruction is chosen.
\par}
\end{example}

\noindent
We now define the kind of 
probabilistic programs considered in this paper.

\begin{definition}[Probabilistic Program]\label{def-program}
  \hspace*{.02cm} A \hspace*{.02cm} \emph{pro-}
  {\makeatletter
\let\par\@@par
\par\parshape0
\everypar{}\begin{wrapfigure}[6]{r}{3.65cm}
   \centering
   \vspace*{-1.25cm}
 \hspace*{-.15cm}\fbox{\begin{minipage}{5cm}\begin{tabbing}
  \= \hspace*{.2cm}\=\kill
    \>\emph{\texttt{while}} $(\vec{a} \bullet \vec{x} > b) \; \{$\\
    \>\>$\vec{x}=\vec{x} + \vec{c}_1$ \hspace*{.2cm}\=$[p_{\vec{c}_1}(\vec{x})];$\\[-0.2cm]
    \>\>\hspace*{.4cm}$\vdots$\>\\[-0.1cm]
    \>\>$\vec{x}=\vec{x} + \vec{c}_{n}$\>$[p_{\vec{c}_{n}}(\vec{x})];$\\[0.1cm]
    \>\>$\vec{x} = \vec{d}$\>$[\dashed{p}(\vec{x})];$\\
    \>$\}$
\end{tabbing}
   \end{minipage}}
\end{wrapfigure}
\noindent{}\emph{gram} has the form on the right,
   where $\vec{x}=(x_1,\ldots,x_r)$ for some $r \geq 1$
is a tuple of pairwise different program va\-riables,
$\vec{a},\vec{c}_1,\ldots,\vec{c}_{n}\in \IZ^{r}$ are tuples
of integers, the 
$\vec{c}_j$ are pairwise distinct, $b\!\in\!\IZ$, $\bullet$ is the scalar
product (i.e., $(a_1,\ldots,a_{r}) \bullet (x_1,\ldots,x_{r}) =
a_1 \cdot x_1 + \ldots + a_{r} \cdot x_{r}$), and $\vec{d} \in\IZ^{r}$ 
with $\vec{a} \bullet \vec{d} \leq b$. We require $p_{\vec{c}_1}(\vec{x}), \ldots,
p_{\vec{c}_{n}}(\vec{x}), \dashed{p}(\vec{x}) \in \IR_{\geq 0}= \{r \in \IR \mid r \geq
0\}$ and $\sum\nolimits_{1 \leq j \leq n} p_{\vec{c}_j}(\vec{x}) + \dashed{p}(\vec{x}) =
1$ for all $\vec{x} \in \IZ^{r}$. It is a program with \emph{direct termination} if there is an
$\vec{x} \in \IZ^{r}$ with
$\vec{a} \bullet \vec{x} > b$ and $p'(\vec{x}) > 0$.
If all probabilities are constant, i.e., if there are
$p_{\vec{c}_1}, \ldots,p_{\vec{c}_n}, p' \in \IR_{\geq 0}$ such that $p_{\vec{c}_j}(\vec{x}) = p_{\vec{c}_j}$ and $p'(\vec{x}) =p'$ for all
$1 \leq j \leq n$ and all $\vec{x} \in \IZ^{r}$, we call it a \emph{constant probability
  (CP) program}.
\par}

\end{definition}
\noindent
Such a 
program means that the integer variables $\vec{x}$ are changed to $\vec{x}+\vec{c}_j$
with\linebreak[4]
probability $p_{\vec{c}_j}(\vec{x})$. For inputs $\vec{x}$ with $\vec{a} \bullet \vec{x} \leq b$ the program terminates
immediately. Note that the program in \cref{exmpl:tortoise_and_hare} has \emph{no} direct
termination (i.e., $p'(\vec{x}) = 0$ for all $\vec{x} \in \IZ^{r}$).
Since the values of the program variables only depend on their values
in the previous loop iteration, our programs correspond to \emph{Markov Chains}
\cite{MDPsPuterman}
and they are related to \emph{random walks}
\cite{probabilityGrimmett,randomWalkSpitzer,feller50}, cf.\
\cameraready{\cite{TECHREP}}\techrep{the appendix} for details.

Clearly, in
general
termination is undecidable and closed forms for the runtimes of programs  are not
computable.
Thus, 
decidability results can only be obtained for suitably  restricted forms of programs.
Our class nevertheless includes many
examples that are often regarded in the literature on 
probabilistic programs. 
So while
other approaches are concerned with \emph{incomplete} 
techniques to analyze
termination and complexity,
we investigate classes of probabilistic programs where one can
\emph{decide} 
the termination behavior, \emph{always} find complexity bounds, and even \pagebreak
compute the expected runtime \emph{exactly}.
Our decision
procedure could be integrated into general tools for termination and complexity analysis
of probabilistic programs:
As soon as one has to investigate a sub-program that falls
into our class, one can use the decision procedure to compute its exact runtime.
Our contributions  provide a starting point for such results and
the considered class of programs can be extended further
in future work.

In probability theory (see, e.g., \cite{probabilityAsh}), given a set $\Omega$ of possible events, the goal is to measure the
probability that events are in certain subsets of $\Omega$. To this end,
one regards a set $\F{F}$ of subsets of $\Omega$, such that $\F{F}$ contains the full set $\Omega$ and is closed under
complement and countable unions. Such a set $\F{F}$ is called a \emph{$\sigma$-field}, and
a pair of $\Omega$ and a corresponding
$\sigma$-field  $\F{F}$ is called a \emph{measurable space}.

A \emph{probability space} $(\Omega, \F{F}, \IP)$ extends 
a measurable space $(\Omega, \F{F})$ by 
a \emph{probabil-}\linebreak \emph{ity measure} $\IP$ which
maps every set
from $\F{F}$ to a number between 0 and 1, with\linebreak  $\IP(\Omega) = 1$, $\IP(\varnothing) =
0$, and $\IP(\biguplus\nolimits_{j \geq 0} A_j) = \sum\nolimits_{j \geq 0}
\IP(A_j)$ for any pairwise disjoint sets $A_0, A_1,
\ldots \in \F{F}$.
So $\IP(A)$ is the probability that an event from $\Omega$ is in
the subset $A$. 
In our setting, we use the probability space $((\IZ^{r})^\omega, \F{F}^{\IZ^{r}},
\IP_{\startvec{x}}^\PP)$ arising from the standard cylinder-set construction of MDP
theory\techrep{, cf.\ \cref{MDP}}\cameraready{, cf.\ \cite{TECHREP}}. Here,
$(\IZ^{r})^\omega$\linebreak corresponds to
all
infinite sequences of program states and
$\IP_{\startvec{x}}^\PP$ is the
probability measure induced by the program $\PP$ when starting in the state
$\startvec{x}\in \IZ^{r}$. For example, if $A \subseteq (\IZ^{2})^\omega$ consists of all
infinite sequences starting with $\column{5}{1}$, $\column{6}{1}$,\linebreak $\column{7}{6}$,
then $\IP_{\column{5}{1}}^{\PP_{race}}(A) = \tfrac{6}{11} \, \cdot \, 
\tfrac{1}{22} = \tfrac{3}{121}$. So, if one starts with $\column{5}{1}$, then
$\tfrac{3}{121}$ is the
probability that the next two states are $\column{6}{1}$ and $\column{7}{6}$. Once a
state is reached that violates the loop guard, then the probability to remain
in this state is 1. Hence,  if $B$ contains all infinite sequences starting with
$\column{7}{8}$, $\column{7}{8}$, then $\IP_{\column{7}{8}}^{\PP_{race}}(B) = 1$. In the following, for any set of numbers $M$ let
$\overline{M} = M \cup \{ \infty \}$.

\begin{definition}[Termination Time]\label{def:termination_time}
 For a program $\PP$ as in \cref{def-program}, its \emph{termination time} is the random
 variable 
 $T^{\PP}:(\IZ^{r})^\omega \to \overline{\IN}$ that maps
  every infinite sequence $\langle \vec{z}_0, \vec{z}_1, \ldots \rangle$
   to the first index $j$ where $\vec{z}_j$ violates $\PP$'s loop guard.
\end{definition}

\noindent
Thus, $T^{\PP_{race}}(\langle \column{5}{1},\column{6}{1},\column{7}{8},\column{7}{8},
\ldots\rangle) 
= 2$ 
and
$T^{\PP_{race}}(\langle \column{5}{1},\column{6}{1},
\column{5}{6},\linebreak \column{8}{6}, \column{9}{6},\ldots\rangle)
= \infty$ (i.e., this sequence always satisfies $\PP_{race}$'s loop guard as the $j$th entry is $\column{5+j}{6}$ for $j\geq 3$).
Now we can define the different notions of termination and the expected runtime of
a probabilistic program. As usual, for any random variable
$X$ on a probability space $(\Omega, \F{F}, \IP)$, $\IP(X = j)$ stands for
$\IP(X^{-1}(\{j\}))$. So $\IP_{\startvec{x}}^\PP(T^{\PP} =
j)$ is the probability that a sequence has termination time $j$. Similarly,
$\IP_{\startvec{x}}^\PP(T^{\PP} < \infty) = \sum\nolimits_{j \in \IN} \IP_{\startvec{x}}^\PP(T^{\PP} =
j)$. The \emph{expected value}
$\IE(X)$ of a random variable $X: \Omega \to \overline{\IN}$
for a probability space $(\Omega, \F{F}, \IP)$
is the weighted average under the probability measure
$\IP$, i.e., $\IE(X) = \sum\nolimits_{j \in \overline{\IN}} \;
j \cdot
\IP(X = j)$, where
$\infty \cdot 0 = 0$ and $\infty \cdot u = \infty$ for all $u \in \IN_{>0}$.

\begin{restatable}[Termination and Expected Runtime]{definition}{defastpast}\label{AST Def}
A program $\PP$ as in \cref{def-program} is \emph{al\-most surely
    terminating (AST)} if 
  $\IP^{\PP}_{\startvec{x}}(T^{\PP}\!\!<\infty)=1$ for any initial value $\startvec{x} \in \IZ^{r}$.
For any $\startvec{x} \in \IZ^{r}$, its \emph{expected runtime} $rt_{\startvec{x}}^\PP$
(i.e., the expected number of loop iterations) 
is defined as the expected
value of the random variable $T^{\PP}$ \pagebreak
under the probability measure $\IP^{\PP}_{\startvec{x}}$,
i.e.,
$rt_{\startvec{x}}^\PP = \expecP{\startvec{x}}{T^{\PP}} =\linebreak \sum\nolimits_{j \in \IN  } \;j \cdot
\IP_{\startvec{x}}^\PP(T^{\PP}\!\!=\!j)$ if $\IP_{\startvec{x}}^\PP(T^{\PP}\!\!<\!\infty)=1$, and  $rt_{\startvec{x}}^\PP\!= \expecP{\startvec{x}}{T^{\PP}}=\infty$
otherwise.

\noindent{}The program $\PP$ is \emph{positively almost surely terminating (PAST)} if
for any initial value $\startvec{x} \in \IZ^{r}$, the expected runtime of $\PP$ is finite, i.e.,
if  $rt_{\startvec{x}}^\PP = \expecP{\startvec{x}}{T^{\PP}} < \infty$.  
\end{restatable}

\begin{example}[Expected Runtime for $\PP_{race}$]
 {\sl By the observations
   in \cref{sec:random_walk_programs} we will infer
 that $\tfrac{2}{3}\cdot (t-h+1) \leq rt_{(t,h)}^{\PP_{race}} \leq \tfrac{2}{3}\cdot (t-h+1)
   + \tfrac{16}{3}$ holds whenever $t-h > -1$,
   cf.\ \cref{exmpl:bounds_tortoise_and_hare_reduced}. So the expected number of steps
   until termination is finite (and 
   linear in the input variables) and thus,  $\PP_{race}$ is PAST. The algorithm in \cref{sec:exact} will
  even be able to compute $rt_{(t,h)}^{\PP_{race}}$ exactly, cf.\ \cref{Exact
    Expected Runtime of Tortoise and Hare}.}
\end{example}

\noindent
If the initial values $\startvec{x}$ violate the loop guard, 
then the runtime is trivially 0.

\begin{restatable}[Expected Runtime for Violating Initial Values]{corollary}{coroviolatinginitialvalues}\label{Expected_Runtime_Violating}
  For any program $\PP$ as in \cref{def-program} and any
$\startvec{x} \in \IZ^{r}$ with
  $\vec{a} \bullet \startvec{x} \leq b$, we have
  $rt^\PP_{\startvec{x}} = 0$.
\end{restatable}

\noindent{}To obtain our results, we use
an alternative, well-known characterization of the expected
runtime, cf.\ e.g.,
\cite{Karp94,Bazzi2003,Tassarotti18,DBLP:conf/focs/Kozen79,DBLP:series/mcs/McIverM05,Esparza2005,Brazdil2013,DBLP:conf/esop/KaminskiKMO16,MDPsPuterman}. To
this end, we search for the \emph{smallest} (or ``\emph{least}'') solution 
of the
recurrence equation that describes the runtime of the program as 1 plus the sum
of the runtimes in the next loop iteration, multiplied with the corresponding
probabilities.
Here,
functions are compared pointwise, i.e.,
for $f,g: \IZ^{r} \to 
\overline{\IR_{\geq_0}}$ we have $f \leq g$ if $f(\vec{x}) \leq g(\vec{x})$ holds for all $\vec{x} \in \IZ^{r}$.  
So we search for
the smallest function $f:\IZ^{r} \to \overline{\IR_{\geq 0}}$  that satisfies
\begin{equation}
   \label{recur} \!\!f(\vec{x}) =  \sum\nolimits_{1 \leq j \leq n} p_{\vec{c}_j}(\vec{x})
   \cdot f(\vec{x}+\vec{c}_j)+p'(\vec{x}) \cdot f(\vec{d})+1 \;\;\; \text{for all $\vec{x}$
     with $\vec{a} \bullet \vec{x} > b$}. \;
\end{equation}
Equivalently,
we can search for the least fixpoint of
the  ``expected runtime trans-\linebreak former'' $\C{L}^\PP$ which
transforms the left-hand side of \eqref{recur} into its right-hand side.

\begin{definition}[$\C{L}^\PP$\protect{\textnormal{, cf.~\cite{MDPsPuterman}}}]\label{expected_runtime_transformer}
  For $\PP$ as in \cref{def-program}, we define the \emph{expected
    runtime transformer}
    $\C{L}^\PP\!: (\IZ^{r}\!\to \overline{\IR_{\geq 0}}) \to
(\IZ^{r}\!\to \overline{\IR_{\geq 0}})$, where for any $f\!: \IZ^{r}\!\to
  \overline{\IR_{\geq 0}}$: 
  \[\C{L}^\PP(f)(\vec{x})=\begin{cases}\sum\nolimits_{1 \leq j \leq n} p_{\vec{c}_j}(\vec{x}) \cdot f(\vec{x}+\vec{c}_j)+\dashed{p}(\vec{x})\cdot f(\vec{d}) +
  1,&\text{if } \vec{a} \bullet \vec{x} > b\\
  f(\vec{x}),&\text{if } \vec{a} \bullet \vec{x} \leq b 
  \end{cases}  
  \]
\end{definition}

\begin{example}[Expected Runtime Transformer  for $\PP_{race}$]
  {\sl For $\PP_{race}$ from \cref{exmpl:tortoise_and_hare}, $\C{L}^{\PP_{race}}$ maps any function  $f: \IZ^{2} \to
    \overline{\IR_{\geq 0}}$ to $\C{L}^{\PP_{race}}(f)$, where
$\C{L}^{\PP_{race}}(f)(t,h) =$}  
  \begin{equation}
    \label{rhs of ert example}
\begin{cases}
    \tfrac{6}{11} \cdot f(t+1, h) + \tfrac{1}{22} \cdot \sum\nolimits_{1 \leq
      j \leq 10} f(t+1, h+j) + 1, &\text{if $t-h > -1$}\\
    f(t,h), &\text{if $t-h \leq -1$}
    \end{cases}
  \end{equation}
 \end{example}

\noindent{}\cref{correctness_of_ert}
recapitulates that the least fixpoint of $\C{L}^\PP$
indeed yields an equivalent characterization of the expected runtime.
In the following, let
$\0:\IZ^{r} \to  \overline{\IR_{\geq 0}}$ be the function with $\0(\vec{x}) = 0$ for all
$\vec{x} \in \IZ^{r}$.

\begin{restatable}[Connection Between Expected Runtime and Least Fixpoint of
    $\C{L}^\PP$\protect{\textnormal{, cf.~\cite{MDPsPuterman}}}]{theorem}{theoremcorrectnessert}\label{correctness_of_ert}
For any $\PP$ as in \cref{def-program}, the expected runtime transformer
$\C{L}^\PP$ is continuous.  Thus, it
has a least fixpoint  $\lfp(\C{L}^P): \IZ^{r} \to \overline{\IR_{\geq_0}} $ with $\lfp(\C{L}^\PP) = \sup \{ \0, \C{L}^\PP(\0),  (\C{L}^\PP)^2(\0), \ldots
\}$.
Moreover, the least fixpoint of $\C{L}^\PP$ is the expected runtime of $\PP$, i.e., for any $\startvec{x} \in \IZ^{r}$, we have
 $\lfp(\C{L}^\PP)(\startvec{x})=rt^\PP_{\startvec{x}}$.
 \end{restatable}
 So the expected runtime $rt^{\PP_{race}}_{(t,h)}$ can also be characterized as the smallest
  function $f\!: \IZ^{2}\!\to
  \overline{\IR_{\geq 0}}$ satisfying $f(t,h)\!=\!\eqref{rhs of ert example}$, i.e.,
 as the least fixpoint of $\C{L}^{\PP_{race}}$.

\section{\hspace*{-0.2cm}Expected Runtime of Programs with Direct Termination}\label{BoundsSection}
We start with stating a decidability
result for the case where for all $\vec{x}$ with
 $\vec{a}\bullet\vec{x} > b$,
the probability
$p'(\vec{x})$ for
direct termination is at least $p'$ for some $p' > 0$. Intuitively, these programs have a termination time whose distribution is closely related to the geometric distribution with parameter $p'$ (which has expected value $\tfrac{1}{p'}$). By using the alternative characterization of $rt^\PP_{\startvec{x}}$
from
\cref{correctness_of_ert},
one obtains that
such programs
are \emph{always PAST} and their expected runtime
is indeed \emph{bounded by the constant} $\tfrac{1}{p'}$.
This result will be used in \cref{sec:exact} when computing the \emph{exact} expected
runtime of such programs.
The more involved case where $p'(\vec{x}) = 0$ is
considered in \cref{sec:random_walk_programs}. 

\begin{restatable}[PAST and Expected Runtime for Programs With Direct Termination]{theorem}{theoremPAST}\label{PAST}
  Let $\PP$ be a program as in \cref{def-program}
  where there is a $p' > 0$ such that
  $p'(\vec{x}) \geq p'$  for all $\vec{x} \in \IZ^{r}$ with $\vec{a}\bullet \vec{x} > b$.
Then $\PP$ is PAST and its expected runtime is at most $\tfrac{1}{p'}$, i.e.,
$rt^\PP_{\startvec{x}} \leq \tfrac{1}{p'}$ if $\vec{a}\bullet\startvec{x} > b$, and $rt^\PP_{\startvec{x}} = 0$ if $\vec{a}\bullet\startvec{x} \leq b$.
\end{restatable}

\begin{example}[\cref{exmpl:tortoise_and_hare}
    with Direct Termination]\label{exmpl:tortoise_and_hare_direct_catch_up}
{\makeatletter
\let\par\@@par
\par\parshape0
\everypar{}
  \begin{wrapfigure}[4]{r}{4.9cm}
    \centering
    \vspace*{-1.3cm}
    \hspace*{-.1cm}\fbox{\begin{minipage}{5cm}
    \begin{tabbing}
      \= \hspace*{.3cm}\=\kill
      \>\texttt{while} $\left(\column{1}{-1}\bullet \column{t}{h} > -1\right) \; \{$\\
      \>\>$\column{t}{h} = \column{t}{h} + \column{1}{0}$ \hspace*{.4cm}\=$[\tfrac{9}{10}];$\\[0.1cm]
      \>\>$\column{t}{h} = \column{7}{8}$\>$[\tfrac{1}{10}];$\\
      \>$\}$
    \end{tabbing}
        \end{minipage}}
  \end{wrapfigure}
\sl  \noindent{}Consider the variant $\PP_{direct}$ of $\PP_{race}$ on the right,
where in each iteration, the hare
  either does nothing with probability  $\tfrac{9}{10}$ or one directly reaches a 
  configuration where  the hare is ahead of the tortoise.
  By \cref{PAST} the program is PAST and its expected runtime is at most
  $\tfrac{1}{\frac{1}{10}}=10$,
  i.e., independent of the initial state it
  takes at most 10 loop iterations on average.
  In \cref{sec:exact}
it will turn
out that
 10 is indeed the
 exact expected runtime, cf.\ \cref{exmpl:tortoise_and_hare_direct_catch_up_cont}.\par}
\end{example}

\section{Expected Runtimes of Constant Probability Programs}\label{sec:random_walk_programs}
Now we present a very simple decision procedure for  termination
of CP programs (\cref{sec:decidability_of_termination}) and show how to infer their asymptotic expected runtimes 
(\cref{sec:computability_of_asymptotic_expected_runtimes}). This will be needed
for the  computation of  exact expected runtimes
in \cref{sec:exact}.

 %SUBSECTION ON REDUCTION TO RANDOM WALKS%
\subsection{Reduction to Random Walk Programs}\label{sec:Restriction_to_Random_Walk_Programs}

{\makeatletter
\let\par\@@par
\par\parshape0
\everypar{}\begin{wrapfigure}[4]{r}{3.1cm}
   \centering
   \vspace*{-1.8cm}
   \hspace*{-.15cm}\fbox{\begin{minipage}{3.1cm}
\begin{tabbing}
      \= \hspace*{.2cm}\=\kill
      \>\texttt{while} $(x > 0) \; \{$\\
      \>\>$x=x+m$ \hspace*{.2cm}\=$[p_m];$\\[-0.2cm]
      \>\>\hspace*{.37cm}$\vdots$\>\\[-0.1cm]
      \>\>$x=x+1$\>$[p_{1}];$\\[0.1cm]
      \>\>$x = x$\>$[p_0];$\\[0.1cm]
      \>\>$x = x - 1$\>$[p_{-1}];$\\[-0.2cm]
      \>\>\hspace*{.37cm}$\vdots$\>\\[-0.1cm]
      \>\>$x = x-k$\>$[p_{-k}];$\\[0.1cm]
      \>\>$x = d$\>$[p'];$\\
      \>$\}$
  \end{tabbing}
        \end{minipage}}
\end{wrapfigure}
\noindent{}As a first step, we show
that we can restrict ourselves to \emph{random walk programs}, i.e.,
 programs with 
 a single program variable $x$ and the loop condition
 $x > 0$.

 \bigskip

\addtocounter{definition}{1}
\setcounter{def-univariate-program}{\value{definition}}

 \noindent
\begin{minipage}{8.75cm}
\textbf{Definition \hypertarget{def-univariate-program}{\arabic{definition}}
 (Random Walk Program).}
\emph{A CP program $\PP$ is called a \emph{random walk program} if there exist $m,k \in \IN$ and  $d \in \IZ$ with $d \leq 0$ such that $\PP\!$
   has the form on the right.
  Here, we require that $m> 0$ implies $p_m > 0$ and that $k > 0$ implies $p_{-k} >
  0$.}
\end{minipage}

\addtocounter{definition}{1}
  \setcounter{Univariate Transformation}{\value{definition}}

  \vspace{2ex}\noindent{}Def.\ \hyperlink{Univariate Transformation}{\arabic{Univariate Transformation}}
    shows how to transform any
     CP program as
in \cref{def-program}
into a random walk program.
The idea is to replace the tuple  $\vec{x}$ by a single
variable $x$ that stands for $\vec{a} \bullet \vec{x} -b$. Thus, the loop condition 
$\vec{a} \bullet \vec{x} > b$ now becomes $x > 0$. Moreover, a change from $\vec{x}$ to
$\vec{x} + \vec{c}_j$ now  becomes a change from $x$ to $x + \vec{a} \bullet \vec{c}_j$.

\medskip
\smallskip

\protect{\begin{wrapfigure}[4]{l}{3.1cm}
   \centering
   \vspace*{-2.1cm}
  \fbox{\begin{minipage}{3.1cm}
 \begin{tabbing}
      \= \hspace*{.2cm}\=\kill
      \>\texttt{while}  $(\vec{a} \bullet \vec{x} > b) \; \{$\\
       \>\>$\vec{x}=\vec{x} + \vec{c}_1$ \hspace*{.2cm}\=$[p_{\vec{c}_1}]$;\\[-0.2cm]
    \>\>\hspace*{.4cm}$\vdots$\>\\[-0.1cm]
    \>\>$\vec{x}=\vec{x} + \vec{c}_{n}$\>$[p_{\vec{c}_{n}}]$;\\[0.1cm]
    \>\>$\vec{x} = \vec{d}$\>$[\dashed{p}]$;\\
    \>$\}$
  \end{tabbing}
     \end{minipage}}
\end{wrapfigure}
\noindent
\begin{minipage}{8.85cm}
\textbf{Definition \hypertarget{Univariate Transformation}{\arabic{definition}}
  (Transforming CP Programs to Random Walk
  Programs).} \emph{Let $\PP$ be the CP program on the left with
  $\vec{x} =(x_1,\ldots,x_r)$ and $\vec{a} \bullet \vec{d} \leq b$. Let $\uni_{\PP}$ denote the
  affine map $\uni_{\PP}\!: \IZ^{r}\!\!\to  \IZ$ with $\uni_\PP(\vec{z})  =
\vec{a} \bullet \vec{z} - b$ for}

\begin{wrapfigure}[4]{r}{3.45cm}
   \centering
   \vspace*{-.85cm}
   \hspace*{-.23cm}
   \fbox{\begin{minipage}{3.1cm}\begin{tabbing}
       \= \hspace*{.15cm}\=\kill
      \>\texttt{while} $(x>0) \; \{$\\
      \>\>$x=x+m_\PP$ \hspace*{.1cm}\=$[p^\unisup_{m_\PP}]$;\\[-0.2cm]
      \>\>\hspace*{.4cm}$\vdots$\>\\[-0.1cm]
      \>\>$x=x - k_\PP$\>$[p^\unisup_{-k_\PP}]$;\\[0.1cm]
      \>\>$x=\uni_\PP(\vec{d})$\>$[p']$;\\
      \>$\}$
  \end{tabbing}
  \end{minipage}}
\end{wrapfigure}
\emph{all $\vec{z} \in \IZ^{r}$.
Thus, $\uni_\PP(\vec{d})\leq 0$. Let  $k_\PP,m_\PP\in\IN$ be  minimal
 such that \hspace*{.03cm} $-k_\PP \leq \vec{a} \bullet \vec{c}_j  \leq m_\PP$ \hspace*{.03cm} holds}
\end{minipage}

\vspace*{.1cm}
\noindent
\begin{minipage}{8.4cm}
\emph{for all $1 \leq j \leq n$.
For all $-k_\PP \leq j \leq m_\PP$, we define
 $p^\unisup_j = \sum\nolimits_{1 \leq u \leq n, \; \vec{a} \bullet \vec{c}_u  = j}
 p_{\vec{c}_u}$. This results in the random walk program
 $\PP^\unisup$ on the right.}
 \end{minipage}}

\bigskip

\begin{example}[Transforming $\PP_{race}$]\label{exmpl:tortoise_and_hare_reduced}
  \sl  For the program $\PP_{race}$  
{\makeatletter
\let\par\@@par
\par\parshape0
\everypar{}
\begin{wrapfigure}[8]{r}{2.8cm}
   \centering
    \vspace*{-1.2cm}
   \fbox{\begin{minipage}{3.1cm}
 \begin{tabbing}
    \= \hspace*{.2cm}\=\kill
    \>\mbox{\rm \texttt{while}} $(x > 0) \; \{$\\
    \>\>$x=x+1$ \hspace*{.2cm}\=$[\tfrac{6}{11}]$;\\[0.1cm]
    \>\>$x=x$\>$[\tfrac{1}{22}]$;\\[0.1cm]
    \>\>$x = x-1$\>$[\tfrac{1}{22}]$;\\[0.1cm]
    \>\>$x=x-2$\>$[\tfrac{1}{22}]$;\\[-0.2cm]
    \>\>\hspace*{.27cm} $\vdots$\\[-0.1cm]
    \>\>$x = x-9$\>$[\tfrac{1}{22}]$;\\
    \>$\}$
  \end{tabbing}
 \end{minipage}}
\end{wrapfigure}
\noindent
of  \cref{exmpl:tortoise_and_hare}, the mapping $\uni_{\PP_{race}}:\IZ^{2} \to \IZ$ is 
  $\uni_{\PP_{race}}\column{t}{h} =  \column{1}{-1}\bullet \column{t}{h} +1 =
  t-h+1$. Hence we obtain the random walk program $\PP^\unisup_{race}$ on the right, where
  $x=\uni_{\PP_{race}}\column{t}{h}$ represents the distance between the tortoise and
  the hare.
  
\rm
 \vspace{2ex}\noindent{}Approaches based on 
supermartingales
(e.g., \cite{BournezGarnier05,Chakarov13,ChatterjeeTOPLAS18,DBLP:conf/popl/FioritiH15,ChatterjeePOPL2017,DBLP:journals/pacmpl/AgrawalC018})
use mappings similar to $\uni_\PP$ in order to infer a 
 real-valued term  which  over-approximates
the expected runtime.
However, in the following (non-trivial) theorem we show
that our transformation is not only an over- or 
under-approximation, but
the termination behavior and the expected
runtime of $\PP$ and $\PP^\unisup$
are  \emph{identical}.
\par}
\end{example}
 \vspace*{-.2cm}

\begin{restatable}[Transformation Preserves Termination \& Expected Runtime]{theorem}{thmtransformationpreservesbehavior}\label{transformation_preserves_behavior}
  Let $\PP$ be a CP program as in \cref{def-program}. Then the termination
  times $T^{\PP}\!$ and $T^{\PP^\unisup}\!\!\!\!$
  are identically distributed w.r.t.~$\uni_\PP$,
  i.e.,  for all
$\startvec{x} \in \IZ^{r}$ with $x_0 =\uni_{\PP}(\startvec{x})$ and all
  $j\!\in\!\overline{\IN}$ we have
  $ \IP_{\startvec{x}}^{\PP}(T^{\PP}\!\!\!=\!j)=\IP^{\PP^\unisup}_{x_0}(T^{\PP^\unisup}\!\!\!=\!j)$. So in particular,  $\IP_{\startvec{x}}^{\PP}(T^{\PP}\!\! <\!\! \infty )\!=  \IP^{\PP^\unisup}_{x_0}(T^{\PP^\unisup}\!\!\!<\!\!\infty )$ 
  and $rt_{\startvec{x}}^\PP\!=\!\IE^\PP_{\startvec{x}}(T^{\PP})
  \!=\! \IE^{\PP^\unisup}_{x_0}(T^{\PP^\unisup})
\!=\! rt_{x_0}^{\PP^\unisup}$. Thus,
the expected runtimes of
  $\PP$ on the input $\startvec{x}$
and of $\PP^\unisup$
  on  $x_0$ 
coincide.  
\end{restatable}

\noindent
The following definition identifies pathological programs that can be disregarded.

\begin{definition}[Trivial Program]
\hspace*{-.08cm}  Let \hspace*{-.08cm} $\PP$ \hspace*{-.08cm} be \hspace*{-.08cm}
a \hspace*{-.08cm}  CP  \hspace*{-.08cm} pro-
  {\makeatletter
\let\par\@@par
\par\parshape0
\everypar{}\begin{wrapfigure}[4]{r}{2.7cm}
   \centering
   \vspace*{-1.4cm}
   \hspace*{-.15cm}\fbox{\begin{minipage}{3.1cm} 
 \begin{tabbing}
      \= \hspace*{.4cm}\=\kill
      \>\emph{\texttt{while}} $(x>0) \; \{$\\
      \>\>$x=x$ \hspace*{.6cm}\=$[1];$\\
      \>$\}$
  \end{tabbing}
      \end{minipage}}
\end{wrapfigure}
\noindent{}gram as in \cref{def-program}. We call $\PP$ \emph{trivial}
  if $\vec{a} = \vec{0} = (0,0,\ldots,0)$ or if $\PP^\unisup$ is the program on the right.\par}
\end{definition}
Note that a random walk program $\PP$ is trivial iff it has the form $\texttt{while}
(x\!>\!0) \{ x = x \;\; [1];\}$, since $\PP\!=\!\PP^\unisup$ holds for random walk programs
$\PP$. 
From now on, we will exclude trivial programs $\PP$ as their termination behavior is
obvious: for inputs $\startvec{x}$ that satisfy
the loop condition $\vec{a} \bullet \startvec{x} > b$, the program never terminates
(i.e., $rt_{\startvec{x}}^\PP = \infty$) and for inputs  $\startvec{x}$
with $\vec{a} \bullet \startvec{x} \leq b$ we have
$rt_{\startvec{x}}^\PP = 0$.
Note that if $\vec{a} = \vec{0}$, then 
 the termination behavior just depends on $b$: if  $b < 0$, then  
$rt_{\startvec{x}}^\PP = \infty$ for all $\startvec{x}$ and if
$b \geq 0$, then 
 $rt_{\startvec{x}}^\PP = 0$ for all $\startvec{x}$.

 %SUBSECTION ON DECIDING TERMINATION WITH SPITZER%
\subsection{Deciding Termination}\label{sec:decidability_of_termination}

We now present a simple decision procedure for (P)AST of random walk programs $\PP$.
By the results of
\cref{sec:Restriction_to_Random_Walk_Programs}, this also yields a 
decision procedure for arbitrary CP programs.
If $p' > 0$, then \cref{PAST} already shows that $\PP$ is PAST
and its expected runtime is bounded by the constant $\tfrac{1}{p'}$. Thus, in the
rest of \cref{sec:random_walk_programs} we regard
\emph{random walk programs without direct termination}, i.e.,   $p' = 0$.

\Cref{Drift} introduces the \emph{drift} of a random walk program, i.e., the expected
value of the change of the program variable in one loop iteration, cf.\ \cite{BournezGarnier05}.

\begin{definition}[Drift]\label{Drift}
  Let $\PP$ be a random walk program $\PP$ as in
  Def.\ \hyperlink{def-univariate-program}{\arabic{def-univariate-program}}.
 Then
  its \emph{drift} is $\mu_\PP = \sum\nolimits_{-k \leq j \leq m} j \cdot p_j$.
\end{definition}

\noindent{}\cref{termination_decidable}
shows that
to decide (P)AST, 
one just has to compute the 
drift.

\begin{restatable}[Decision Procedure for (P)AST  of Random Walk Programs]{theorem}{thmdecidability}\label{termination_decidable}
  Let $\PP$ be a non-trivial random walk program without direct termination.
  \begin{enumerate}
    \item[$\bullet$] If $\mu_\PP > 0$,  then the program is \emph{not} AST.
    \item[$\bullet$] If $\mu_\PP = 0$,
      then the program is AST but  \emph{not} PAST.
    \item[$\bullet$] If $\mu_\PP < 0$, then the program is PAST.
  \end{enumerate}
\end{restatable}

\begin{example}[$\PP_{race}$ is PAST]\label{Prace PAST}
\sl
The drift of   $\PP^\unisup_{race}$ in \cref{exmpl:tortoise_and_hare_reduced}
 is $\mu_{\PP^\unisup_{race}}=1\cdot \tfrac{6}{11}+ \tfrac{1}{22} \cdot\sum\nolimits_{-9 \leq j \leq
    0} j=-\tfrac{3}{2}<0$.
 So on average the distance $x$ between
the  tortoise and the hare decreases in each loop iteration. Hence  by
  \cref{termination_decidable},  $\PP^\unisup_{race}$ is PAST and the following \cref{coro1} implies that
  $\PP_{race}$ is PAST as well.
  \end{example}

\begin{restatable}[Decision Procedure for (P)AST of CP programs]{corollary}{corohelp}\label{coro1}
  \noindent{}For a non-trivial CP program $\PP$, $\PP$ is (P)AST iff
$\PP^\unisup$ is (P)AST. 
  Hence, \cref{transformation_preserves_behavior,termination_decidable} yield a decision
  procedure for
AST and PAST of CP programs.
\end{restatable}

\noindent{}In \techrep{the appendix}\cameraready{\cite{TECHREP}}, we show that
\cref{termination_decidable} 
follows from classical results on random
 walks \cite{randomWalkSpitzer}.
Alternatively, \cref{termination_decidable} could also be proved
by combining several recent results on probabilistic programs: The approach of
\cite{DBLP:journals/pacmpl/McIverMKK18} could be used to show that  $\mu_\PP = 0$ implies
AST. Moreover, one could prove that $\mu_\PP < 0$ implies PAST by
showing that $x$ is a \emph{ranking supermartingale} of the program
\cite{BournezGarnier05,Chakarov13,ChatterjeeTOPLAS18,DBLP:conf/popl/FioritiH15}.
That the program is not PAST if $\mu_\PP
\geq 0$ and not AST if $\mu_\PP > 0$ could be proved by showing that $-x$
is a \emph{$\mu_\PP$-repulsing supermartingale} \cite{ChatterjeePOPL2017}.

While the proof of \cref{termination_decidable} is based on known results, the formulation of
\cref{termination_decidable} shows that there is an extremely \emph{simple} decision procedure
for (P)AST of CP programs, i.e., 
 checking the sign of the drift is much simpler than applying existing (general)
techniques for termination analysis of probabilistic
programs.

 %SUBSECTION ON INFERRING BOUNDS WITH WALDS LEMMA%
\subsection{Computing Asymptotic Expected Runtimes}\label{sec:computability_of_asymptotic_expected_runtimes}

It turns out that for random walk programs (and thus   by
\cref{transformation_preserves_behavior}, also for CP programs), 
one can not only decide termination, but one can also infer
tight bounds on the expected runtime. \cref{bounds_runtime_constant_probability_programs}
shows that the computation of the bounds is again very 
\emph{simple}.

\begin{restatable}[Bounds on the Expected Runtime of CP
    Programs]{theorem}{thmboundsgeneralized}\label{bounds_runtime_constant_probability_programs}
  
  \noindent{}Let $\PP$ be a non-trivial CP program as in
\cref{def-program}
  without direct
 termination which is PAST (i.e., 
 $\mu_{\PP^\unisup}<0$).
 Moreover, let $k_\PP$ be obtained according to the transformation from Def.\ \hyperlink{Univariate Transformation}{\arabic{Univariate Transformation}}. 
If $\uni_\PP(\startvec{x}) \leq 0$, then
$rt_{\startvec{x}}^\PP = 0$. If  $\uni_\PP(\startvec{x}) >
 0$, then $\PP$'s expected runtime is asymptotically linear
and
 we have  
  \[-\tfrac{1}{\mu_{\PP^\unisup}}\cdot \uni_{\PP}(\startvec{x}) \quad \leq \quad rt_{\startvec{x}}^\PP
  \quad \leq \quad
  -\tfrac{1}{\mu_{\PP^\unisup}}\cdot \uni_{\PP}(\startvec{x}) +
  \tfrac{1-k_\PP}{\mu_{\PP^\unisup}}. \]
\end{restatable}

\begin{example}[Bounds on the Runtime of $\PP_{race}$]\label{exmpl:bounds_tortoise_and_hare_reduced} 
 {\sl In \cref{Prace PAST} we saw that the  program
  $\PP^\unisup_{race}$ from \cref{exmpl:tortoise_and_hare_reduced}
  is PAST as it has the drift $\mu_{\PP^\unisup_{race}}=-\tfrac{3}{2}<0$. Note that here $k=9$. Hence by
  \cref{bounds_runtime_constant_probability_programs} we get that whenever
  $\uni_{\PP_{race}}\column{t}{h} = t-h+1$ is positive, the 
  expected runtime $rt^{\PP_{race}}_{\column{t}{h}}$ is between $-\tfrac{1}{\mu_{\PP^\unisup_{race}}}\cdot \uni_{\PP_{race}}\column{t}{h}=\tfrac{2}{3}\cdot (t-h+1)$ and $-\tfrac{1}{\mu_{\PP^\unisup_{race}}}\cdot \uni_{\PP_{race}}\column{t}{h} + \tfrac{1-k}{\mu_{\PP^\unisup_{race}}}=\tfrac{2}{3}\cdot (t-h+1)
  + \tfrac{16}{3}$.\linebreak
The same upper bound $\tfrac{2}{3}\cdot (t-h+1)
  + \tfrac{16}{3}$ 
was inferred in \cite{DBLP:conf/pldi/NgoC018}
by an incomplete technique
based on several inference rules and linear programming solvers. In contrast, \cref{bounds_runtime_constant_probability_programs} allows us to
read off such bounds directly from the program.}
\end{example}

\noindent{}Our proof of \cref{bounds_runtime_constant_probability_programs}\techrep{ in
  the appendix}\cameraready{ in \cite{TECHREP}}
again uses the connection to random
walks and shows that
the classical Lemma of Wald
\cite[Lemma 10.2(9)]{probabilityGrimmett}
 directly yields both the upper and the lower bound for the expected
 runtime.
Alternatively, the upper bound in \cref{bounds_runtime_constant_probability_programs} could also be proved by considering that $\uni_{\PP}(\startvec{x})+(1-k_{\PP})$ is a
\emph{ranking supermartingale} \cite{BournezGarnier05,Chakarov13,ChatterjeeTOPLAS18,DBLP:conf/popl/FioritiH15,DBLP:journals/pacmpl/AgrawalC018} whose expected decrease in each loop iteration is $\mu_\PP$.
The lower bound could also be inferred  by considering the \emph{difference-bounded submartingale} $-\uni_{\PP}(\startvec{x})$ \cite{DBLP:conf/vmcai/FuC19,DBLP:conf/icalp/BrazdilKNW12}.

\section{Computing Exact Expected Runtimes}\label{sec:exact}

\noindent{}While \cref{PAST,bounds_runtime_constant_probability_programs}
state how to deduce the \emph{asymptotic} expected runtime, we
now show that based on these results one can compute the runtime of CP programs
\emph{exactly}.
In general, whenever it is possible, then inferring the exact
runtimes of programs is preferable to asymptotic 
runtimes which ignore the 
``coefficients'' of the runtime.

Again, we first consider \emph{random walk} programs and 
generalize our technique to CP programs
using \cref{transformation_preserves_behavior} afterwards.
Throughout \cref{sec:exact},
for any random walk program $\PP$
as in  Def.\ \hyperlink{def-univariate-program}{\arabic{def-univariate-program}},
we require that $\PP$ is PAST,
i.e., that $p'>0$ (cf.\ \cref{PAST}) or that the drift $\mu_\PP$ is
negative if $p'=0$ (cf.\ \cref{termination_decidable}). Note that whenever $k=0$ and $\PP$
is PAST, then $p'>0$.\footnote{If $p' = 0$ and $k = 0$ then $\mu_\PP \geq 0$.}

To compute $\PP$'s expected runtime
exactly, we use its characterization as the least fixpoint of \pagebreak
the expected runtime transformer $\mathcal{L}^\PP$ (cf.~\cref{correctness_of_ert}), i.e., $rt^\PP_{x}$ is the smallest 
function $f:\IZ \to \overline{\IR_{\geq 0}}$  satisfying the 
constraint
\begin{equation}
   \label{constraint1OLD} f(x) \;= \; \sum\nolimits_{-k \leq j \leq m} p_j \cdot f(x+j)+p' \cdot f(d)+1 \quad \text{for all $x > 0$},
\end{equation}
cf.\ \cref{recur}.
Since  $\PP$ is PAST, $f$ never returns $\infty$, i.e.,
$f: \IZ \to \IR_{\geq 0}$.
Note that the smallest function $f: \IZ \to \IR_{\geq 0}$ that
satisfies \eqref{constraint1OLD} also satisfies
\begin{equation}
  \label{constraint2}  f(x) \;= \; 0 \quad \text{for all $x \leq 0$}.
\end{equation}

\noindent
Therefore, as $d \leq 0$, the constraint \eqref{constraint1OLD} can be  simplified to
\begin{equation}
   \label{constraint1} f(x) \;= \; \sum\nolimits_{-k \leq j \leq m} p_j \cdot f(x+j)+1 \quad \text{for all $x > 0$}.
\end{equation}

\noindent{}In \cref{sec:all solutions} we recapitulate how to compute all solutions of
such inhomogeneous
recurrence equations (cf., e.g., \cite[Ch.~2]{differenceEquationsSaber}).
However, to compute $rt^\PP_{x}$, the challenge is 
to find the \emph{smallest} solution $f: \IZ \to \IR_{\geq 0}$
of the  equation \eqref{constraint1}. 
Therefore, in \cref{sec:smallest solution} we will exploit the knowledge gained in
\cref{PAST,bounds_runtime_constant_probability_programs} to show that there is only a
\emph{single} function $f$ that satisfies both \eqref{constraint2} and
\eqref{constraint1} \emph{and} is bounded by a constant (if $p' > 0$, cf.\ \cref{PAST})
resp.\ by a linear function (if $p' = 0$,
cf.\ \cref{bounds_runtime_constant_probability_programs}).
This observation then allows us to compute $rt^\PP_x$
exactly. 
So the crucial prerequisites for this result are \cref{correctness_of_ert} (which
characterizes the expected runtime as the smallest solution of the equation
\eqref{constraint1}), \cref{termination_decidable} (which allows the restriction
to negative drift if $p' = 0$), and in particular
\cref{PAST,bounds_runtime_constant_probability_programs} (since \cref{sec:smallest
  solution} will show that the results of
\cref{PAST,bounds_runtime_constant_probability_programs} on the asymptotic runtime can be
translated into suitable conditions on the solutions of \eqref{constraint1}).

\subsection{Finding All Solutions of the Recurrence Equation}\label{sec:all solutions}

\begin{example}[Modification of $\PP^\unisup_{race}$]\label{exmpl:truncated_race}
  \sl
\hspace*{-.09cm}  To \hspace*{-.09cm} illustrate \hspace*{-.09cm} our \hspace*{-.09cm} ap-
 {\makeatletter
\let\par\@@par
\par\parshape0
\everypar{}\begin{wrapfigure}[7]{r}{2.8cm}
   \centering
   \vspace*{-1.2cm}
   \hspace*{-.2cm}\fbox{\begin{minipage}{3.1cm}
 \begin{tabbing}
 \= \hspace*{.2cm}\=\kill
    \>\mbox{\rm \texttt{while}} $(x > 0) \; \{$\\
    \>\>$x=x+1$ \hspace*{.2cm}\=$[\tfrac{6}{11}]$;\\[0.1cm]
    \>\>$x=x$\>$[\tfrac{1}{11}]$;\\[0.1cm]
    \>\>$x = x-1$\>$[\tfrac{1}{22}]$;\\[0.1cm]
    \>\>$x=x-2$\>$[\tfrac{7}{22}]$;\\
    \>$\}$
  \end{tabbing}
        \end{minipage}}
\end{wrapfigure}
\noindent{}proach, we use a modified version of
   $\PP^\unisup_{race}$
   from \cref{exmpl:tortoise_and_hare_reduced} to ease readability.
  In \cref{Conclusion}, we will consider the original program  $\PP^\unisup_{race}$ resp.\
   $\PP_{race}$ from
\cref{exmpl:tortoise_and_hare_reduced}
resp.\ \cref{exmpl:tortoise_and_hare}  again and
show its exact expected runtime inferred by the implementation of our approach.
 In the  modified program $\PP^{mod}_{race}$ on the right, the distance
between the tortoise and the hare still increases with probability $\tfrac{6}{11}$, but the
probability of decreasing by more than two is distributed to the cases where it stays the same and
where it decreases by two. 
  We have $p'=0$ and
the drift is $\mu_{\PP^{mod}_{race}} =1 \cdot \tfrac{6}{11} + 0 \cdot \tfrac{1}{11} -1 \cdot
\tfrac{1}{22} -2 \cdot \tfrac{7}{22}=-\tfrac{3}{22}<0$. So by \cref{termination_decidable},
$\PP^{mod}_{race}$
 is PAST.
By \cref{correctness_of_ert}, $rt^{\PP^{mod}_{race}}_{x}$ is the smallest function $f : \IZ \to \IR_{\geq 0}$ satisfying
  \begin{equation}
    \label{recurrence_tortoise_hare} f(x) \;= \;
    \tfrac{6}{11}\cdot f(x+1) +
\tfrac{1}{11}\cdot f(x) + 
     \tfrac{1}{22}\cdot f(x-1) +
    \tfrac{7}{22}\cdot f(x-2)
    +1  \; \text{
     for all $x > 0$}.
  \end{equation}\par}
\end{example}

\noindent{}Instead of searching for the \emph{smallest} $f:\IZ \to \IR_{\geq 0}$ 
satisfying \eqref{constraint1}, 
we first cal\-culate the set of \emph{all} functions $f:\IZ \to \IC$  that satisfy
\eqref{constraint1}, i.e.,
we also consider functions returning negative 
or complex numbers. Clearly, \cref{constraint1} is equivalent to
\begin{equation}
  \label{receq}
  \begin{array}{r@{\;\;}c@{\;\;}l}
	0&=&p_m \cdot f(x+m)+\ldots+p_1 \cdot f(x+1)+(p_0-1) \cdot f(x) \;+\\
	&&p_{-1} \cdot f(x-1)+\ldots+p_{-k} \cdot f(x-k)+1 \hspace*{2.1cm} \text{for all $x>0$.}
  \end{array}
  \end{equation}
The set
of solutions on $\IZ \to \IC$
of this
linear, inhomogeneous recurrence equation
 is an affine space which can be written as an arbitrary particular 
solution   of the inhomogeneous equation plus any linear combination of $k+m$
linearly independent solutions of the corresponding homogeneous recurrence equation.

We start with computing a solution to the inhomogeneous equation
\cref{receq}. To this end, we use the bounds for $rt^\PP_{x}$ from \cref{PAST,bounds_runtime_constant_probability_programs} (where we
take the upper bound
$\tfrac{1}{p'}$ if $p' > 0$ and the lower bound $-\tfrac{1}{\mu_\PP} \cdot x$ if $p' =
0$). So we define
\[ C_{const}=\tfrac{1}{p'}, \; \text{if $p' > 0$ \qquad and \qquad}
C_{lin}=-\tfrac{1}{\mu_\PP}, \; \text{if $p' = 0$.}\]
One easily shows that if 
$p' > 0$,  then $f(x) = C_{const}$
is a solution
of the  inhomo\-geneous recurrence equation \cref{receq} and if $p' = 0$, then
$f(x)=C_{lin} \cdot x$ solves   \cref{receq}.

\begin{example}[\cref{exmpl:truncated_race} cont.]\label{exmpl:truncated_race_clin}
 {\sl In the program $\PP^{mod}_{race}$ of \cref{exmpl:truncated_race}, we have $p'
   = 0$ and
   $\mu_{\PP^{mod}_{race}} =-\tfrac{3}{22}$. Hence $C_{lin}= 
    \tfrac{22}{3}$ and
  $C_{lin}\cdot x$ is a solution of  \cref{recurrence_tortoise_hare}.}
\end{example}

\noindent{}After having determined one particular solution of the inhomogeneous recurrence
equation
\cref{receq}, now we compute the solutions of the
\emph{homogeneous} recurrence equation which
results from \cref{receq} by replacing the add-on ``+ 1'' with 0. To this end, we
consider the corresponding \emph{characteristic polynomial} $\chi_\PP$:\footnote{If $m=0$ then
  $\chi_\PP(\lambda)=(p_0-1) \cdot \lambda^k+p_{-1} \cdot \lambda^{k-1}+\ldots+p_{-k}$, and if
  $k=0$ then $\chi_\PP(\lambda)=p_m \cdot \lambda^{m}+\ldots+p_1 \cdot
  \lambda+(p_0-1)$.
  Note that
  $p_0 \neq 1$ since $\PP$ is PAST
  and in  Def.\ \hyperlink{def-univariate-program}{\arabic{def-univariate-program}}
  we required
  that $m> 0$ implies $p_m > 0$ and  $k > 0$ implies $p_{-k} > 0$.
 Hence, the characteristic polynomial has exactly the degree $k + m$, even if $m = 0$ or
 $k = 0$.}
\begin{eqnarray}
\label{chi} \chi_\PP(\lambda)=p_m \cdot \lambda^{k+m}+\ldots+p_1 \cdot \lambda^{k+1}+(p_0-1) \cdot \lambda^k
+ p_{-1} \cdot \lambda^{k-1}+\ldots+p_{-k} \;\;\;\;
\end{eqnarray}
Let $\lambda_1,\ldots,\lambda_c$
denote the pairwise different (possibly complex) roots of the cha\-rac\-te\-ris\-tic
polynomial $\chi_\PP$. 
For all $1 \leq j \leq c$, let
 $v_j\in \mathbb N\setminus \{0\}$ be the multiplicity of the root
$\lambda_j$.
Thus, we have $v_1 + \ldots + v_c = k + m$.

Then we obtain the following $k+m$ linearly independent solutions of the homogeneous
recurrence equation resulting from \eqref{receq}:
\[ \lambda_j^x \cdot x^u \quad \text{ for all $1 \leq j \leq c$ and all $0 \leq u \leq v_j-1$}\]
So $f\!:\!\IZ\!\to\!\IC$ is a solution of 
\cref{constraint1} (resp.\ \eqref{receq})  iff there exist coefficients $a_{j,u}\!\in\!\IC$
with
\begin{equation}
  \label{sol} f(x) \;= \; 
  C(x) \, + \, \sum\nolimits_{1 \leq j \leq c} \;\; \sum\nolimits_{0\leq u \leq v_j-1}a_{j,u} \cdot
  \lambda_j^x \cdot x^u \quad 
  \text{for all $x > -k$,}
\end{equation}
where $C(x) = C_{const} = \tfrac{1}{p'}$ if
 $p' > 0$ and $C(x) = C_{lin} \cdot x = -\tfrac{1}{\mu_\PP} \cdot x$ if $p' = 0$.
The reason for requiring \eqref{sol} for all $x > -k$ is that $-k+1$ is the smallest
argument where $f$'s value is taken into account in \cref{constraint1}. 

\begin{example}[\cref{exmpl:truncated_race_clin} cont.]\label{exmpl:truncated_race_chi}
 {\sl The characteristic polynomial for the program $\PP^{mod}_{race}$ of \cref{exmpl:truncated_race}
has the  degree $k+m=2+1=3$ and is given by
  \[\chi_{\PP^{mod}_{race}}(\lambda)\;=\;\tfrac{6}{11}\cdot \lambda^3 -\tfrac{10}{11}\cdot \lambda^2 +\tfrac{1}{22}\cdot \lambda +\tfrac{7}{22}.\]

\noindent{}Its roots are $\lambda_1=1$, $\lambda_2=-\tfrac{1}{2}$, and $\lambda_3=\tfrac{7}{6}$. So
here, all roots are real numbers and they all have the multiplicity $1$. Hence, three linearly
independent solutions of the homogeneous part of \cref{recurrence_tortoise_hare} are the
functions $1^x=1$,  $(-\tfrac{1}{2})^x$, and
$(\tfrac{7}{6})^x$. Therefore, a function $f:\IZ \to \IC$
satisfies \cref{recurrence_tortoise_hare} iff  there are
  $a_1,a_2,a_3 \in \IC$  such that 
\begin{equation}
\label{allSolutionsTortoise}   \begin{array}{r@{\;\;}c@{\;\;}l}
    f(x)&=&C_{lin} \cdot  x + a_1 \cdot 1^x+a_2 \cdot (-\tfrac{1}{2})^x\!+ a_3
      \cdot (\tfrac{7}{6})^x\\
     &=&\tfrac{22}{3}\cdot x + a_1+a_2 \cdot
   (-\tfrac{1}{2})^x\!+ a_3\cdot   (\tfrac{7}{6})^x  \hspace*{1.1cm}
    \text{for $x > -2$}.
\end{array}
\end{equation}}
 \end{example}

\subsection{Finding the Smallest Solution of the Recurrence Equation}\label{sec:smallest solution}

In \cref{sec:all solutions}, we recapitulated
the standard approach for solving inhomogeneous recurrence
equations
which shows
that any function $f:\IZ \to \IC$ that satisfies the constraint \eqref{constraint1}
is of the form \eqref{sol}.
Now we will present a novel
 technique to compute
$rt^\PP_x$,  i.e.,
 the \emph{smallest non-negative} solution $f:\IZ \to \IR_{\geq 0}$ of \cref{constraint1}.
 By
   \Cref{PAST,bounds_runtime_constant_probability_programs}, this function $f$ is bounded by a constant
   (if $p' > 0$) resp.\ linear (if $p' = 0$). So, when representing
   $f$ in the form \eqref{sol}, we must have 
 $a_{j,u} = 0$ whenever
   $|\lambda_j| > 1$.
The following lemma shows how many roots with absolute value less or equal
to 1 there are (i.e., these are the only roots that we have to consider). 
It is proved using Rouch\'e's Theorem
which allows us to infer the number of roots whose absolute value is below a certain bound.
Note that  $1$ is a root of the characteristic polynomial iff $p' = 0$, since
$\sum\nolimits_{-k \leq j \leq m} p_j  = 1 - p'$.

\begin{restatable}[Number of Roots With Absolute Value $\leq 1$]{lemma}{lemmanumberroots}\label{Number of Roots Lemma}
Let $\PP$ be a random walk program as in
  Def.\ \hyperlink{def-univariate-program}{\arabic{def-univariate-program}}  that is
  PAST. 
  Then the characteristic
  polynomial $\chi_\PP$ has $k$ 
roots $\lambda\in \IC$ (counted with multiplicity) with $|\lambda|\le 1$.
\end{restatable}

\begin{example}[\cref{exmpl:truncated_race_chi} cont.]\label{exmpl:truncated_race_zero}
{\sl In  $\PP^{mod}_{race}$ of \cref{exmpl:truncated_race} we have $k = 2$. So
by \cref{Number 
  of Roots Lemma}, $\chi_\PP$ has exactly two roots with absolute value $\leq
1$. Indeed, the roots of  $\chi_\PP$ are 
$\lambda_1=1$, $\lambda_2=-\tfrac{1}{2}$, and $\lambda_3=\tfrac{7}{6}$,  cf.\
\cref{exmpl:truncated_race_chi}. So  $|\lambda_3|>1$, but
$|\lambda_1|\leq 1$ and $|\lambda_2|\leq 1$.}
\end{example} 

\noindent
Based on \cref{Number of Roots Lemma},
the following lemma shows that when imposing
the restriction that 
 $a_{j,u} = 0$ whenever
$|\lambda_j| > 1$,
 then there is only a \emph{single} function of the form \eqref{sol} that also satisfies the constraint \eqref{constraint2}.
Hence, this must be the function that we are searching for, because 
the
desired smallest solution $f:\IZ \to \IR_{\geq 0}$ of \eqref{constraint1}  also satisfies
\eqref{constraint2}.

\begin{restatable}[Unique Solution of \eqref{constraint2} and \eqref{constraint1} when
    Disregarding Roots With Absolute Value $> 1$]{lemma}{lemmadisregarding}\label{reduction}
    Let $\PP$ be a
  random walk program as in
  Def.\ \hyperlink{def-univariate-program}{\arabic{def-univariate-program}}
  that is PAST.
Then there is exactly one function $f:\IZ \to \IC$ which satisfies both \eqref{constraint2} and
\eqref{constraint1} (thus, it  has the form  \eqref{sol}) and
has
$a_{j,u} = 0$ whenever
$|\lambda_j| > 1$.
\end{restatable}

\noindent{}The main theorem of
\cref{sec:exact} now shows how to compute the expected runtime exactly.
By  \Cref{PAST,bounds_runtime_constant_probability_programs}
 on the bounds for the expected
 runtime
and by \cref{reduction}, 
we no longer have to search for the \emph{smallest}  function
that satisfies \eqref{constraint2} and \eqref{constraint1}, but we just search for
\emph{any} solution of 
\eqref{constraint2} and \eqref{constraint1} which has
$a_{j,u} = 0$ whenever  $|\lambda_j| > 1$
(because there is
just a single such solution). 
So one only has to
determine the values of the remaining $k$ coefficients $a_{j,u}$ for $|\lambda_j| \leq 1$, which can be done by
exploiting that $f(x)$ has to satisfy both \eqref{constraint2} for all $x \leq 0$ and it has to be of the form
\eqref{sol} for all $x > -k$. In other words, 
the function in
\eqref{sol} must be 0 \pagebreak for $-k+1 \leq x \leq 0$.
        
\begin{restatable}[Exact Expected Runtime for Random Walk Programs]{theorem}{theoremsolution}\label{solution}
 \hspace*{-.22cm} Let $\PP$ be a random walk program as
  in  Def.\ \hyperlink{def-univariate-program}{\arabic{def-univariate-program}} that is PAST and let
  $\lambda_1,\ldots,\lambda_c$ be the roots of its characteristic polynomial 
 with multiplicities $v_1,\ldots,v_c$. Moreover, let $C(x) = C_{const} = \tfrac{1}{p'}$ if
 $p' > 0$ and $C(x) = C_{lin} \cdot x = -\tfrac{1}{\mu_\PP} \cdot x$ if $p' = 0$. Then the
 expected runtime of $\PP$ is $rt^\PP_x = 0$ for $x \leq 0$ and

 \vspace*{-.4cm}

 \[ rt^\PP_x \; = \;
C(x)\;\; + \; \; \sum\nolimits_{1 \leq j \leq c, \; |\lambda_j|\le 1} \; \; \;
\sum\nolimits_{0\leq u \leq v_j-1} \; a_{j,u} \cdot
\lambda_j^x \cdot x^u \quad \text{for $x > 0$},\]
where 
the coefficients
$a_{j,u}$ are the unique solution of the $k$ linear equations:

\vspace*{-.4cm}

\begin{equation}
  \label{initial}
0\; = \; C(x) +\sum\nolimits_{1 \leq j \leq c, \; |\lambda_j|\le 1}  \,  \sum\nolimits_{0\leq u \leq v_j-1}a_{j,u} \cdot
\lambda_j^x \cdot x^u \quad \text{for $-k+1 \leq x \leq 0$}
\end{equation}
So in the special case where $k = 0$, we have $rt^\PP_x = C(x) = C_{const} =
 \tfrac{1}{p'}$ for $x > 0$.
\end{restatable}

\noindent{}Thus for $x>0$, the expected runtime $rt^\PP_x$ can be computed by summing up the bound $C(x)$ and an add-on 
$\sum\nolimits_{1 \leq j \leq c, \; |\lambda_j|\le 1} \; \sum\nolimits_{0\leq u \leq
  v_j-1} \ldots\;$
Since 
 $C(x)$ is an upper bound for $rt^\PP_x$ if $p' >0$ and a lower
bound for  $rt^\PP_x$ if $p'=0$,
this add-on is  non-positive if $p' > 0$ and non-negative if $p' = 0$.

\begin{example}[\cref{exmpl:truncated_race_zero} cont.]\label{exmpl:truncated_race_exact_runtime}
 {\sl By \cref{solution}, the expected runtime
   of the  program $\PP^{mod}_{race}$ from \cref{exmpl:truncated_race} is
   $rt^{\PP^{mod}_{race}}_{x} = 0$ for $x \leq 0$ and

\vspace*{-.2cm}
   
  \[ rt^{\PP^{mod}_{race}}_{x} \; = \; \tfrac{22}{3} \cdot x
  +a_1+a_2\cdot (-\tfrac{1}{2} )^x \quad \text{for $x > 0$, $\;$ cf.\ \cref{allSolutionsTortoise}.}
  \]
  The coefficients $a_1$ and $a_2$ are the unique solution of the $k=2$ linear
equations

\vspace*{-.5cm}

 \begin{align*}
    0 &= \tfrac{22}{3}\cdot 0+a_1+a_2\cdot (-\tfrac{1}{2})^0=a_1+a_2\\
    0 &= \tfrac{22}{3}\cdot
    (-1)+a_1+a_2\cdot (-\tfrac{1}{2})^{-1}=-\tfrac{22}{3}+a_1-2\cdot a_2 
 \end{align*}

 \vspace*{-.3cm}

 \noindent{}So $a_1=\tfrac{22}{9}$, $a_2=-\tfrac{22}{9}$, and hence
 $rt^{\PP^{mod}_{race}}_{x} \; = \; \tfrac{22}{3}\cdot
 x+\tfrac{22}{9}-\tfrac{22}{9}\cdot (-\tfrac{1}{2})^x$ for $x > 0$.
}
\end{example}

\noindent{}By \cref{transformation_preserves_behavior}, 
we can lift
 \cref{solution} to arbitrary CP programs $\PP$
 immediately.

\begin{restatable}[Exact Expected Runtime for CP Programs]{corollary}{coroexactruntime}\label{coro3}
 For any CP program,  its expected runtime can be
  computed exactly.
\end{restatable}

\noindent
Note that irrespective of the degree of the
 characteristic polynomial,  its roots can always 
 be approximated numerically with any chosen precision. Thus, ``exact computation''
 of the expected runtime in the corollary above
 means that a closed form for $rt^\PP_{\vec{x}}$
 can also be computed with any desired precision.

\begin{example}[Exact Expected Runtime of $\PP_{direct}$]\label{exmpl:tortoise_and_hare_direct_catch_up_cont}
  \sl
  \hspace*{-.2cm}
Reconsi-
{\makeatletter
\let\par\@@par
\par\parshape0
\everypar{} \begin{wrapfigure}[4]{r}{2.8cm}
   \centering
   \vspace*{-1.2cm}
   \hspace*{-.2cm}\fbox{\begin{minipage}{3.1cm}
 \begin{tabbing}
     \= \hspace*{.2cm}\=\kill
     \>\mbox{\rm \texttt{while}} $(x > 0) \; \{$\\
     \>\>$x=x+1$ \hspace*{.2cm}\=$[\tfrac{9}{10}];$\\[0.1cm]
     \>\>$x=0$\>$[\tfrac{1}{10}];$\\
     \>$\}$
 \end{tabbing}
        \end{minipage}}
 \end{wrapfigure}
\noindent{}der the program $\PP_{direct}$ of
  \cref{exmpl:tortoise_and_hare_direct_catch_up} with 
  the probability $p' = \tfrac{1}{10}$ for direct termination. $\PP_{direct}$ is PAST
  and its expected runtime is \emph{at most} $\tfrac{1}{p'} = 10$, cf.\
  \cref{exmpl:tortoise_and_hare_direct_catch_up}.
  The random walk
  program $\PP_{direct}^\unisup$  on the right is obtained by
the \pagebreak
  transformation of Def.\ \hyperlink{Univariate Transformation}{\arabic{Univariate
      Transformation}}.
As   $k = 0$, by \cref{solution} we obtain $rt^{\PP_{direct}^\unisup}_x =  \tfrac{1}{p'} =
10$ for $x > 0$.\linebreak
By 
 \cref{transformation_preserves_behavior}, this implies 
$rt^{\PP_{direct}}_{\column{t}{h}}
= rt^{\PP_{direct}^\unisup}_{\uni_{\PP_{direct}}\column{t}{h}} =
10$ if $\uni_{\PP_{direct}}(t,h) = t - h + 1 > 0$, i.e., 10 is indeed
the \emph{exact} expected runtime of $\PP_{direct}$.\par} \end{example}

Note that \cref{solution,coro3} imply that for any $\startvec{x} \in \IZ^r$,
 the expected runtime $rt^\PP_{\startvec{x}}$ of a CP
program $\PP$ that is PAST and has only \emph{rational} probabilities $p_{\vec{c}_1},\ldots,p_{\vec{c}_n},p' \in \mathbb{Q}$
is always an algebraic number.
Thus, one could also compute a closed form for the
exact expected runtime $rt^\PP_{\vec{x}}$ using a representation with algebraic numbers
instead of numerical approximations.

Nevertheless,
\cref{solution} may yield a
representation of $rt^\PP_x$ which con\-tains complex numbers $a_{j,u}$ and $\lambda_j$,
although $rt^\PP_x$ is always real. However, one can easily obtain a
more intuitive representation of $rt^\PP_x$ without complex numbers:

Since the characteristic polynomial $\chi_\PP$ only has real coefficients, whenever $\chi_\PP$ has
a complex root $\lambda$ of multiplicity $v$, its conjugate $\overline{\lambda}$ is also a root
of $\chi_\PP$ with the same multiplicity $v$.
So the pairwise different roots $\lambda_1, \ldots, \lambda_c$
can be distinguished into pairwise different real roots $\lambda_1, \ldots,
\lambda_s$, and into pairwise different non-real complex roots $\lambda_{s+1},
\overline{\lambda_{s+1}}, \ldots, \lambda_{s+t},
\overline{\lambda_{s+t}}$, where $c = s + 2 \cdot t$.

For any coefficients $a_{j,u}, a_{j,u}' \in \IC$
with $j \in \{s+1,\ldots,s+t\}$ and
$u \in \{0,\ldots,
v_{j}-1\}$ let
 $b_{j,u} = 2
\cdot \mathrm{Re}(a_{j,u})\in \IR$ and $b_{j,u}' = -2\cdot \mathrm{Im}(a_{j,u})\in \IR$.
Then $a_{j,u} \cdot \lambda_{j}^x
 + a_{j,u}' \cdot \overline{\lambda_{j}}^x
 \; = \;  b_{j,u} \cdot \mathrm{Re}(\lambda_{j}^x)
 + b_{j,u}' \cdot \mathrm{Im}(\lambda_{j}^x)$. 
Hence, by \cref{solution} we get the  following representation of the expected runtime
which only uses \emph{real} numbers:
\begin{equation}
  \label{final}
  \mbox{\small \hspace*{-.2cm}$rt^\PP_x\!=\!\left\{\begin{array}{lll}
 C(x) &+ \hspace*{-.5cm}
\sum_{1 \leq j \leq s, \; |\lambda_j|\le 1} \; \; \;\sum_{0\leq u \leq v_j-1}a_{j,u} \cdot
\lambda_j^x \cdot x^u\\
&+  \hspace*{-.5cm}
\sum_{s+1 \leq j \leq s+t, \; |\lambda_{j}|\le 1} \; \sum_{0 \leq u \leq
  v_{j}-1} \hspace*{-.3cm} \left(b_{j,u}\!\cdot\!\mathrm{Re}(\lambda_{j}^x)  + 
b_{j,u}'\!\cdot\!\mathrm{Im}(\lambda_{j}^x)\right) \cdot x^u, & \text{for $x > 0$}\\
0, && \text{for $x \leq 0$}
\end{array}\hspace*{-.2cm}
\right.$}\end{equation}

\noindent{}To compute $\mathrm{Re}(\lambda_{j}^x)$ and $\mathrm{Im}(\lambda_{j}^x)$,
take the polar representation of the non-real roots
$\lambda_{j} = w_j\cdot e^{\theta_j \cdot i}$.
Then 
$\mathrm{Re}(\lambda_{j}^x) = w^x_j \cdot \cos(\theta_j \cdot x)$ and
$\mathrm{Im}(\lambda_{j}^x) =w^x_j \cdot \sin(\theta_j \cdot x)$.

Therefore,  we obtain the following algorithm to deduce the exact expected
runtime automatically.

\begin{algorithm}[Computing the Exact Expected Runtime]\label{algorithm_exact}
  To infer the runtime of a CP program $\PP$ as in
  Def.\ \hyperlink{def-univariate-program}{\arabic{def-univariate-program}}
   that is
  PAST, we proceed as follows:
  \begin{enumerate}
    \item Transform $\PP$ into $\PP^\unisup$ by the transformation of
      Def.\ \hyperlink{Univariate Transformation}{\arabic{Univariate Transformation}}.
            Thus, $\PP^\unisup$ is a random walk program as in
   Def.\ \hyperlink{def-univariate-program}{\arabic{def-univariate-program}}.
    \item
      Compute the solution $C(x) = C_{const} = \tfrac{1}{p'}$ 
      resp.\ $C(x) = C_{lin}\cdot x  =- \tfrac{1}{\mu_{\PP^\unisup}} \cdot x$ of the inhomogeneous recurrence
      equation  \cref{receq}.
    \item Compute the $k+m$ (possibly complex) roots of the characteristic polynomial
      $\chi_{\PP^\unisup}$ (cf.\ \eqref{chi}) and
      keep  the $k$ roots $\lambda$ with $|\lambda|\le 1$.
          \item Determine the coefficients $a_{j,u}$ by solving the $k$ linear equations in
      \eqref{initial}.

    \item Return the solution \eqref{final} where
    $b_{j,u} = 2
    \cdot \mathrm{Re}(a_{j,u})$, $b_{j,u}' = -2
    \cdot \mathrm{Im}(a_{j,u})$,
      and for $\lambda_{j} = 
    w_j\cdot e^{\theta_j \cdot i}$ we 
    have $\mathrm{Re}(\lambda_{j}^x) = w^x_j \cdot \cos(\theta_j \cdot x)$ and
    $\mathrm{Im}(\lambda_{j}^x) =w^x_j \cdot \sin(\theta_j \cdot x)$. Moreover, $x$ must be
    replaced by $\uni_\PP(\vec{x})$.
  \end{enumerate}
\end{algorithm}

\section{Conclusion, Implementation, and Related Work}\label{Conclusion}

We presented decision procedures for termination and complexity
of classes of probabilistic programs.
They are based on the connection between the expected runtime of a program
and the smallest solution of its corresponding recurrence equation, cf.\ \cref{ConnectionSection}.
For our notion of
probabilistic pro\-grams, if the probability for leaving the loop directly is at least $p'$
for some $p' > 0$, then the program is always PAST and its expected runtime is
asymptotically constant, cf.\
\cref{BoundsSection}.
In \cref{sec:random_walk_programs} we showed that a very simple decision procedure
for AST and PAST of CP programs can be obtained by
classical results from random walk theory  and that
the expected runtime is asymptotically linear
if the program is
PAST. Based on these results,
in \cref{sec:exact} we presented our algorithm to
automatically infer a closed form for
the \emph{exact} expected runtime of CP programs
(i.e., with arbitrarily high  precision). All proofs and a collection of examples to
demonstrate our algorithm can be found in \techrep{the appendix}\cameraready{\cite{TECHREP}}.

  \paragraph{\textbf{Implementation.}}
We implemented \cref{algorithm_exact} in our tool \textsf{KoAT} \cite{KoAT},
    which was
  already one of the leading  tools for complexity analysis of (non-probabilistic) integer
  programs. The implementation is written in \textsf{OCaml} and
  uses the \textsf{Python} libraries  \textsf{MpMath} \cite{mpmath}
    and \textsf{SymPy} \cite{sympy} for solving linear equations and for finding the roots of the characteristic polynomial.
 In addition to the closed form for the exact expected runtime, our implementation can also
    compute the concrete number of expected loop iterations if the user specifies the
    initial values of the variables.
    For further details, a set of benchmarks, and to download our implementation, we refer
  to 
  \url{https://aprove-developers.github.io/recurrence/}.

\begin{example}[Computing the Exact Expected Runtime of $\PP_{race}$ Automatically]\label{Exact
    Expected Runtime of Tortoise and Hare}\linebreak
 {\sl For the tortoise and hare program $\PP_{race}$ from \cref{exmpl:tortoise_and_hare}, our
  implementation in \textsf{KoAT} computes the following expected runtime within 0.49~s on an Intel Core  i7-6500 with 8 GB memory (when selecting a precision of 2 decimal places):
\[\mbox{\scriptsize $\begin{array}{rl}
    rt^{\PP_{race}}_{\column{t}{h}}=
    &0.049\cdot{0.65}^{(t-h+1)}\;\cdot\sin \left(2.8 \cdot (t-h+1) \right) 
    - 0.35\cdot{0.65}^{(t-h+1)}\!\cdot \cos \left(2.8\cdot (t-h+1) \right)\\
    &+ 0.15\cdot{0.66}^{(t-h+1)}\! \cdot \sin \left( 2.2\cdot (t-h+1) \right)
    - 0.35\cdot{0.66}^{(t-h+1)}\! \cdot \cos \left( 2.2\cdot (t-h+1) \right)\\
    &+ 0.3\cdot{0.7}^{(t-h+1)}\! \cdot \sin \left( 1.5 \cdot (t-h+1) \right)
    - 0.39\cdot{0.7}^{(t-h+1)}\! \cdot \cos \left( 1.5\,(t-h+1) \right)\\
    &+ 0.62\cdot{0.75}^{(t-h+1)}\! \cdot \sin \left( 0.83\cdot (t-h+1) \right)
    - 0.49\cdot{0.75}^{(t-h+1)}\! \cdot \cos  \left(  0.83 \cdot (t-h+1) \right)\\
        &+\tfrac{2}{3}\cdot (t-h)  \; 
    + \; 2.3\\    
\end{array}$}\]
So when starting in a state with $t = 1000$ and $h = 0$, according to our implementation
the number of expected loop iterations is   $rt^{\PP_{race}}_{\column{1000}{0}}=
670$.
}
\end{example}

  \paragraph{\textbf{Related Work.}}
 Many techniques to
analyze (P)AST  
have been developed, which mostly rely on  ranking supermartingales, e.g.,
\cite{BournezGarnier05,Chakarov13,DBLP:conf/popl/FioritiH15,
   DBLP:journals/pacmpl/McIverMKK18,DBLP:journals/pacmpl/AgrawalC018,DBLP:conf/pldi/NgoC018,
  ChatterjeeTOPLAS18,ChatterjeePOPL2017,DBLP:conf/vmcai/FuC19}. Indeed, several of these works (e.g., \cite{BournezGarnier05,DBLP:conf/popl/FioritiH15,DBLP:journals/pacmpl/AgrawalC018,DBLP:conf/vmcai/FuC19}) present
complete criteria for (P)AST, although (P)AST is undecidable. However, the
corresponding automation of these techniques is of course incomplete.
In \cite{ChatterjeeTOPLAS18} it is shown that for
affine probabilistic programs, a superclass of our CP programs, the existence of a linear ranking supermartingale is decidable. However, 
the existence of a linear ranking supermartingale is sufficient but not necessary
for PAST or  an at most linear expected
runtime.

Classes of programs where termination is decidable have already been
studied for deterministic programs. In \cite{Tiwari04} it was shown that for a
class of linear loop programs over the reals, the halting problem is decidable.
This result was transferred to the rationals \cite{Braverman06} and under certain
conditions to integer programs \cite{Braverman06,Ouaknine15,CAV19}. Termination analysis for probabilistic programs is
substantially harder than for non-probabilistic ones
\cite{DBLP:conf/mfcs/KaminskiK15}. Nevertheless, there is some previous work on classes of
probabilistic programs where termination is decidable and asymptotic bounds on the
expected runtime are computable. For instance, 
in
\cite{DBLP:conf/fsttcs/BrazdilBE10}
it was shown that AST is decidable for certain stochastic games and
\cite{ChatterjeeFM17} presents an automatic approach for inferring asymptotic upper bounds
on the expected runtime by considering
uni- and bivariate recurrence equations.

However, our algorithm is the first
which computes a general formula (i.e., a closed form) for the \emph{exact} expected runtime of arbitrary CP
programs. To our
knowledge,
up to now such a formula was only known for the very restricted special \hspace*{-.04cm}
case \hspace*{-.04cm} of \hspace*{-.04cm}
\emph{bounded} \hspace*{-.04cm} simple \hspace*{-.04cm} random \hspace*{-.04cm} walks
(cf.\ \cite{feller50}), \hspace*{-.04cm}
i.e., \hspace*{-.04cm} programs \hspace*{-.04cm} of \hspace*{-.04cm} the
{\makeatletter
\let\par\@@par
\par\parshape0
\everypar{}
  \begin{wrapfigure}[5]{r}{3.3cm}
    \centering
    \vspace*{-.87cm}
    \hspace*{-.1cm}\fbox{\begin{minipage}{3cm}\begin{tabbing}
      \= \hspace*{.2cm}\=\kill
      \>\texttt{while} $(b > x > 0) \; \{$\\
      \>\>$x = x + 1$ \hspace*{.2cm}\=$[p];$\\[0.1cm]
  \>\>$x = x - 1$\>$[1-p];$\\
      \>$\}$
    \end{tabbing}\end{minipage}}
  \end{wrapfigure}
  \noindent{}form on the right for some $1 \geq p \geq 0$ and some $b \in \IZ$. Note that due to the
\emph{two} boundary conditions $x > 0$ and $b > x$, the resulting recurrence equation for
the expected runtime of the program only
has a \emph{single} solution $f: \IZ \to \IR_{\geq 0}$ that also satisfies $f(0) = 0$
and $f(b) = 0$. Hence, standard techniques for solving recurrence equations suffice to
compute this solution. In contrast, 
we developed an
algorithm to compute the exact expected runtime of \emph{unbounded arbitrary} CP programs
where the loop condition only has \emph{one} boundary condition $x > 0$, i.e., $x$ can
grow infinitely large. For that reason, here the challenge is to find an algorithm which
computes the \emph{smallest} solution  $f: \IZ \to \IR_{\geq 0}$ of the resulting
recurrence equation. We showed that this can be done using the information on
the asymptotic bounds of the expected runtime from \cref{BoundsSection,sec:random_walk_programs}.

\paragraph{\textbf{Future Work.}}
There are several directions for future work.
In \cref{sec:Restriction_to_Random_Walk_Programs} we reduced
CP programs to
random walk programs.
 In future work, we will consider more advanced 
 reductions in order to extend the class of probabilistic programs where termination and
 complexity are decidable.   Moreover, we want to develop techniques
 to automatically \emph{over- or under-approximate} the runtime of a program $\PP$ by the runtimes of  corresponding CP
 programs $\PP_1$ and $\PP_2$   such that $rt^{\PP_1}_{\vec{x}} \leq rt^\PP_{\vec{x}} \leq rt^{\PP_2}_{\vec{x}}$ holds for all
 $\vec{x} \in \IZ^{r}$.   
 Furthermore, we will integrate the easy inference of runtime bounds for
 CP programs into existing techniques for analyzing more general probabilistic
 programs.

\vspace*{-.25cm}
 
\subsubsection*{Acknowledgments}
We would like to thank Nicos Georgiou and
Vladislav Vysotskiy for drawing our attention to Wald's Lemma and to the work of Frank
Spitzer on random walks, and Benjamin Lucien Kaminski and Christoph Matheja for many
helpful discussions. Furthermore, we thank Tom K\"uspert who helped
with the implementation of our technique in our tool \textsf{KoAT}.

\clearpage

\bibliographystyle{splncs04}
\bibliography{recurrence}
\techrep{
\appendix
\clearpage

\noindent
\textbf{\Large Appendix}

\bigskip

\noindent
This appendix contains a collection of examples to demonstrate the application of our
algorithm in \cref{Case Studies} and 
all proofs 
(in \cref{MDP}-\ref{app:exact}).

\medskip

\section{Case Studies}\label{Case Studies}

In
this section, we demonstrate
our approach for the computation of the exact expected runtime on further examples.

 \begin{example}[Example with Direct Termination and Non-Constant Exact Runtime]\label{Example with Direct Termination and Non-Constant Runtime}
 {\sl
  As an example with $p' > 0$ where the exact expected runtime is not constant,
consider the following program $\PP$. 
  \begin{tabbing}
   \hspace*{4.5cm} \= \hspace*{.2cm}\=\kill
    \>\mbox{\rm \texttt{while}} $(x > 0) \; \{$\\
    \>\>$x = x+1$ \hspace*{.2cm}\=$[\tfrac{1}{8}];$\\[0.1cm]
    \>\>$x = x$\>$[\tfrac{1}{2}];$\\[0.1cm]
    \>\>$x = x-1$\>$[\tfrac{1}{4}];$\\[0.1cm]
    \>\>$x = 0 $\>$[\tfrac{1}{8}];$\\
    \>$\}$
  \end{tabbing}
\noindent{}The 
 characteristic polynomial is
 $\chi_\PP(\lambda)=\tfrac{1}{8} \cdot \lambda^2-\tfrac{1}{2} \cdot \lambda+\tfrac{1}{4}$. It has the
 $k+m=1+1=2$ roots $2 \pm \sqrt{2}$. So the only  root with absolute value $\leq 1$ is
 $2-\sqrt{2}$.
By \cref{solution} we obtain $rt^\PP_x = 0$ for $x \leq 0$ and
  \[rt^\PP_x  = 8+a_1\cdot (2-\sqrt{2})^x \quad \text{for $x > 0$}.\]
  Here, $a_1$ is the unique solution of the linear equation  $0 =
  8+a_1 \cdot  (2-\sqrt{2})^0 \cdot 0^0 =
  8+a_1$, i.e., $a_1 =
-8$.
So for $x>0$ we have
\[rt^\PP_x = 8-8\cdot (2-\sqrt{2})^x,\]
i.e., here the negative add-on $-8\cdot (2-\sqrt{2})^x$
is added to the upper bound 8.}
 \end{example}

\begin{example}[Example with Complex Roots]\label{Example with Complex Roots}
{\sl
   \noindent{}To show that complex roots are indeed possible,
 we apply \cref{algorithm_exact}
 to the 
 following program $\PP$, where
 $p'=0$ and $\mu_\PP=-\tfrac{13}{30}$. Thus, $C_{lin} = \tfrac{30}{13}$
 and $C(x) = \tfrac{30}{13} \cdot x$.
  \begin{tabbing}
      \hspace*{4.5cm}  \= \hspace*{.2cm}\=\kill
      \>\mbox{\rm \texttt{while}} $(x > 0) \; \{$\\
      \>\>$x = x+1$ \hspace*{.2cm}\=$[\tfrac{5}{36}];$\\[0.1cm]
      \>\>$x = x$\>$[\tfrac{1}{2}];$\\[0.1cm]
      \>\>$x = x-1$\>$[\tfrac{13}{60}];$\\[0.1cm]
      \>\>$x = x-2 $\>$[\tfrac{7}{90}];$\\[0.1cm]
      \>\>$x = x-3$\>$[\tfrac{1}{15}];$\\
      \>$\}$
    \end{tabbing}
 
    \noindent{}The characteristic polynomial
      $\chi_\PP(\lambda)\superfluous{=\tfrac{5}{36}\lambda^4-\tfrac{1}{2}\lambda^3+\tfrac{13}{60}\lambda^2+\tfrac{7}{90}\lambda
      +\tfrac{1}{15}}$ has the roots 1, 3, and the two complex roots
      $\tfrac{-1\pm \sqrt{3}\, i}{5}$.
  Hence, the $k = 3$ roots with absolute value $\leq 1$ are $1$ and $\tfrac{-1\pm
    \sqrt{3}\, i}{5}$.
 By \cref{solution} we obtain the following general solution:
    \[f(x)=\tfrac{30}{13} \cdot x+a_1+a_2 \cdot (\tfrac{-1+\sqrt{3}\, i}{5})^{x} +
    a_3 \cdot (\tfrac{-1-\sqrt{3}\, i}{5})^{x} \; \text{for $x > -3$}\]
    The coefficients $a_1,a_2,a_3$ are determined by the three
linear equations $0 = f(x)$ for $-2 \leq x \leq 0$, cf.\ \cref{initial}.
  They have the unique solution $a_1=\tfrac{180}{169}$, $a_2=-{\tfrac {90}{169}}-{\tfrac
    {2}{169}}\cdot \sqrt {3} \, i$, and $a_3=-{\tfrac {90}{169}}+{\tfrac {2}{169}}\cdot \sqrt
  {3} \, i$. Thus,  $b_2=2\cdot
      \mathrm{Re}(a_2)=-\tfrac{180}{169}$, and $\dashed{b}_2=-2\cdot
      \mathrm{Im}(a_2)=\tfrac{4}{169}\cdot \sqrt{3}$.
      The polar representation of   $\lambda
= \tfrac{-1 + \sqrt{3}\, i}{5}$ is $\tfrac{2}{5}\cdot e^{\tfrac{2\pi}{3} \cdot i}$. Hence,
      $\mathrm{Re}(\lambda^x)=(\tfrac{2}{5})^x\cdot\cos(\tfrac{2\pi}{3}\cdot x)$ and
      $\mathrm{Im}(\lambda^x)=(\tfrac{2}{5})^x\cdot\sin(\tfrac{2\pi}{3}\cdot x)$. 
Thus, we get $rt^\PP_x=0$ for $x\le 0$ and for $x>0$ we have 
      \[rt^\PP_x=\tfrac{30}{13} \cdot x+\tfrac{180}{169}-\tfrac{180}{169}\cdot
      \left(\tfrac{2}{5}\right)^{x}\!\!\cdot\, \cos\left(\tfrac{2\pi}{3}\cdot x\right)
      +\tfrac{4}{169}\cdot \sqrt{3}\cdot \left(\tfrac{2}{5}\right)^x\!\!\cdot\, \sin\left(\tfrac{2\pi}{3}\cdot x\right).\]}
\end{example}

\begin{example}[Example with Root of Higher Multiplicity]
{\sl  As an example where the characteristic polynomial has a root with multiplicity greater
  than 1, consider the following program $\PP$.
  \begin{tabbing}
    \hspace*{4.5cm} \= \hspace*{.4cm}\=\kill
    \>\emph{\texttt{while}} $(x > 0) \; \{$\\
    \>\>$x = x+1$ \hspace*{1cm}\=$[\tfrac{5}{21}];$\\[0.1cm]
    \>\>$x = x$\>$[\tfrac{4}{7}];$\\[0.1cm]
    \>\>$x = x-1$\>$[\tfrac{3}{35}];$\\[0.1cm]
    \>\>$x = x-2 $\>$[\tfrac{7}{75}];$\\[0.1cm]
    \>\>$x = x-3$\>$[\tfrac{2}{175}];$\\
    \>$\}$
  \end{tabbing}
  We use the approach of \cref{algorithm_exact} to infer the exact expected runtime. Step
  1 is not necessary, since we already have a random walk program.
  \begin{enumerate}
    \item[2.] We have $p'=0$, $\mu_\PP=-\tfrac{12}{175}$, and thus, $C_{lin} = \tfrac{175}{12}$. 
    
    \item[3.] The characteristic polynomial has the
    degree $k+m=3+1=4$ and is given by
      $\chi_\PP(\lambda)=\tfrac{5}{21} \cdot \lambda^4-\tfrac{3}{7} \cdot
      \lambda^3+\tfrac{3}{35} \cdot \lambda^2+\tfrac{7}{75} \cdot \lambda+\tfrac{2}{175}$. It 
      has 
      the roots $\lambda_1=1$  with multiplicity
      $1$, $\lambda_2=\tfrac{6}{5}$ with multiplicity $1$, and $\lambda_3=-\tfrac{1}{5}$ with multiplicity $2$. 
                    Hence, 
the three roots with absolute value $\leq 1$ are $1$ and $-\tfrac{1}{5}$
    with multiplicity 2. As proved in \cref{Number of Roots Lemma} we have $1+2=3=k$ such roots counted with multiplicity.
    \item[4.] By \cref{solution}, the general solution is
    \[f(x)=\tfrac{175}{12} \cdot x+a_{1,0}+a_{2,0}\cdot (-\tfrac{1}{5})^x +
    a_{2,1}\cdot x \cdot (-\tfrac{1}{5})^x \; \text{for $x > -3$}.\] The coefficients
$a_{1,0}$, $a_{2,0}$, and $a_{2,1}$ are determined by the following linear equations,
    cf.\ \cref{initial}:
    \begin{eqnarray*}
      0&=&f(0)=a_{1,0}+a_{2,0}\\
      0&=&f(-1)=-\tfrac{175}{12}+a_{1,0}-5\cdot a_{2,0} + 5 \cdot a_{2,1} \\
      0&=&f(-2)=-\tfrac{175}{6}+a_{1,0}+25\cdot a_{2,0}-50\cdot a_{2,1} 
    \end{eqnarray*}
   They have the unique solution
    $a_{1,0}=\tfrac{175}{36}$, $a_{2,0}=-\tfrac{175}{36}$, and $a_{2,1}=-\tfrac{35}{12}$.
    Hence, $rt^\PP_x=0$ for $x\le 0$ and for $x>0$ we have 
    \[rt^\PP_x=\tfrac{175}{12}x+\tfrac{175}{36}-\tfrac{175}{36}\cdot (-\tfrac{1}{5})^x-\tfrac{35}{12}\cdot x \cdot (-\tfrac{1}{5})^x.\]
  \end{enumerate}}
\end{example}

\begin{example}[Negative Binomial Loop from \protect{\cite[Sect.\ 5.1]{DBLP:journals/pacmpl/McIverMKK18}}]
 {\sl Consider the following program $\PP$ from \cite[Sect.\ 5.1]{DBLP:journals/pacmpl/McIverMKK18}.
  \begin{tabbing}
    \hspace*{4.5cm} \= \hspace*{.4cm}\=\kill
    \>\emph{\texttt{while}} $(x > 0) \; \{$\\
    \>\>$x = x-1$ \hspace*{1cm}\=\kill
    \>\>$x = x$\>$[\tfrac{1}{2}];$\\[0.1cm]
    \>\>$x = x-1$ \>$[\tfrac{1}{2}];$\\
    \>$\}$
  \end{tabbing}
  The drift of this program is $\mu_\PP=-\tfrac{1}{2}<0$ and by \cref{termination_decidable}
  we can conclude that the negative binomial loop is positive almost surely
  terminating. Furthermore, as $k=1$ and $m=0$ we obtain that the expected runtime $rt^\PP_x$
  of this program satisfies $2\cdot x\leq rt^\PP_x\leq 2\cdot x$ for all $x>0$ by \cref{bounds_runtime_constant_probability_programs}, i.e., 
    \[rt^\PP_x=\begin{cases} 2 \cdot x, & \text{ if } x>0\\ 0,&  \text{ if } x \leq 0\end{cases}\] 
    So with our approach, the expected runtime of
    this example can be determined with clearly less effort than with the technique
  presented in \cite{DBLP:journals/pacmpl/McIverMKK18}. On the other hand, the reasoning of
  \cite{DBLP:journals/pacmpl/McIverMKK18} can be applied to arbitrary probabilistic
  programs which may even include non-determinism.}
\end{example}

\begin{example}[Symmetric Random Walk]\label{Symmetric Random Walk}
 {\sl Consider the following program $\PP$.
  \begin{tabbing}    
    \hspace*{4.5cm} \= \hspace*{.4cm}\=\kill
    \>\emph{\texttt{while}} $(x > 0) \; \{$\\
    \>\>$x = x-1$ \hspace*{1cm}\=\kill
    \>\>$x = x+1$\>$[\tfrac{1}{2}];$\\[0.1cm]
    \>\>$x = x-1$ \>$[\tfrac{1}{2}];$\\
    \>$\}$
  \end{tabbing}
  One easily calculates the drift $\mu_\PP=\tfrac{1}{2}-\tfrac{1}{2}=0$. So by
  \cref{termination_decidable} we immediately obtain the well-known result that this program is almost
  surely terminating but not positive almost surely terminating, i.e., the expected
  runtime is infinite.}
\end{example}

\begin{example}[Example with Irrational Runtime from \protect{\cite[Ex.\ 5.1]{ChatterjeeTOPLAS18}}]
 {\sl Consider the following program $\PP$ which was presented
  in \cite[Ex.\ 5.1]{ChatterjeeTOPLAS18} to show that expected runtimes can be irrational.
  \begin{tabbing}    
    \hspace*{4.5cm} \= \hspace*{.4cm}\=\kill
    \>\emph{\texttt{while}} $(x > 0) \; \{$\\
    \>\>$x = x-1$ \hspace*{1cm}\=\kill
    \>\>$x = x+1$\>$[\tfrac{1}{2}];$\\[0.1cm]
    \>\>$x = x-2$ \>$[\tfrac{1}{2}];$\\
    \>$\}$
  \end{tabbing}
  Its drift is $\mu_{\PP}=\tfrac{1}{2}\cdot 1 + \tfrac{1}{2}\cdot (-2) = -\tfrac{1}{2} <0$,
  so by \cref{termination_decidable} this program is indeed PAST. As $k=2$, we obtain
  the following  bounds on the expected runtime by \cref{bounds_runtime_constant_probability_programs} for any
  positive initial value $x>0$:
  \[2\cdot x \leq rt^\PP_x\leq 2 \cdot x + 2\]
 The characteristic polynomial of this program is $\chi_\PP(x)=\tfrac{1}{2}\cdot
  x^3-x^2+\tfrac{1}{2}$. It has the three roots $1$, $\tfrac{1+\sqrt{5}}{2}$, and
  $\tfrac{1-\sqrt{5}}{2}$. So the $k=2$ roots of absolute value $\leq 1$ are
  $1$ and $\tfrac{1-\sqrt{5}}{2}$.
  By \cref{solution}, the general solution is
\[ f(x) = 2\cdot x+ a_1 + a_2 \cdot  (\tfrac{1-\sqrt{5}}{2})^{x} \; \text{for $x > -2$}.\]
The coefficients $a_1$ and $a_2$ are determined by the following equations:
 \begin{eqnarray*}
      0&=&f(0)=a_{1}+a_{2}\\
      0&=&f(-1)= -2+ a_1 + a_2 \cdot \tfrac{2}{1-\sqrt{5}}
      \end{eqnarray*}
They have the unique solution $a_1 = 3-\sqrt{5}$ and $a_2 = \sqrt{5}-3$. Hence, 
  we infer the following exact expected
  runtime for 
  $x>0$:
  \[rt^\PP_x=2\cdot x+3-\sqrt{5}+ (\sqrt{5}-3) \cdot  (\tfrac{1-\sqrt{5}}{2})^{x}\]
  So in particular, $rt^\PP_1=1+\sqrt{5}$.
The expected runtime obtained in
\cite[Ex.\ 5.1]{ChatterjeeTOPLAS18} is 
slightly different (they obtain  $2\cdot (5+\sqrt{5})$), because  \cite{ChatterjeeTOPLAS18} counts the number of executed
statements whereas we count loop iterations.}
\end{example}

\begin{example}[Example from \protect{\cite[Sect.\ 3.1]{DBLP:conf/pldi/NgoC018}}]
 {\sl Consider the following program $\PP$. It was used in
  \cite[Sect.\ 3.1]{DBLP:conf/pldi/NgoC018} to show how one can infer the expected runtime
  of a probabilistic program by solving a recurrence equation.
  However, the authors of \cite{DBLP:conf/pldi/NgoC018}  
conclude that
recurrence equations are not well suited for runtime analyses, while our paper shows that
for CP programs, an automated runtime analysis based on recurrence equations is
feasible.
  \begin{tabbing}    
    \hspace*{4.5cm} \= \hspace*{.4cm}\=\kill
    \>\emph{\texttt{while}} $(x > 0) \; \{$\\
    \>\>$x = x-1$ \hspace*{1cm}\=\kill
    \>\>$x = x+1$\>$[\tfrac{1}{4}];$\\[0.1cm]
    \>\>$x = x-1$ \>$[\tfrac{3}{4}];$\\
    \>$\}$
  \end{tabbing}
  Its drift is $\mu_{\PP}=\tfrac{1}{4}\cdot 1 + \tfrac{3}{4}\cdot (-1) = -\tfrac{1}{2}
  <0$, so by \cref{termination_decidable} this program is indeed PAST.
  By \cref{bounds_runtime_constant_probability_programs},
we can infer the following  bounds on the expected runtime for any positive initial value $x>0$:
    \[2\cdot x \leq rt^\PP_x\leq 2 \cdot x.\]
Hence, in this example we can directly conclude that for any $x>0$ the expected runtime is
$rt^\PP_x=2 \cdot x$,
without having to solve the corresponding recurrence equation with
 \cref{solution} resp.\ \cref{algorithm_exact}.}
\end{example}

\section{Proofs for \cref{ConnectionSection}}\label{MDP}

We begin with introducing some auxiliary definitions that will be needed in the
proofs. To define the
\emph{run} of a program,
we use the  ``Kronecker-Delta''
where for any $\vec{y}, \vec{z} \in \IZ^{r}$ with $\vec{y} \neq \vec{z}$ we have $\delta_{\vec{y},\vec{z}} = 0$ and $\delta_{\vec{y},\vec{y}}
=1$.

\begin{definition}[Run of a Program]\label{Run of a Program}
  For any program $\PP$ as in \cref{def-program},
a \emph{run}  is an infinite sequence
  $\run{\vec{z}_0,\vec{z}_1,\vec{z}_2, \ldots}
  \in (\IZ^{r})^{\omega}$ and a   
   \emph{prefix run} is a finite sequence
   $\run{\vec{z}_0,\vec{z}_1,\ldots, \vec{z}_j} \in (\IZ^{r})^{j+1}$ for
   some $j \in \IN$.
   For a prefix run $\pi$, its \emph{cylinder set} $\Cyl^{\IZ^{r}}(\pi) \subseteq  (\IZ^{r})^\omega$
consists of all runs with prefix $\pi$.

\noindent{}For any  initial value $\startvec{x} \in \IZ^{r}$ of the program variables, we define a function
  $pr^\PP_{\startvec{x}}$ 
that maps any prefix run $\pi$ to its probability
 (i.e., $0 \leq pr^\PP_{\startvec{x}}(\pi) \leq 1$).
 Thus, for any prefix run $\run{\vec{z}_0,\vec{z}_1,\ldots, \vec{z}_j}$,
 let
 $pr^\PP_{\startvec{x}}(\run{\vec{z}_0}) = \delta_{\startvec{x},\vec{z}_0}$ and if $j \geq
 1$, we define:
\[\begin{array}{rcl}
pr^\PP_{\startvec{x}}(\run{\vec{z}_{0},\ldots,\vec{z}_j}) &\!=\!\!&\left\{ \begin{array}{l}
  pr^\PP_{\startvec{x}}(\run{\vec{z}_0,\ldots,\vec{z}_{j-1}}) \cdot (p_{\vec{z}_j - \vec{z}_{j-1}}(\vec{z}_{j-1}) +
  \delta_{\vec{z}_j,\vec{d}} \cdot p'(\vec{z}_{j-1})),\\
  \hspace*{\fill} \!\text{if $\vec{a}\bullet\vec{z}_{j-1} > b$}\\
  pr^\PP_{\startvec{x}}(\run{\vec{z}_0,\ldots,\vec{z}_{j-1}}) \cdot \delta_{\vec{z}_{j-1},\vec{z}_j},
    \hspace*{\fill} \text{if $\vec{a}\bullet\vec{z}_{j-1} \leq b$}
  \end{array}\right.
\end{array}\]
\end{definition}

\begin{example}[Run in $\PP_{race}$]\label{run for tortoise}
{\sl For  $\PP_{race}$ from \cref{exmpl:tortoise_and_hare} and a start configuration where
the tortoise is 10 steps ahead of the hare (e.g., $\startvec{x} = \column{11}{1}$), the
prefix run $\run{\column{11}{1}, \column{12}{1}, 
\column{13}{6}}$ has the probability $pr^{\PP_{race}}_{\column{11}{1}}\left(\run{\column{11}{1}, \column{12}{1},
\column{13}{6}}\right) = \delta_{\column{11}{1}, \column{11}{1}} \, \cdot \, p_{\column{12}{1} -
  \column{11}{1}}\column{11}{1} \, \cdot \,  p_{\column{13}{6} -
  \column{12}{1}}\column{12}{1} =  p_{\column{1}{0}}\column{11}{1} \, \cdot \,
p_{\column{1}{5}}\column{12}{1} = \tfrac{6}{11} \, \cdot \, 
\tfrac{1}{22} = \tfrac{3}{121}$.
  So we take into account
whether the prefix run starts with $\startvec{x} = \column{11}{1}$ and multiply the probability to get
from $\vec{x}=\column{11}{1}$ to $\vec{x}=\column{12}{1}$ with the probability to get from $\vec{x}=\column{12}{1}$ to $\vec{x}=\column{13}{6}$.}
\end{example}

\noindent{}In our setting, we regard a measurable space $(\Omega, \F{F})$ where
$\Omega$ is the set of runs $(\IZ^{r})^\omega$ and we want to measure the
probability that a run starts with a certain sequence $\pi$ of numbers. So we
regard the smallest $\sigma$-field $\F{F}^{\IZ^{r}}$ that contains all cylinder sets
$Cyl^{\IZ^{r}}(\pi)$ for all prefix runs $\pi$. Moreover, we consider the
probability space $((\IZ^{r})^\omega, \F{F}^{\IZ^{r}}, \IP_{\startvec{x}}^\PP)$. Here,
the probability measure $\IP_{\startvec{x}}^\PP$ for a program $\PP$ is defined
such that the probability  that a run is in $Cyl^{\IZ^{r}}(\pi)$ is the probability
$pr^\PP_{\startvec{x}}(\pi)$ of the prefix run $\pi$.

\begin{definition}[Probability Measure for a Program]\label{probability_measures}
   For any program $\PP$ as in \cref{def-program} and any $\startvec{x} \in \IZ^{r}$,
   let $\IP_{\startvec{x}}^\PP:\F{F}^{\IZ^{r}} \to
   [0,1]$  be the unique
  \emph{probability measure} such that  we have
  $\IP_{\startvec{x}}^\PP(Cyl^{\IZ^{r}}(\pi))=pr_{\startvec{x}}^\PP(\pi)$ for any prefix run $\pi$.
\end{definition}

\begin{example}[Probability Measure for $\PP_{race}$]
{\sl 
  $Cyl^{\IZ^{2}}(\langle\column{11}{1},
  \column{12}{1},
\column{13}{6}\rangle)$ consists of all runs
 that start with $\column{11}{1}$, $\column{12}{1}$, $\column{13}{6}$. If the initial
 value is $\startvec{x} = \column{11}{1}$,
 then the
probability that a run is in  $Cyl^{\IZ^{2}}(\run{\column{11}{1},
  \column{12}{1}, 
\column{13}{6}})$ is
\[\mbox{\small $\IP_{\column{11}{1}}^{\PP_{race}}(Cyl^{\IZ^{2}}(\run{\column{11}{1}, \column{12}{1},\column{13}{6}})) 
= pr_{\column{11}{1}}^{\PP_{race}}(\run{\column{11}{1},\column{12}{1},\column{13}{6}}) = \tfrac{3}{121}.$}\]}
\end{example}

\noindent{}Now we introduce a stochastic process
$\mathbf{X}^{\IZ^{r}}$ (i.e., a family of random variables $X_j^{\IZ^{r}}$) which
corresponds to the values of the program variables during a run.

\begin{definition}[Stochastic Process $\mathbf{X}^{\IZ^{r}}$]\label{stochastic_process}
 For $r\!\geq\!1$, let $\mathbf{X}^{\IZ^{r}}\!\!=\!(X^{\IZ^{r}}_{j})_{j \in \IN}$ where
  $X^{\IZ^{r}}_j\!\!:(\IZ^{r})^\omega\!\!\to \IZ^{r}$ is defined as $X^{\IZ^{r}}_{j}\!(\run{\vec{z}_0,\ldots,\vec{z}_j,\ldots
 }) = \vec{z}_j$,
  i.e., when applied to a run, $\!X^{\IZ^{r}}_{j}\!\!\!$ returns the values of the
  program variables after the $j$-th loop iteration.
\end{definition}

\noindent
So 
$X^{\IZ^{2}}_0(\run{\column{11}{1},\column{12}{1},\ldots}) =
\column{11}{1}$
and  $X^{\IZ^{2}}_1(\run{\column{11}{1},
\column{12}{1},\ldots}) = \column{12}{1}$.

Using $\mathbf{X}^{\IZ^{r}}$, the \emph{termination time} of a program
(cf.\ \cref{def:termination_time}) can 
also be defined as $T^{\PP}(\pi) =  \inf \{j \in \IN \mid \vec{a} \bullet
X^{\IZ^{r}}_j(\pi) \leq b\}$ for any $\pi \in (\IZ^{r})^\omega$. As shown in \cref{AST Def}, the
termination time is needed to define the expected runtime of a program.
We first prove that if the initial values $\startvec{x}$ violate the loop guard,
then the expected
runtime is trivially 0.

\coroviolatinginitialvalues*
\begin{proof}
  We have
  $\IP_{\startvec{x}}^\PP(X_0^{\IZ^{r}} = \startvec{x}) = \IP^\PP_{\startvec{x}}((X_0^{\IZ^{r}})^{-1}(\{\startvec{x}\}))  = \IP^\PP_{\startvec{x}}(\Cyl^{\IZ^{r}}(\startvec{x})) =
  pr^\PP_{\startvec{x}}(\startvec{x}) =\delta_{\startvec{x},\startvec{x}} = 1$.
  Thus,  for $\startvec{x}$ with $\vec{a} \bullet \startvec{x} \leq b$, we obtain  $\IP^\PP_{\startvec{x}}(T^{\PP} = 0)
= \IP^\PP_{\startvec{x}}(\vec{a} \bullet X_0^{\IZ^{r}} \leq  b) \leq \IP^\PP_{\startvec{x}}(X_0^{\IZ^{r}} = \startvec{x})
  = 1$ and hence
  $rt^\PP_{\startvec{x}} = \expecP{\startvec{x}}{T^{\PP}} = 0$.
  \qed
  \end{proof}

\noindent{}To prove \cref{correctness_of_ert} we show how to translate any probabilistic
program into a Markov Decision Process (MDP) and then re\-use existing corresponding
results for MDPs \cite{MDPsPuterman}.
In this section we recapitulate the needed concepts for MDPs and after the introduction of
any concept, we show how it is related to the corresponding notions for probabilistic programs.

\noindent{}We consider infinite time horizon MDPs, where we restrict ourselves to deterministic MDPs
without final states, i.e., to
\emph{Discrete Time Markov Chains (DTMCs)}. So there is one unique action for every state
of the MDP.

\begin{definition}[Discrete Time Markov Chain]\label{def_mdp}
  A \emph{Discrete Time Markov Chain (DTMC)}
  without final states 
   $\C{M}=(\C{S},P,rew)$ consists of 
  the following components:
  \begin{itemize}
    \item[$\bullet$]  $\C{S}$ is a set of states.
    \item[$\bullet$]  $P: \C{S} \times \C{S} \to [0,1]$ is a transition
      probability function such that for all states $s \in \C{S}$ we have
      $\sum\nolimits_{\dashed{s} \in \C{S}} \; P(s,\dashed{s}) = 1$.
    \item[$\bullet$]  $rew:\C{S} \to \IR$ is the reward function.
  \end{itemize}
\end{definition}

\noindent{}\cref{Translating Probabilistic Programs to DTMCs} shows how to translate any
probabilistic program $\PP$ to a corresponding DTMC $\C{M}_\PP$.  This is possible for our
notion of probabilistic programs, because
the values of the program variables only depend on their values
in the previous loop iteration.
To ease notation,
let the probabilities $p_{\vec{c}}(\vec{x})$  be constant zero for all $\vec{c} \in
\IZ^{r} \setminus \{\vec{c}_1, \ldots, \vec{c}_n \}$.

\begin{definition}[Translating Probabilistic Programs to DTMCs]\label{Translating Probabilistic Programs to DTMCs}
  Let $\PP$ be a  program as in \cref{def-program}. Its \emph{corresponding
    DTMC} $\C{M}_\PP = (\C{S},P,rew)$ is given by
  \begin{itemize}
    \item[$\bullet$] $\C{S}=\IZ^{r}$
    \item[$\bullet$] For states satisfying the loop guard, the probability function $P$ is induced by the probabilities $p_{\vec{c}_j}$,
      and for states that do not satisfy the loop guard, the probability to remain in the state is 1:
      \[P(s,s') = \left\{ \begin{array}{ll}
        p_{s'-s}(s) + \delta_{s',\vec{d}} \cdot p'(s), &\text{if $\vec{a} \bullet s > b$}\\
        \delta_{s,s'}, &\text{if $\vec{a} \bullet s \leq b$}
      \end{array}
      \right.\]
 \item[$\bullet$] The reward function is given by $rew(s)=\begin{cases}1, &\text{if }
     \vec{a} \bullet s > b\\ 0, &\text{if } \vec{a} \bullet s \leq b\end{cases}$
    \end{itemize}
\end{definition}
\noindent

\noindent{}For a DTMC $\C{M} = (\C{S}, P, rew)$ and each initial state $\startvec{x} \in \C{S}$, we examine a stochastic process
$\mathbf{X}^{\C{S}}$ using a probability measure $\IP^{\C{M}}_{\startvec{x}}$ for the measurable
space 
 $(\C{S}^\omega,\F{F}^{\C{S}})$.  
The  definitions of  $\F{F}^{\C{S}}$, $\IP^{\C{M}}_{\startvec{x}}$,
and $\mathbf{X}^{\C{S}}$
are generalizations of the corresponding definitions
from \cref{ConnectionSection} to 
arbitrary state spaces.

\noindent{}Moreover, instead of (prefix) runs we now regard \emph{histories} resp.\ \emph{sample
  paths} and instead of the probability $pr^\PP_{\startvec{x}}$ of a run with the
initial variable assignment $\startvec{x}$ we regard the probability  $pr^{\C{M}}_{\start{x}}$ of a
sample path with the initial state $\start{x}$.

\begin{definition}[Probability Measure for a DTMC]\label{stochastics_mdp}
  Let $\C{M}=(\C{S},P,rew)$ be a DTMC.
  \begin{itemize}
  \item[$\bullet$] A \emph{sample path}  is an infinite sequence
      $\run{s_0,s_1,s_2,\ldots} \in \C{S}^{\omega}$  and a
 \emph{history}  is a finite sequence $\run{s_0,s_1,\ldots,
 s_j} \in \C{S}^{j+1}$ for some $j \in \IN$.
 The \emph{cylinder set} $Cyl^{\C{S}}(\pi)$ of a history $\pi$ consists of all sample paths with
 prefix $\pi$.

\item[$\bullet$] For any $\start{x} \in \C{S}$,  $pr^{\C{M}}_{\start{x}}: \bigcup\nolimits_{j \in \IN}
  \; \C{S}^{j+1} \to [0,1]$ is the function that maps any history
$\run{s_0,\ldots,s_j}$
  to its probability if
  $\start{x}$ is the initial state. Thus, let
  $pr^{\C{M}}_{\start{x}}(\run{s_0})=\delta_{\start{x}, s_0}$ and
   if $j \geq
 1$, we define:
      \begin{align*}
        pr^{\C{M}}_{\start{x}}(\run{s_0,\ldots,s_j})&= pr^{\C{M}}_{\start{x}}(\run{s_0,\ldots,s_{j-1}})\cdot P(s_{j-1},s_j)
      \end{align*} 
   
    \item[$\bullet$] The (canonical)  \emph{measurable space} for a DTMC is
      $(\C{S}^\omega,\F{F}^{\C{S}})$, where $\F{F}^{\C{S}}$ is the smallest
      $\sigma$-field containing all cylinder sets
$Cyl^{\C{S}}(\pi)$ for all
      histories $\pi$.
    
  \item[$\bullet$] For any $\start{x} \in \C{S}$,  the \emph{probability measure}
    $pr^{\C{M}}_{\start{x}}:
    \F{F}^{\C{S}} \to [0,1]$ \emph{for the DTMC $\C{M}$ and the initial state $\start{x}$}
    is the unique probability measure such that for any history $\pi$ we have
    $\IP^{\C{M}}_{\start{x}}(Cyl^{\C{S}}(\pi))=pr^{\C{M}}_{\start{x}}(\pi)$.
    
    \item[$\bullet$] The stochastic process $\mathbf{X}^{\C{S}}  = (X^{\C{S}}_j)_{j \in \IN}$  is defined as
      $X^{\C{S}}_j : \C{S}^\omega \to \C{S}$, where
      $X^{\C{S}}_j(s_0,\ldots,s_j,\ldots) = s_j$.
  \end{itemize}
\end{definition}

The following corollary shows that for any probabilistic program $\PP$, the probability
spaces for $\PP$ and for its corresponding  DTMC $\C{M}_\PP$
are
the same.

\begin{corollary}[$\PP$ and $\C{M}_\PP$ Have the Same Probability Measure]\label{isomorphism}
For any program $\PP$  as in \cref{def-program} and its corresponding DTMC $\C{M}_\PP$, the corresponding probability
spaces are the same. So in particular,  we have
 $\IP^{\PP}_{\startvec{x}} = \IP^{\C{M}_\PP}_{\startvec{x}}$ for
any $\startvec{x} \in \IZ^{r}$. 
\end{corollary}
\begin{proof}
  By \cref{Translating Probabilistic Programs to DTMCs,stochastics_mdp}, the measurable
space for $\C{M}_\PP$ is $((\IZ^{r})^\omega, \F{F}^{\IZ^{r}})$, which is also the measurable space for
$\PP$.
Moreover, \cref{stochastics_mdp} implies $pr_{\startvec{x}}^{\PP} = pr_{\startvec{x}}^{\C{M}_\PP}$ and thus, 
$\IP^{\PP}_{\startvec{x}} = \IP^{\C{M}_\PP}_{\startvec{x}}$ for
any $\startvec{x} \in \IZ^{r}$.
\qed
\end{proof}

\noindent{}For
DTMCs, one is interested in the  \emph{expected total reward}.
For a DTMC $\C{M} = (\C{S}, P, rew)$ and the stochastic process $\mathbf{X}^{\C{S}}$, the expected
total reward maps any initial state $s_0 \in \C{S}$ to the expected value of
$\sum\nolimits_{j\in \IN} \; rew(X^{\C{S}}_{j})$
under the probability measure $\IP^{\C{M}}_{s_0}$ (if this expected value exists).
Note that if $rew(s) \geq 0$ for all $s \in \C{S}$, then
the sum $\sum\nolimits_{j\in \IN} rew(X^{\C{S}}_{j}):\C{S}^\omega \to \overline{\IR_{\geq 0}}$
is a non-negative\footnote{The non-negativity of $rew$
  ensures that the infinite sum of all  $rew(X^{\C{S}}_{j})$ is a value in
  $\overline{\IR_{\geq 0}}$. In contrast, if we have positive and negative rewards, then
  the infinite sum might diverge and neither converge to $-\infty$ nor to $\infty$.}
  random variable. Hence, its expected value under the probability measure
$\IP_{s_0}^{\C{M}}$ is well defined.
In particular, this holds for the DTMCs $\C{M}_\PP$ corresponding to programs $\PP$,
because for any run  $\pi =
\run{\vec{z}_0, \vec{z}_1,  \ldots} \in (\IZ^{r})^\omega$,  $rew(X^{\IZ^{r}}_{j}(\pi)) = rew(\vec{z}_j)$ is 1 if the
$j$-th tuple $\vec{z}_j$ in the run does not violate 
the loop condition $\vec{a} \bullet \vec{z}_j > b$
and 0, otherwise (i.e., $rew(\vec{z}) \in \{0,1\}$
for all $\vec{z} \in \IZ^{r}$).

\begin{definition}[Expected Total Reward]\label{Expected total reward}
  Let $\C{M}=(\C{S},P,rew)$ be a\linebreak DTMC. For any $s_0 \in \C{S}$, 
the \emph{expected total reward}
$tr^\C{M}_{s_0} \in \IR \cup \{-\infty, \infty\}$
of $\C{M}$ is
  \[tr^\C{M}_{s_0}\; = \; \lim \limits_{u \to \infty} \expecM{s_0}{\sum\limits_{0 \leq j
      \leq u}
    rew(X^{\C{S}}_j)}\]
  whenever this limit exists in $\IR \cup \{-\infty,\infty\}$.
As argued above, the limit always exists in the special case of non-negative
rewards.
Therefore, in the case where $rew(s) \in \{0,1\}$ for all $s \in \C{S}$, we have \[tr^\C{M}_{s_0} = 
 \sum_{u \in \overline{\IN}} u \cdot \IP_{s_0}^{\C{M}}(
\sum\limits_{j \in \IN}
rew(X^{\C{S}}_j) = u).\]
\end{definition}

\noindent{}The following lemma shows the connection between the
termination time and the total reward of a run. In the following, we say that a run $\pi =
\run{\vec{z}_0, \vec{z}_1, \ldots}$ is
\emph{constant on violating states} if $\vec{a}\bullet\vec{z}_j \leq b$ implies $\vec{z}_j = \vec{z}_{j+1}$ for all $j \in \IN$.

\begin{lemma}[Total Reward is Termination Time]\label{Total Reward is Termination Time}
  Let $\PP$ be a program as in \cref{def-program}.
  For every run $\pi$ that is constant on violating states, we have
  $\sum\nolimits_{j\in \IN} \; rew(X^{\IZ^{r}}_{j}(\pi)) = T^{\PP}(\pi)$.
\end{lemma}
\begin{proof}
  First, we show that the equality holds for runs $\pi = \run{\vec{z}_0,\vec{z}_1,\ldots}$
 where $T^{\PP}(\pi) =u<\infty$. So $\vec{a}\bullet\vec{z}_{j}>b$ for all $j<u$ and since $\pi$ is constant on violating states, we 
 have $\vec{a}\bullet\vec{z}_{j} \leq b$ for all $j \geq u$. Here we obtain
 \begin{align*}
   \sum\limits_{j \in \IN}  rew(X^{\IZ^{r}}_{j}(\pi)) 
  &=\sum\limits_{j \in \IN}  rew(\vec{z}_j) \\
&= \sum\limits_{0 \leq j < u} 1 + \sum\limits_{j \geq u} 0\\
   &= u\\
   &= T^{\PP}(\run{\vec{z}_0,\vec{z}_1,\ldots})\\
   &= T^{\PP}(\pi).
   \end{align*}

  Now we consider a run $\pi = \run{\vec{z}_0,\vec{z}_1,\ldots}$ such that
  $T^{\PP}(\pi)=\infty$, i.e., $\vec{a}\bullet\vec{z}_{j}>b$ for all $j \in \IN$. Then we have 
  \begin{align*}
     \sum\limits_{j \in \IN} rew(X^{\IZ^{r}}_{j}(\pi))
    &= \sum\limits_{j \in \IN} rew(\vec{z}_j)\\
    &= \sum\limits_{j \in \IN} 1 \\
    &= \infty\\
  &= T^{\PP}(\run{\vec{z}_0,\vec{z}_1,\ldots})\\
   &= T^{\PP}(\pi). \hspace*{4.1cm} \qed \hspace*{-4.1cm}
  \end{align*} 
\end{proof}

\noindent{}With \cref{isomorphism,Total Reward is Termination Time} 
we can show that the expected runtime of a program $\PP$ is identical to the
expected total reward of its corresponding DTMC $\C{M}_\PP$. This is the crucial theorem which allows us to
apply results on DTMCs also for probabilistic programs.

\begin{theorem}[Expected Total Reward is Expected Runtime]\label{etr_is_ert}
For any program $\PP$ as in \cref{def-program}, the expected runtime of $\PP$ and the expected total reward of the
corresponding DTMC
$\C{M}_\PP$ are the same, i.e., for any $\startvec{x} \in \IZ^{r}$ we have
$rt^\PP_{\startvec{x}} = tr^{\C{M}_\PP}_{\startvec{x}}$.
\end{theorem}
\begin{proof}
Due to \cref{Expected total reward}
we have
$tr^{\C{M}_\PP}_{\startvec{x}} = 
\sum\nolimits_{u \in \overline{\IN}} \; u \cdot \IP_{\startvec{x}}^{\C{M}_\PP}(A_u)$,
where $A_u  =
\{\pi \in (\IZ^{r})^\omega \mid \sum\nolimits_{j \in \IN} \;
rew(X^{\IZ^{r}}_j(\pi)) = u \}$. Note that $pr^{\C{M}_\PP}_{\startvec{x}}(\pi) = 0$ if $\pi$ is not
constant on violating states. Thus, $\IP_{\startvec{x}}^{\C{M}_\PP}(A_u) =
\IP_{\startvec{x}}^{\C{M}_\PP}(A_u')$ where
\[\mbox{\small $\begin{array}{rcl}
A_u' &=& \{\pi\!\in\!(\IZ^{r})^\omega \mid \sum\nolimits_{j \in \IN} \;
rew(X^{\IZ^{r}}_j(\pi)) = u \text{ and $\pi$ is constant on violating states} \} \\
&=&
\{\pi\!\in\!(\IZ^{r})^\omega \mid 
T^{\PP}(\pi)\!=\!u \text{ and $\pi$ is constant on violating states} \} \; \text{by \cref{Total Reward is Termination Time}.}
\end{array}$}\]
Hence, we obtain
\[ \begin{array}{rcll}
  tr^{\C{M}_\PP}_{\startvec{x}}  &=&
  \sum\limits_{u \in \overline{\IN}}  u \cdot \IP_{\startvec{x}}^{\C{M}_\PP}(A_u')\\
    &=&
  \sum\limits_{u \in \overline{\IN}}  u \cdot \IP_{\startvec{x}}^{\PP}(A_u')&\text{by
    \cref{isomorphism}.}
\end{array}\]
Note that $pr^{\PP}_{\startvec{x}}(\pi) = 0$ if $\pi$ is not
constant on violating states. Thus, $\IP_{\startvec{x}}^{\PP}(A_u') =
\IP_{\startvec{x}}^{\PP}(A_u'')$ where
$A_u'' = \{\pi \in (\IZ^{r})^\omega \mid 
T^{\PP}(\pi) = u \}$.
So we get
\[ \begin{array}{rcll}
  tr^{\C{M}_\PP}_{\startvec{x}}  &=& \sum\limits_{u \in \overline{\IN}}  u \cdot
  \IP_{\startvec{x}}^{\PP}(A_u'')&\\[.5cm]
  &=& \sum\limits_{u \in \overline{\IN}}  u \cdot
  \IP_{\startvec{x}}^{\PP}(T^{\PP} = u)&\\[.5cm]
  &=& \expecP{\startvec{x}}{T^{\PP}} &\\[.2cm]
  &=&rt^\PP_{\startvec{x}}. \hspace*{5.75cm} \qed \hspace*{-5.75cm}
\end{array}\]
\end{proof}

\noindent{}Now we introduce the transformer $\C{L}$ that is used for DTMCs and corresponds to the expected
runtime transformer for probabilistic programs. In the following, we restrict ourselves to
DTMCs with non-negative rewards to ensure that the expected total reward
exists.

\begin{definition}[$\C{L}^{\C{M}}$, cf.\ \protect{\cite[Eq.\ 7.1.5]{MDPsPuterman}}]\label{mapping_L}
  Let $\C{M}=(\C{S},P,rew)$ be a DTMC with only non-negative
  rewards.
  We define the mapping
  $\C{L}^{\C{M}}:(\C{S} \to \overline{\IR_{\geq_0}})\to (\C{S} \to
  \overline{\IR_{\geq_0}})$ such that for every function $f:\C{S} \to
  \overline{\IR_{\geq_0}}$ and every $s \in \C{S}$, we have
  \[\C{L}^{\C{M}}(f)(s)=rew(s)+\sum\limits_{\dashed{s}\in \C{S}} P(s,\dashed{s})\cdot f(\dashed{s}).\]
\end{definition}

The following corollary shows that the expected runtime transformer  $\C{L}^\PP$ of a
program $\PP$ is the same as the transformer $\C{L}^{\C{M}_\PP}$ of the corresponding DTMC $\C{M}_\PP$.

\begin{corollary}[$\C{L}^{\C{M}_\PP}$ is Expected Runtime Transformer $\C{L}^P$]\label{ert_reward}
For any program $\PP$, the expected runtime transformer
  $\C{L}^\PP$ from \Cref{expected_runtime_transformer} is identical to the transformer
 $\C{L}^{\C{M}_\PP}$ from \Cref{mapping_L}.   
\end{corollary}

\begin{proof}
 Let $\PP$ be a program as in \cref{def-program} and let $\C{M}_\PP = (\IZ^{r}, P,
 rew)$. Consider an arbitrary function $f:\IZ^{r}
 \to \overline{\IR_{\geq_0}}$ and an $s \in \IZ^{r}$. If $\vec{a} \bullet s \leq b$
 then $rew(s)=0$, $P(s,s)=1$, and 
$P(s,s')=0$ for $s' \neq s$.
Hence 
  \[\begin{array}{rcl}\C{L}^{\C{M}_\PP}(f)(s)&=&rew(s)+\sum\limits_{\dashed{s}\in \C{S}}
  P(s,\dashed{s})\cdot f(\dashed{s})\\
  &=&rew(s)+f(s)\\
  &=& 0 + f(s)\\
  &=& f(s)\\
  &=&\C{L}^{\PP}(f)(s).
  \end{array}\]
  If $\vec{a} \bullet s > b$ then
  \[\begin{array}{rcl}
  \C{L}^{\C{M}_\PP}(f)(s)&=&rew(s)+\sum\limits_{\dashed{s}\in \C{S}}
  P(s,\dashed{s})\cdot f(\dashed{s})\\
  &=& 1 + \sum\limits_{\dashed{s}\in \C{S}} \left( p_{s'-s}(s) + \delta_{s', \vec{d}}
  \cdot p'(s) \right) \cdot f(s')\\
  &=& 1 +   \sum\limits_{1 \leq j \leq n}  p_{\vec{c}_j}(s) \cdot
  f(s+\vec{c}_j)  + p'(s) \cdot f(\vec{d})\\
    &=&\C{L}^{\PP}(f)(s). \hspace*{6.1cm} \qed \hspace*{-6.1cm}
  \end{array}
  \]
 \end{proof}

\noindent{}Now that we know that the transformers $\C{L}^\PP$ and $\C{L}^{\C{M}_\PP}$ are the same, we
can use existing results on DTMCs to obtain results for programs
$\PP$. More precisely, for any DTMC $\C{M} = (\C{S},P,rew)$ with only non-negative rewards,
it is known that the expected total reward function $tr^{\C{M}}: \C{S} \to
\overline{\IR_{\geq 0}}$ with $tr^{\C{M}}(s_0) = tr^{\C{M}}_{s_0}$ for any $s_0 \in \C{S}$ 
is a fixpoint of $\C{M}$'s
transformer $\C{L}^{\C{M}}$.

\begin{theorem}[Expected Total Reward is Fixpoint]\label{fp_mdps}
  Let $\C{M}$ be a DTMC with only non-negative rewards. Then
$tr^{\C{M}}$ is a fixpoint of $\C{L}^{\C{M}}$.
\end{theorem}

\begin{proof}
  The proof can be found in \cite[Thm.\ 7.1.3]{MDPsPuterman}. Note that it requires the
  assumption that the expected total reward exists \cite[Assumption
    7.1.1]{MDPsPuterman} which is ensured by a non-negative reward function.
  \qed
\end{proof}

\noindent{}Moreover, the expected total reward function is smaller or equal than any other fixpoint of $\C{L}^{\C{M}}$
(and than every function $f$ which satisfies the inequality $f \geq \C{L}^{\C{M}}(f)$).

\begin{theorem}[Expected Total Reward is Smaller Than Other Fix-\linebreak points]\label{upper_bound_mdps}
  Let $\C{M}= (\C{S},P,rew)$ be a DTMC with only non-negative
  rewards and let there be a function $f:\C{S} \to\overline{\IR_{\geq_0}}$ such that $f \geq \C{L}^{\C{M}}(f)$. Then $f \geq tr^{\C{M}}$.
\end{theorem}
\begin{proof}
  The proof of the finite case, i.e., $f(s)< \infty$ for all $s \in \C{S}$,
  can be found in \cite[Thm.\ 7.2.2]{MDPsPuterman}. Note that in our case there is a
  unique strategy (since we restrict ourselves to DTMCs)
  and we have only non-negative rewards. Therefore, the proof holds for functions $f$ that
  map to
  infinity as well.
  \qed
\end{proof}

\noindent{}\cref{fp_mdps,upper_bound_mdps} imply that the expected total reward function $tr^{\C{M}}$
is the least fixpoint
of the transformer $\C{L}^{\C{M}}$.

\begin{corollary}[Expected Total Reward is Least Fixpoint]\label{lfp_mdps}
  Let $\C{M}= (\C{S},P,rew)$ be a DTMC with only non-negative rewards. Then
 $tr^{\C{M}}$ is the least fixpoint of $\C{L}^{\C{M}}$, i.e., for
any $s_0 \in \C{S}$ we have $\lfp(\C{L}^{\C{M}})(s_0) = tr^{\C{M}}_{s_0}$.
\end{corollary}

\noindent{}The following theorem shows that $\C{L}^{\C{M}}$ is continuous for any DTMC $\C{M}$ with
 only non-negative rewards. This is needed to apply
Kleene's Fixpoint Theorem, i.e., to show that the least fixpoint of  $\C{L}^{\C{M}}$ is $\sup
\{ \0, \C{L}^{\C{M}}(\0), (\C{L}^{\C{M}})^2(\0), \ldots \}$.
\begin{theorem}[Continuity of $\C{L}^{\C{M}}$, cf.\  \protect{\cite[Lemma 7.1.5.c]{MDPsPuterman}}]\label{continuity_of_ert}
  Let $\C{M}$ be a  DTMC with only non-negative
  rewards. Then $\C{L}^{\C{M}}$ is continuous.
\end{theorem}
\begin{proof}
  Let $S=\{f_0, f_1, \ldots\}$ be a chain in
  $\C{S}\to \overline{\IR_{\geq_0}}$,
  i.e.,
we have $f_j \leq f_{j+1}$   
for all $j \in \IN$.
Then $(\sup S)$ is the function $(\sup S): \C{S} \to\overline{\IR_{\geq_0}}$ with
$(\sup S)(s) = \sup_{j \in \IN} \left(f_j(s)\right)$ for all $s \in \C{S}$. Therefore,
for any $s \in \C{S}$ we have
  \begin{align*}
    \C{L}^{\C{M}}(\sup S)(s)&= rew(s)+\sum\limits_{\dashed{s}\in \C{S}} P(s,\dashed{s})\cdot (\sup S)(\dashed{s})\\ 
    &= rew(s)+\sum\limits_{\dashed{s}\in \C{S}} P(s,\dashed{s})\cdot \sup_{j \in \IN} \left(f_j(\dashed{s})\right)\\
    &= \sup_{j \in \IN} \left( rew(s)+\sum\limits_{\dashed{s}\in \C{S}}
   P(s,\dashed{s})\cdot  f_j(\dashed{s})\right) &\hspace*{-.3cm}\text{as all operations are linear}\\
    &=\sup_{j \in \IN}\left(\C{L}^{\C{M}}(f_j)(s)\right)\\
    &=(\sup \C{L}^{\C{M}}(S))(s),
  \end{align*}
  where $\C{L}^{\C{M}}(S) = \{\C{L}^{\C{M}}(f_0),  \C{L}^{\C{M}}(f_1),   \ldots \}$.
  \qed
\end{proof}

\noindent{}Now we can prove \cref{correctness_of_ert} which states that the expected runtime of a program $\PP$
is the least fixpoint of its expected runtime transformer $\C{L}^\PP$ and that it can be
obtained as the supremum of  $\{ \0, \C{L}^\PP(\0),  (\C{L}^\PP)^2(\0), \ldots
\}$.
As usual, a function $f: \IZ^{r} \to \overline{\IR_{\geq 0}}$ is a \emph{fixpoint} of
$\C{L}^\PP$ if $\C{L}^\PP(f) = f$.   Such a fixpoint $f$ is the  
\emph{least} fixpoint of $\C{L}^\PP$ (i.e., $f = \lfp(\C{L}^\PP)$) if  $f \leq
g$
for any other fixpoint $g$ of $\C{L}^\PP$.

\theoremcorrectnessert*
\begin{proof}
  By \cref{etr_is_ert},
the expected runtime of $\PP$ is the same as the expected total reward 
 of the corresponding
 DTMC $\C{M}_P$. 
\cref{lfp_mdps} showed that
the expected total reward is the least
fixpoint of the  transformer $\C{L}^{\C{M}_\PP}$,
 and
 $\C{L}^{\C{M}_\PP}$ is the same as the expected runtime transformer $\C{L}^\PP$
 due to \cref{ert_reward}. 

\noindent{}As the continuity of $\C{L}^{\C{M}_P} = \C{L}^P$ was shown in
 \cref{continuity_of_ert}, by Kleene's Fixpoint Theorem we have 
 $\lfp(\C{L}^P) = \sup \{ \0, \C{L}^P(\0), (\C{L}^P)^2(\0), \ldots\}$.
 \qed
 \end{proof}

\section{Proofs for \cref{BoundsSection}}\label{BoundsSectionApp}
\theoremPAST*
\begin{proof}
The expected runtime transformer  $\C{L}^\PP$ is continuous (and thus, monotonic) by
\cref{correctness_of_ert}. Hence,  by
induction on $j$ one can show that
$f \geq \C{L}^\PP(f)$ implies $f \geq (\C{L}^\PP)^j(\0)$
for any function $f: \IZ^{r} \to \overline{\IR_{\geq 0}}$ and
any $j \in \IN$.
So  $f \geq \C{L}^\PP(f)$ implies
$f \geq \sup \{ \0, \C{L}^\PP(\0),  (\C{L}^\PP)^2(\0), \ldots
      \} = \lfp(\C{L}^\PP)$. By \cref{correctness_of_ert}, this means that $f(\startvec{x}) \geq
      \lfp(\C{L}^\PP)(\startvec{x})=rt^\PP_{\startvec{x}}$ for all $\startvec{x} \in \IZ^{r}$.

Hence, to prove \cref{PAST}, it suffices to show $f \geq \C{L}^\PP(f)$ for the
function $f: \IZ^{r} \to \overline{\IR_{\geq 0}}$ with $f(\vec{x}) = \tfrac{1}{p'}$ if $\vec{a} \bullet \vec{x} > b$ and
$f(\vec{x}) =0$ if $\vec{a} \bullet \vec{x} \leq b$.

For $\vec{x}$ with $\vec{a}\bullet\vec{x} \leq b$, we have
$\C{L}^\PP(f)(\vec{x})=f(\vec{x})$. If $\vec{a}\bullet\vec{x} > b$, then we get   
  \begin{align*}
    \C{L}^\PP(f)(\vec{x}) &=\sum_{1 \leq j \leq n} p_{\vec{c}_j}(\vec{x}) \cdot f(\vec{x}+\vec{c}_j)+\dashed{p}(\vec{x})\cdot f(\vec{d}) + 1\\
                  &\leq \sum_{1 \leq j \leq n} p_{\vec{c}_j}(\vec{x}) \cdot
    \tfrac{1}{\dashed{p}}+\dashed{p}(\vec{x})\cdot 0 + 1\\ 
                   &=\tfrac{1}{\dashed{p}}\cdot\sum_{1 \leq j \leq n} p_{\vec{c}_j}(\vec{x}) + 1\\
    &= \tfrac{1-\dashed{p}(\vec{x})}{\dashed{p}} +1\quad
    \leq \quad \tfrac{1-\dashed{p}}{\dashed{p}} +1 \quad      
                  = \quad \tfrac{1}{\dashed{p}} \quad = \quad f(\vec{x}) \hspace*{.9cm} \qed \hspace*{-.9cm}
  \end{align*}
\end{proof}

\section{Proofs for \cref{sec:random_walk_programs}}\label{app:random_walk_programs}
In this section we present the proofs of \cref{sec:random_walk_programs}. It is divided into three subsections in which we will give the proofs for the respective subsections of \cref{sec:random_walk_programs}.
\subsection{Proofs for \cref{sec:Restriction_to_Random_Walk_Programs}}\label{app:Restriction_to_Random_Walk_Programs}

To prove \cref{transformation_preserves_behavior}, we need an auxiliary lemma.

\begin{lemma}[Connections between $\PP$ and $\PP^\unisup$]\label{preparation}
  Let $\PP$ be a CP program  as in \cref{def-program} and let
  $\uni_{\PP}^{\omega}$ be the function which applies $\uni_\PP$ 
  componentwise to runs. Then we have:
  
  \begin{enumerate}[(a)]
    \item $T^{\PP^\unisup} \circ \uni_{\PP}^{\omega} = T^{\PP}$
    \item
    Let 
    $\startvec{x} \in \IZ^{r}$.
    Then for any prefix run $\run{y_0, \ldots, y_j} \in \IZ^{j+1}$ we have: 
    \[\IP^{\PP^\unisup}_{\uni_{\PP}(\startvec{x})}(Cyl^\IZ(\run{y_0,\ldots,y_j})=\IP^\PP_{\startvec{x}}((\uni_{\PP}^{\omega})^{-1}(Cyl^\IZ(\run{y_0,\ldots,y_j}))).\]
    Here, for any $M \subseteq \IZ^\omega$ we have $(\uni_{\PP}^{\omega})^{-1}(M) = \{ \pi \in
    (\IZ^{r})^\omega \mid \uni_{\PP}^{\omega}(\pi) \in M \}$.
  \end{enumerate}
\end{lemma}

\begin{proof}
  \begin{enumerate}[(a)]
  \item Let
    $\run{\vec{z}_0,\vec{z_1},\ldots} \in (\IZ^{r})^{\omega}$ such
    that $T^{\PP}(\run{\vec{z}_0,\vec{z_1},\ldots}) 
    =j \in \overline{\IN}$. So if $j \in \IN$, then
    $\uni_{\PP}(\vec{z}_0), \ldots, \uni_{\PP}(\vec{z}_{j-1})>0$ and
    $\uni_{\PP}(\vec{z}_j)\leq 0$.
Similarly, if $j = \infty$, then  $\uni_{\PP}(\vec{z}_j)>0$ for
    every $j \in \IN$. So in both cases, we have
    \[\begin{array}{rcl} j\; =\;
  T^{\PP^\unisup}(\run{\uni_{\PP}(\vec{z}_0),\uni_{\PP}(\vec{z}_1),\ldots})
  &= &T^{\PP^\unisup}(\uni_{\PP}^{\omega}(\run{\vec{z}_0,\vec{z_1},\ldots}))\\
  &=&(T^{\PP^\unisup}\circ 
    \uni_{\PP}^{\omega})(\run{\vec{z}_0,\vec{z_1},\ldots}).
    \end{array}\] 

    \item
First note that 
 for any prefix run $\run{y_0, \ldots, y_j} \in \IZ^{j+1}$, we have
 \begin{equation}
   \label{claim b}
     \hspace*{-.2cm}    (\uni_{\PP}^{\omega})^{-1}(Cyl^\IZ(\run{y_0,\ldots,y_j})) \; =\hspace*{-1cm} \biguplus
\limits_{\vec{z_0},\ldots,\vec{z_j}\in \IZ^{r} \text{ such that}\atop
  \uni_{\PP}(\vec{z}_0)=y_0, \ldots, \uni_{\PP}(\vec{z}_j)=y_j} \hspace*{-1.2cm}
Cyl^{\IZ^{r}}(\run{\vec{z}_0,\ldots,\vec{z}_j}).
\end{equation}
As usual, ``$\uplus$'' denotes the disjoint union, i.e., we have
$Cyl^{\IZ^{r}}(\pi) \cap Cyl^{\IZ^{r}}(\pi')\linebreak 
 = \emptyset$ for prefix runs $\pi \neq \pi'$ of the same length. 
Note that both sides of the
 equality \cref{claim b} can be empty, i.e.,
there might not be any $\vec{z}_u$ with $\uni_\PP(\vec{z}_u) = y_u$ for some $1 \leq u
\leq j$.   For $\start{x} = \uni_\PP(\startvec{x})$,
we prove that
    \[\IP^{\PP^\unisup}_{\start{x}}(Cyl^\IZ(\run{y_0,\ldots,y_j}) \;= \; \IP^\PP_{\startvec{x}}\left(\biguplus\limits_{
      \vec{z}_0,\ldots,\vec{z}_j
        \in \IZ^{r} \text{ such that}\atop
 \uni_{\PP}(\vec{z}_0)=y_0, \ldots, \uni_{\PP}(\vec{z}_j)=y_j}  \hspace*{-1.2cm}
Cyl^{\IZ^{r}}(\run{\vec{z}_0,\ldots,\vec{z}_j})\right).\] The result then follows by
\cref{claim b}.
For the left-hand side we get
$\IP^{\PP^\unisup}_{\start{x}}(Cyl^\IZ(\langle y_0,\linebreak
\ldots,y_j\rangle) = 0$ if $y_0 \neq \start{x}$ and
otherwise, we have
\[\begin{array}{rcl}
\IP^{\PP^\unisup}_{\start{x}}(Cyl^\IZ(\run{y_0,\ldots,y_j}) &=&  \prod\limits_{1 \leq u \leq
  j} (p^\unisup_{y_{u}-y_{u-1}}+\delta_{y_u,\uni_{\PP}(\vec{d})}\cdot p')\\
&=&\prod\limits_{1 \leq u \leq j}\left(\sum\limits_{1 \leq v \leq
n, \; \vec{a} \bullet \vec{c}_v = y_{u}-y_{u-1}} \hspace*{-.7cm}
p_{\vec{c}_t}+\delta_{y_u,\uni_{\PP}(\vec{d})}\cdot p'\right). 
\end{array}\]

For the right-hand side recall that $\uni_\PP(\startvec{x}) = \start{x}$ and that we only regard
tuples $\vec{z}_0$ where  $\uni_\PP(\vec{z}_0) = y_0$. So if
$y_0 \neq \start{x}$, then
all of these tuples $\vec{z}_0$ are different from $\startvec{x}$. Hence, then the right-hand
side is also 0. Otherwise, we have the following, where $d_{\PP}=\uni_{\PP}(\vec{d})$:
\[\begin{array}{rl}
&\IP^\PP_{\startvec{x}}\left(\biguplus\limits_{\vec{z}_0,\ldots,\vec{z}_j
        \in \IZ^{r} \text{ such that }
  \uni_{\PP}(\vec{z}_0)=y_0, \ldots, \uni_{\PP}(\vec{z}_j)=y_j}
Cyl^{\IZ^{r}}(\run{\vec{z}_0,\ldots,\vec{z}_j})\right)\\
=&\IP^\PP_{\startvec{x}}\left(\biguplus\limits_{\vec{z}_1,\ldots,\vec{z}_j
        \in \IZ^{r} \text{ such that }
  \uni_{\PP}(\vec{z}_1)=y_1, \ldots, \uni_{\PP}(\vec{z}_j)=y_j} \hspace*{-.5cm}
Cyl^{\IZ^{r}}(\run{\startvec{x},\vec{z}_1,\ldots,\vec{z}_j})\right)\\
=&\sum\limits_{\vec{z}_1,\ldots,\vec{z}_j
        \in \IZ^{r} \text{ such that }
  \uni_{\PP}(\vec{z}_1)=y_1, \ldots, \uni_{\PP}(\vec{z}_j)=y_j} \hspace*{-.5cm}
\IP^\PP_{\startvec{x}}\left(Cyl^{\IZ^{r}}(\run{\startvec{x},\vec{z}_1,\ldots,\vec{z}_j})\right)\\
=& \sum\limits_{\vec{z}_1,\ldots,\vec{z}_j
        \in \IZ^{r} \text{ such that}\atop
        \uni_{\PP}(\vec{z}_1)=y_1, \ldots, \uni_{\PP}(\vec{z}_j)=y_j} \hspace*{-.5cm}
(p_{\vec{z}_1 - \startvec{x}} + \delta_{\vec{z}_1,\vec{d}}\cdot p') \cdot \prod\limits_{2 \leq u \leq j} (p_{\vec{z}_u-\vec{z}_{u-1}} + \delta_{\vec{z}_u,\vec{d}}\cdot p')\\
=& \hspace*{-.5cm}\sum\limits_{\vec{c}_{v_1},\ldots,\vec{c}_{v_j} \in
  \{\vec{c}_1,\ldots,\vec{c}_n\}  \text{ such that }
\uni_\PP(\startvec{x} + \vec{c}_{v_1}) = y_1,\atop
\uni_\PP(\startvec{x} + \vec{c}_{v_1}+ \vec{c}_{v_2}) = y_2,  \ldots,
\uni_\PP(\startvec{x} + \vec{c}_{v_1}+ \ldots + \vec{c}_{v_j}) = y_j} \hspace*{-1.5cm}
(p_{\vec{c}_{v_1}}+\delta_{y_1,d_{\PP}}\cdot p')\; \cdot \; \ldots \; \cdot\; (p_{\vec{c}_{v_j}}+\delta_{y_j,d_\PP}\cdot p')\\
\stackrel{(\dagger)}{=}&\hspace*{-.5cm}\sum\limits_{\vec{c}_{v_1},\ldots,\vec{c}_{v_j} \in
  \{\vec{c}_1,\ldots,\vec{c}_n\}  \text{ such that }
  y_1 - y_0 = \vec{a} \bullet \vec{c}_{v_1}, \ldots,
 y_j - y_{j-1} = \vec{a} \bullet \vec{c}_{v_j}} \hspace*{-3.5cm}
(p_{\vec{c}_{v_1}}+\delta_{y_1,d_{\PP}}\cdot p')\; \cdot \; \ldots \; \cdot\; (p_{\vec{c}_{v_j}}+\delta_{y_j,d_{\PP}}\cdot p')\\
=& (\sum\limits_{\hspace*{-.5cm}\vec{c}  \in
  \{\vec{c}_1,\ldots,\vec{c}_n\}  \text{ such that }
  y_1 - y_0 = \vec{a} \bullet \vec{c}} \hspace*{-1.5cm}
p_{\vec{c}}+\delta_{y_1,d_{\PP}}\cdot p') \; \cdot \; \ldots \; \cdot \;
(\hspace*{-.3cm}\sum\limits_{\hspace*{-.5cm}\vec{c}  \in
  \{\vec{c}_1,\ldots,\vec{c}_n\}  \text{ such that }
  y_j - y_{j-1} = \vec{a} \bullet\vec{c}} \hspace*{-1.5cm} p_{\vec{c}}+\delta_{y_j,d_{\PP}}\cdot p')\\
=&\prod\limits_{1 \leq u \leq j}\left(\sum\limits_{1 \leq t \leq n, \; \vec{a} \bullet \vec{c}_t
  = y_{u}-y_{u-1}} p_{\vec{c}_t}+\delta_{y_u,d_{\PP}}\cdot p'\right).
\end{array}\]

For Equation $(\dagger)$, note that  $\uni_\PP(\startvec{x} + \vec{c}_{v_1})
= \vec{a} \bullet (\startvec{x} + \vec{c}_{v_1})
 - b = \vec{a} \bullet \startvec{x}  + \vec{a} \bullet \vec{c}_{v_1}
- b =  \uni_\PP(\startvec{x}) + \vec{a} \bullet \vec{c}_{v_1} = y_0+ \vec{a} \bullet \vec{c}_{v_1}$.
Hence,
$\uni_\PP(\startvec{x} + \vec{c}_{v_1}) = y_1$ means that
$y_1 - y_0 = \vec{a} \bullet \vec{c}_{v_1}$. Similarly, $\uni_\PP(\startvec{x}
+ \vec{c}_{v_1}+ \vec{c}_{v_2}) =  y_0+ \vec{a} \bullet \vec{c}_{v_1}
 + \vec{a} \bullet  \vec{c}_{v_2} = y_1 +  \vec{a} \bullet \vec{c}_{v_2}$. So
 $\uni_\PP(\startvec{x} + \vec{c}_{v_1}+ \vec{c}_{v_2}) = y_2$ means that $y_2 - 
y_1 = \vec{a} \bullet \vec{c}_{v_2}$, etc. \qed
  \end{enumerate}
\end{proof}
\thmtransformationpreservesbehavior*

\begin{proof}
   For any $j \in \IN$ and any $\startvec{x} \in \IZ^{r}$ we obtain the following.
  \[\begin{array}{rl@{\hspace*{-.1cm}}l}
  &\IP_{\startvec{x}}^{\PP}(T^{\PP} = j) \hspace*{7cm}\\
  =& \IP_{\startvec{x}}^{\PP}(T^{\PP^\unisup} \circ
  \uni_{\PP}^{\omega}=j) & \text{by \cref{preparation} (a)}\\ 
                                                               \end{array}\]
  \pagebreak
  \[\begin{array}{rl@{\hspace*{-.1cm}}l}
      =& \IP_{\startvec{x}}^{\PP}\left((\uni_{\PP}^{\omega})^{-1}\left((T^{\PP^\unisup})^{-1}\left(\{j\}\right)\right)\right)\\
                                  =&
                                  \IP_{\startvec{x}}^{\PP}\left((\uni_{\PP}^{\omega})^{-1}\left(\biguplus\limits_{y_0,\ldots,y_{j-1}\in
                                    \IZ_{>0}, y_j \in \IZ_{\leq
                                      0}}Cyl^\IZ(\run{y_0,\ldots,y_j})\right)\right)\\
                                  =& \IP_{\startvec{x}}^{\PP}\left(\biguplus\limits_{y_0,\ldots,y_{j-1}\in \IZ_{>0}, y_j \in \IZ_{\leq 0}}(\uni_{\PP}^{\omega})^{-1}\left(Cyl^\IZ(\run{y_0,\ldots,y_j})\right)\right)\\
                                  =& \sum\limits_{y_0,\ldots,y_{j-1}\in \IZ_{>0}, y_j \in
                                    \IZ_{\leq
                                      0}}\IP_{\startvec{x}}^{\PP}\left((\uni_{\PP}^{\omega})^{-1}\left(Cyl^\IZ(\run{y_0,\ldots,y_j}\right)\right)\\
                                  =& \sum\limits_{y_0,\ldots,y_{j-1}\in \IZ_{>0}, y_j \in
      \IZ_{\leq 0}}\IP^{\PP^\unisup}_{\uni_{\PP}(\startvec{x})}\left(Cyl^\IZ(\run{y_0,\ldots,y_j})\right)
                                 & \text{by \cref{preparation} (b)}\\  
                                  =&
    \IP^{\PP^\unisup}_{\uni_{\PP}(\startvec{x})}\left(\biguplus\limits_{y_0,\ldots,y_{j-1}\in \IZ_{>0}, y_j
      \in \IZ_{\leq 0}}Cyl^\IZ(\run{y_0,\ldots,y_j})\right)\\
       =&
 \IP^{\PP^\unisup}_{\uni_{\PP}(\startvec{x})}\left(T^{\PP^\unisup} = j\right). 
  \end{array}\]
  As the above equality holds for every $j\in \IN$ it also holds for $j = \infty$. \qed
\end{proof}

\subsection{Proofs for \cref{sec:decidability_of_termination}}\label{app:decidability_of_termination}
\noindent{}For the proof of \cref{termination_decidable}, we use results on random walks
\cite{probabilityGrimmett,randomWalkSpitzer,feller50}.
We first recapitulate the required notions from probability
theory.

Consider a probability space $(\Omega, \F{F}, \IP)$ (i.e., for every $A \in \F{F}
\subseteq 2^\Omega$, $\IP(A)$ is the probability
that an event from the set $\Omega$ is in the subset $A$) and a  stochastic process
 $\mathbf{Y}=(Y_j)_{j \in \IN}$ where each $Y_j:\Omega \to \IZ$
is a random variable. $\mathbf{Y}$ is \emph{independent and identically distributed}
\emph{(i.i.d.)} on 
$(\Omega, \F{F}, \IP)$ if for all
$j,j' \in \IN$ with $j \neq j'$ and all   $y,z \in \IZ$:
  \begin{itemize}
  \item[$\bullet$] $Y_j$ and $Y_{\dashed{j}}$ are identically distributed, i.e.,
    $\IP(Y_{j} = z) = \IP(Y_{j'} = z)$ 
  \item[$\bullet$] $Y_j$ and $Y_{\dashed{j}}$ are independent, i.e., 
    $\IP(Y_{j} = y, Y_{j'} = z) = \IP(Y_{j} = y)\cdot\IP(Y_{j'} = z)$
     \end{itemize}
  Here, $\IP(Y_{j} = y, Y_{j'} = z) = \IP(Y_j^{-1}(\{y\}) \cap
  Y_{j'}^{-1}(\{z\}))$ is the probability that an event $\pi \in \Omega$ satisfies
  both $Y_j(\pi) = y$ and $Y_{j'}(\pi) = z$.
  So independence means that 
  one random variable does not influence the value
of the other.

Now we recapitulate the notion of a random walk created by an i.i.d.\ stochastic process.
\begin{definition}[Random Walk\protect{%
    \cite{probabilityGrimmett}}]\label{def_random_walk}
  Let $\mathbf{Y}=(Y_j)_{j \in \IN}$ be an i.i.d.\ stochastic process  for a probability space $(\Omega, \F{F}, \IP)$ 
  with $Y_j: \Omega \to \IZ$
  and let $X_0: \Omega \to \IZ$ be a random variable
   such that $\IP(X_0=\start{x})=1$ for some $\start{x}\in \IZ$.
  The (one-dimensional)
  \emph{random walk for} $(\Omega, \F{F}, \IP)$ induced by $\mathbf{Y}$ with
  starting point $X_0$ is the sequence $\mathbf{S} = (S_j)_{j \in \IN}$
  of random variables\footnote{Note that we define $S_j = X_0 + \sum\nolimits_{0 \leq u \leq
      j-1} \;
    Y_u$ instead of $S_j = \start{x} + \sum \nolimits_{0 \leq u \leq
      j-1} \;
    Y_u$. In this way, the random variables $X_0, Y_0, Y_1, \ldots$ 
only generate a single random walk that does not depend on $\start{x}$. Instead, the different
possible initial
values $\start{x}$ are taken care of by choosing different probability spaces $(\Omega, \F{F},
\IP_{\start{x}})$ where  $\IP_{\start{x}}(X_0=\start{x})=1$.}
  $S_j = X_0 + \sum\nolimits_{0 \leq u \leq j-1} \; Y_u$.
  We denote the random walk $\mathbf{S}$ by $(X_0,\mathbf{Y})$.
\end{definition}

\noindent{}Analogous to the termination time for programs from \cref{def:termination_time},
the \emph{hitting time} is the time when the random walk ``hits'' a
certain subset of $\IZ$ for the first time.

\begin{definition}[Hitting Time]\label{def_hitting}
  The \emph{hitting time} 
  for a random walk
  $(S_j)_{j \in \IN}$
  is the random variable $T^{hit}:\Omega \to
  \overline{\IN}$  with
  $T^{hit}(\pi)=\inf\{j \in \IN\mid S_j(\pi) \leq 0\}$.
\end{definition}

\noindent{}If $\mathbf{Y}=(Y_j)_{j \in \IN}$  
is i.i.d., then  $\IE(Y_0) = \IE(Y_j)$ for all $j \in \IN$. Hence, we define
$\mu= \IE(Y_0)$ to be the \emph{drift}, i.e., the expected change in each
step of the random walk. For such random walks, a result similar to \cref{termination_decidable} is
already known.

\begin{restatable}[Drift and Hitting Time \protect{\cite[Thm.~17.1, Prop.~18.1]{randomWalkSpitzer}}]{lemma}{lemmadrifthittingtime}\label{drift_hitting_random_walk}
  Let $\mathbf{Y}$ be i.i.d.\ for a probability space $(\Omega, \F{F}, \IP)$ and let
  $(X_0,\mathbf{Y})$ be a random walk for $(\Omega, \F{F}, \IP)$ such that $\mu=\IE(Y_0)<\infty$ (note that the drift
  $\mu$ does not depend on $X_0$). Let $T^{hit}$ be the hitting time for
  $(X_0,\mathbf{Y})$. Then we have:
  \begin{itemize}
    \item[$\bullet$] If $\mu>0$, then $\IP(T^{hit} = \infty)>0$.
    \item[$\bullet$] If $\mu = 0$ and $\IP(Y_0 = 0) \neq 1$, then $\IP(T^{hit} = \infty)=0$ but $\IE(T^{hit}) = \infty$.
    \item[$\bullet$] If $\mu < 0$, then $\IE(T^{hit}) < \infty$.
  \end{itemize}
\end{restatable}

\noindent{}In order to use \cref{drift_hitting_random_walk} to prove \cref{termination_decidable},
our aim is to represent
the stochastic process $\mathbf{X}^{\IZ}$
from \cref{stochastic_process} (for $r = 1$)
as a random walk $\mathbf{X}^{\IZ} = (X_0^\IZ, \mathbf{Y}^\IZ)$ for a suitable stochastic
process $\mathbf{Y}^\IZ$.

To this end, 
we take the stochastic process $\mathbf{Y}^\IZ = (Y_j^\IZ)_{j \in \IN}$ with $Y_j^\IZ=(X_{j+1}^\IZ
-X_{j}^\IZ)$ for all $j \in \IN$, i.e., $Y_j^\IZ$ is the change of the program variable in
the $(j+1)$-th loop iteration.
Then $\mathbf{X}^\IZ$ can be obtained as the random walk
$(X_0^\IZ, \mathbf{Y}^\IZ)$, since   $\IP_{\start{x}}^\PP(X_0^\IZ=\start{x})=1$ and
 $X_j^\IZ = X_0^\IZ + \sum\nolimits_{0 \leq u \leq j-1}
(X_{u+1}^\IZ -X_{u}^\IZ)
=  X_0^\IZ + \sum\nolimits_{0 \leq u \leq j-1} Y_u^\IZ$ for all $j \in \IN$.

Unfortunately,  $\mathbf{Y}^\IZ$ is not i.i.d.\
for
the probability measure $\IP_{\start{x}}^\PP$, because the probability that $Y_j^\IZ = 0$
(i.e., that $X_{j+1}^\IZ = X_j^\IZ$ holds)
depends on $j$. More precisely, the probability for
 $X_{j+1}^\IZ = X_j^\IZ$ is $p_0$ plus
the probability that the program has already reached a value $x \leq 0$ (i.e., that the
program's termination time is at most $j$). The reason is that according to the
probability measure $\IP_{\start{x}}^\PP$,
the value of $x$ remains
unchanged as soon as $x \leq 0$. Thus, we obtain
$\IP_{\start{x}}^\PP(Y_j^\IZ =0)=p_0 +
\IP_{\start{x}}^\PP(T^{\PP} \leq j)$, where $\IP_{\start{x}}^\PP(T^{\PP} \leq j)$ clearly depends on $j$.

\noindent{}Therefore, we now introduce a new adapted probability measure $\BP_{\start{x}}^\PP$
such that
 $\mathbf{Y}^\IZ$ is i.i.d.\ on the
probability space $(\IZ^{\omega}, \F{F}^{\IZ}, \BP_{\start{x}}^\PP)$
and at the same time,
$\BE_{\start{x}}^\PP(T^{\PP})\linebreak =
\expecP{\start{x}}{T^{\PP}}$, where $\BE_{\start{x}}^\PP(T^{\PP})$ denotes the
expected value of
the termination time $T^{\PP}$ under the probability measure $\BP_{\start{x}}^\PP$.
In the following definition, $q_{\start{x}}^\PP$ corresponds to the function $pr_{\start{x}}^\PP$ from
\cref{Run of a Program} that maps any prefix run to its probability if $\start{x}$ is the initial value of
the program variable. When defining $pr_{\start{x}}^\PP$,
the probability for a prefix run $\run{z_0,\ldots,z_{j-1},z_j}$ where $z_{j-1} \leq 0$
and $z_{j-1} \neq z_j$ was 0. In contrast, for $q_{\start{x}}^\PP$  we
continue to execute the program also if $x \leq 0$.
This corresponds to a variant of the program where the loop
condition $x>0$ is replaced by \emph{true}.

\begin{definition}[Probability Measure $\BP_{\start{x}}^\PP$]\label{new_measures}
  For any
random walk program
 $\PP\!$  as in   Def.\ \hyperlink{def-univariate-program}{\arabic{def-univariate-program}} without direct termination,
 any $\start{x} \in \IZ$,
 and any prefix run $\run{z_0,z_1,\ldots, z_j}$,
let $q_{\start{x}}^\PP(\run{z_0}) = \delta_{\start{x},z_0}$
 and if $j \geq
 1$, we define:
  \[\begin{array}{rcl}
q_{\start{x}}^\PP(\run{z_0,\ldots,z_j}) &=&q_{\start{x}}^\PP(\run{z_0,\ldots,z_{j-1}}) \cdot p_{z_j - z_{j-1}}
\end{array}\]
 $\BP_{\start{x}}^\PP$ is the 
  probability measure with
  $\BP_{\start{x}}^\PP(Cyl^\IZ(\pi))\!=\!q_{\start{x}}^\PP(\pi)$ for any prefix run $\pi$.
\end{definition}

\begin{example}[Adapted Probability Measure for  $\PP^\unisup_{race}$]
 {\sl Consider runs that start with $1$, $-2$, and $-6$. Here, we have $Y^\IZ_0(\run{1, -2,
 -6,\ldots}) = (-2) \, - \, 1 = -3$ and  $Y^\IZ_1(\run{1, -2,
 -6,\ldots}) = (-6) \, -  \,(-2) = -4$. For $\PP^\unisup_{race}$ of
  \cref{exmpl:tortoise_and_hare_reduced}, when
 using the probability measure $\IP^{\PP^\unisup_{race}}_1$ from \cref{probability_measures}, we obtain
 $\IP^{\PP^\unisup_{race}}_1(Cyl^\IZ(\run{1, -2,
 -6})) = pr^{\PP^\unisup_{race}}_1(\run{1, -2,
 -6}) = p^\unisup_{-3} \cdot p^\unisup_{-4} \cdot \delta_{-2, -6} = 0$, since the value of $x$
should not change anymore after reaching the non-positive value $-2$.
 In contrast, the adapted probability measure $\BP^{\PP^\unisup_{race}}_1$ from \cref{new_measures} yields
 $\BP^{\PP^\unisup_{race}}_1(Cyl^\IZ(\langle 1,-2,$
$-6\rangle)) = q^{\PP^\unisup_{race}}_1(Cyl^\IZ(\run{1, -2,
 -6})) = p^\unisup_{-3} \cdot p^\unisup_{-4}  = \tfrac{1}{22} \cdot \tfrac{1}{22} =
 \tfrac{1}{484}$.}
 \end{example}

\noindent{}For the termination time $T^{\PP}$ one only regards the time that it takes until the
program variable $x$ is
non-positive for the first time. Thus, it does not matter whether $x$ is kept
unchanged afterwards (as in the probability measure $\IP_{\start{x}}^\PP$) or whether the
loop body is
executed further afterwards (as in $\BP_{\start{x}}^\PP$). So the expected
runtime is the same, no matter whether one uses $\IE_{\start{x}}^\PP$ or $\BE_{\start{x}}^\PP$.

\begin{restatable}[$T^\PP$ is Identically Distributed Under $\IP_{\start{x}}^\PP$ and $\BP_{\start{x}}^\PP$]{lemma}{lemmertunaffected}\label{ert_unaffected}
  For any
  random walk program $\PP$ without direct termination,
  any  $\start{x} \in \IZ$, and any $j \in \overline{\IN}$, we have
    $\IP_{\start{x}}^\PP\left(T^{\PP} = j\right) = \BP_{\start{x}}^\PP\left(T^{\PP} = j\right)$.
Thus, $\expecP{\start{x}}{T^{\PP}} = \BE_{\start{x}}^\PP\left(T^{\PP}\right)$.
\end{restatable}
\begin{proof}
  First of all, by the definition of $T^{\PP}$, for any $j \in \IN$ we have
  \begin{equation}
    \label{T inverse}
  (T^{\PP})^{-1}(\{j\})=\biguplus \limits_{\pi = \run{z_0, \ldots, z_j} \in \IZ_{>0}^j\times
      \IZ_{\leq 0}} \Cyl^\IZ(\pi).
  \end{equation}
  
\noindent{}First, we consider $\start{x} \leq 0$. Then any cylinder set with positive probability
  w.r.t.~$\IP_{\start{x}}^\PP$ resp. $\BP^\PP_{\start{x}}$ has the form $\Cyl^\IZ(\pi)$ where $\pi$ starts with
  $\start{x} \leq 0$. But for any run $\tau \in \Cyl^\IZ(\pi)$ we have $T^{\PP}(\tau)=0$. Therefore, we
  conclude $\IP_{\start{x}}^\PP(T^{\PP}=0)=1=\BP_{\start{x}}^\PP(T^{\PP}=0)$.

  We now show that for $\start{x}>0$
  \begin{equation}
    \label{BP equals IP}
    \IP_{\start{x}}^\PP\left(\Cyl^\IZ(\pi)\right) = \BP^\PP_{\start{x}}\left(\Cyl^\IZ(\pi)\right) \;\;
    \text{for any $\pi = \run{z_0, \ldots, z_j} \in \IZ_{>0}^j\times \IZ_{\leq 0}$.}
    \end{equation}  
\noindent{}The reason is that we have:
  \begin{align*}
    \IP^\PP_{\start{x}}(\Cyl^\IZ(\pi))   &= pr^\PP_{\start{x}}(\pi) &\text{by \cref{probability_measures}}\\
  &= \delta_{\start{x},z_0} \cdot \prod\limits_{0 \leq u \leq j-1} p_{z_{u+1}-z_u} &\text{by \cref{Run of a Program} as $z_0,\dots, z_{j-1} >
      0$}\\
    &= q^\PP_{\start{x}}(\pi) \\
      &=\BP^\PP_{\start{x}}(\Cyl^\IZ(\pi))
    &\text{by \cref{new_measures}}
     \end{align*}
\noindent{}Therefore,  for all $j \in \IN$ we obtain:
\[\begin{array}{llr}
  &\IP^\PP_{\start{x}}\left(T^{\PP} = j\right)&\\ =&   \IP^\PP_{\start{x}}( (T^{\PP})^{-1}(\{j\}))&\\
   =&
   \IP^\PP_{\start{x}}\left(\biguplus \limits_{\pi = \run{z_0, \ldots, z_j}
  \in \IZ_{>0}^j\times \IZ_{\leq 0}} \Cyl^\IZ(\pi) \right) &\text{by (\ref{T inverse})}\\
=& \sum \limits_{\pi = \run{z_0, \ldots, z_j} \in \IZ_{>0}^j\times \IZ_{\leq
    0}}\IP^\PP_{\start{x}}\left(\Cyl^\IZ(\pi) \right) &\text{by additivity of prob.\ measures}\\
=& \sum \limits_{\pi = \run{z_0, \ldots, z_j} \in \IZ_{>0}^j\times \IZ_{\leq
    0}}\BP^\PP_{\start{x}}\left(\Cyl^\IZ(\pi) \right) &\text{by \cref{BP equals IP}}\\
 =& \BP^\PP_{\start{x}}\left(\biguplus \limits_{\pi = \run{z_0, \ldots, z_j}
  \in \IZ_{>0}^j\times \IZ_{\leq 0}} \Cyl^\IZ(\pi) \right) &\text{by additivity of prob.\
  measures}\\
=& \BP^\PP_{\start{x}}\left(T^{\PP} = j\right) &\text{by (\ref{T inverse})}
  \end{array}\]
  
\noindent{}Finally, $\IP^\PP_{\start{x}}\left(T^{\PP} = \infty\right)=1-\sum\limits_{j \in
    \IN}\IP^\PP_{\start{x}}\left(T^{\PP} = j\right)=1-\sum\limits_{j \in \IN}\BP^\PP_{\start{x}}\left(T^{\PP}=j
  \right) =\linebreak \BP^\PP_{\start{x}}\left(T^{\PP} = \infty\right)$.
    \qed 
\end{proof}

\noindent{}Now we show that the process $\mathbf{Y}^\IZ$ with $Y_j^\IZ = X_{j+1}^\IZ - X_j^\IZ$ is
i.i.d.\ w.r.t.\ the probability measure $\BP^\PP_{\start{x}}$
and thus,  $(X_0^\IZ,\mathbf{Y}^\IZ)$ is a random walk for $(\IZ^{\omega}, \F{F}^\IZ, \BP^\PP_{\start{x}})$.
So the expected value of $Y_j^\IZ$ 
under $\BP^\PP_{\start{x}}$,
is the
same for all $j$. In fact, this expected value is the drift $\mu_\PP$ of the
program, irrespective of the start  value $\start{x}$. 

\begin{restatable}[Y is i.i.d.\ and its Expected Value is the Drift of the Program]{lemma}{lemmaiid}\label{i_i_d}
  Let $\mathbf{X}^{\IZ}$ be the stochastic
  process as in \cref{stochastic_process}.
  We define the process $\mathbf{Y}^{\IZ} = (Y_j^{\IZ})_{j \in \IN}$ by $Y_j^{\IZ} = X_{j+1}^{\IZ}-X_j^{\IZ}$ for all $j
  \in \IN$.  Then for any random walk program $\PP$ without direct termination and any
  $\start{x} \in \IZ$, $\mathbf{Y}^{\IZ}$ is
  i.i.d.\ w.r.t.~$(\IZ^{\omega}, \F{F}^\IZ, \BP^\PP_{\start{x}})$ and thus,
  $(X_0^{\IZ},\mathbf{Y}^{\IZ})$ is a random walk for this probability space.
  Furthermore, for any $\start{x} \in \IZ$ and any $j \in \IN$, we
  have $\BE^\PP_{\start{x}}(Y_j^\IZ)=\mu_\PP$.
\end{restatable}
\begin{proof}
\noindent{}We first show that the $Y_j^{\IZ}$ are identically distributed. More precisely, we
  prove that for all $u, \start{x} \in \IZ$ and all $j \in \IN$ we have
 $\BP^\PP_{\start{x}}(Y_j^{\IZ} = u)=p_u$. Similar to our handling of multivariate programs in
  \cref{MDP},
  for any random walk program $\PP$ as in
  Def.\ \hyperlink{def-univariate-program}{\arabic{def-univariate-program}}
  we
  define $p_v = 0$ for $v > m$ or $v < -k$. 
  \[\begin{array}{llr}
  &\BP^\PP_{\start{x}}(Y_j^{\IZ} = u)\hspace*{11cm}\\
  =& \BP^\PP_{\start{x}}(X_{j+1}^{\IZ}-X_j^{\IZ} = u)\\
                      =& \BP^\PP_{\start{x}}\left(\{\run{z_0, \ldots} \in \IZ^{\omega} \mid z_{j+1}-z_j=u\}\right)\\              
                             =& \BP^\PP_{\start{x}}\left(\biguplus \limits_{\pi=\run{ z_0, \ldots,
                                 z_{j+1} } \in \IZ^{j+2}, z_{j+1}-z_j = u} \Cyl^\IZ(\pi)\right)
                                             \end{array}\]
\pagebreak
\[\begin{array}{ll@{\hspace*{-2.1cm}}r}
     =& \sum \limits_{\pi= \run{z_0, \ldots z_{j+1}} \in \IZ^{j+2},
      z_{j+1}-z_j = u}\BP^\PP_{\start{x}}\left(\Cyl^\IZ(\pi)\right) & \text{by additivity of
                           prob.\ measures}\\
                         =& \sum \limits_{\pi=\run{z_0, \ldots, z_{j+1}} \in \IZ^{j+2},
                           z_{j+1}-z_j = u} q_{\start{x}}^\PP(\pi) & \text{by \cref{new_measures}}\\
                         =& \sum \limits_{\pi=\run{z_0,z_1, \ldots, z_{j+1}}  \in \IZ^{j+2}, z_0 = \start{x},
                        z_{j+1}-z_j = u} \quad \prod\limits_{0 \leq v \leq j} p_{z_{v+1}-z_v} & \text{by \cref{new_measures}}
                      	\\
                      =& \sum \limits_{\pi=\run{z_0, z_1, \ldots, z_{j}} \in
      \IZ^{j+1}, z_0 = \start{x}} \; \left( 
                      p_u \cdot \prod\limits_{0 \leq v \leq j-1} p_{z_{v+1}-z_v} \right) \\
                      =& \left(\sum \limits_{\pi= \run{z_0,z_1, \ldots, z_{j}} \in \IZ^{j+1},
                        z_0 = \start{x}} \quad
    \prod\limits_{0 \leq v \leq j-1} p_{z_{v+1}-z_v} \right)\cdot p_u \\
    =& \left(\sum \limits_{v_1,\ldots,v_j \in \IZ} p_{v_1} \cdot \ldots \cdot p_{v_j}\right)\cdot p_u \\
                      =& \left(\sum \limits_{v \in \IZ} p_v\right)^j \cdot p_u\\
                      =& 1^j \cdot p_u = p_u. &\text{as
                        $\PP$ does not  have direct termination}
  \end{array}\]
\noindent{}As $p_u$ is independent of $j$, the $Y_j^\IZ$ are identically distributed. Furthermore,  the expected value of $Y_j^\IZ$ under
$\BP^\PP_{\start{x}}$ is
   \[\BE_{\start{x}}^\PP(Y_j^\IZ) = \sum \limits_{-k \leq u \leq m} u \cdot p_u = \mu_\PP,\]
  which is the drift of the program.

  It remains to show the independence of the random variables. Let $j \neq j' \in \IN$ and w.l.o.g.~assume $\dashed{j} > j$.
  \[\begin{array}{ll}
  &\BP^\PP_{\start{x}}(Y_{j}^\IZ= u, Y_{j'}^\IZ = u') \hspace*{8cm}\\
  =& \BP^\PP_{\start{x}}(X_{j+1}^\IZ-X_{j}^\IZ = u, X_{j'+1}^\IZ-X_{j'}^\IZ = u')\\
                      =& \BP^\PP_{\start{x}}\{\run{z_0, \ldots} \in \IZ^{\omega} \mid z_{j+1}-z_{j}=u, z_{j'+1}-z_{j'} = u'\}\\
                      =& \BP^\PP_{\start{x}}\left(\biguplus \limits_{\pi=\run{z_0, \ldots,
                        z_{\dashed{j}+1}} \in \IZ^{\dashed{j}+2}, z_{j+1}-z_{j} = u, z_{j'+1}-z_{j'} = u'} \Cyl^\IZ(\pi)\right)\\
                      =& \sum \limits_{\pi=\run{z_0, \ldots, z_{\dashed{j}+1}}
                        \in \IZ^{\dashed{j}+2}, z_{j+1}-z_{j} = u, z_{j'+1}-z_{j'} =
                        u'}\BP^\PP_{\start{x}}\left(\Cyl^\IZ(\pi)\right)\\ 
                      =& \sum \limits_{\pi= \run{z_0, \ldots, z_{j'+1}} \in
                        \IZ^{\dashed{j}+2}, z_{j+1}-z_{j} = u, z_{j'+1}-z_{j'} = u'}
                      q_{\start{x}}^\PP(\pi)\\
                        =& \sum \limits_{\pi=\run{z_0,z_1, \ldots z_{\dashed{j}+1}} \in
                        \IZ^{\dashed{j}+2}, z_0 = \start{x}, z_{j+1}-z_{j} =
                       u, z_{j'+1}-z_{j'} = u'} \quad \prod\limits_{0 \leq v \leq j'} p_{z_{v+1}-z_v} 	\\
  \end{array}\]
  \pagebreak
  \[\begin{array}{ll}
  =& \sum \limits_{\pi= \run{z_0,z_1, \ldots z_{\dashed{j}+1}} \in
                        \IZ^{\dashed{j}+2},\atop {z_0 = \start{x}, z_{j+1}-z_{j} = u,\atop z_{j'+1}-z_{j'} = u'}}
                                           \hspace*{-.5cm}
  \mbox{\footnotesize $p_u \cdot \left(\prod\limits_{0 \leq v \leq j-1} \hspace*{-.1cm} p_{z_{v+1}-z_v}\right)  \cdot p_{u'} \cdot
    \left(\prod\limits_{j+1 \leq v \leq \dashed{j}-1} \hspace*{-.1cm}
    p_{z_{v+1}-z_v}\right)$}\\
  =& \left(\sum \limits_{\pi=\run{z_0,z_1, \ldots z_{\dashed{j}}} \in
      \IZ^{\dashed{j}+1}, z_0 = \start{x}, z_{j+1} - z_j = u}\quad \prod\limits_{0 \leq v \leq
                        \dashed{j}-1, \, v \neq
                        j} p_{z_{v+1}-z_v}\right) \cdot
    p_{u} \cdot p_{u'}\\
  =& \left(\sum \limits_{v_1,\ldots,v_{j'-1} \in \IZ} p_{v_1} \cdot \ldots \cdot p_{v_{j'-1}}\right)\cdot  p_{u} \cdot p_{u'} \\   
                      =& \left(\sum \limits_{v \in \IZ} p_v\right)^{\dashed{j}-1} \cdot p_{u} \cdot p_{u'}\\
                      =& 1^{\dashed{j}-1} \cdot p_{u} \cdot p_{u'} = p_{u} \cdot p_{u'}\\
                      =&  \BP^\PP_{\start{x}}(Y_{j}^\IZ = u) \cdot \BP^\PP_{\start{x}}(Y_{j'}^\IZ = u'). \hspace*{6.3cm} \qed \hspace*{-6.3cm}
  \end{array}\]
\end{proof}

\noindent
Now we can prove  \cref{termination_decidable} based on the results of \cref{drift_hitting_random_walk} for random walks. 

\thmdecidability*
\begin{proof}
    Due to \cref{i_i_d}, $\mathbf{Y}^\IZ$ is i.i.d.\ w.r.t.\ $(\IZ^{\omega}, \F{F}^\IZ,
    \BP^\PP_{\start{x}})$
    and thus,
    $\mathbf{S}^\IZ = (X_0^\IZ, \mathbf{Y}^\IZ)$ is a random walk w.r.t.\ this probability
    space for any $\start{x} \in
    \IZ$. By \cref{def_random_walk} we have
    $S_j^\IZ = X_0^\IZ + \sum \nolimits_{0 \leq u \leq j-1}Y_u^\IZ = X_j^\IZ$ for any $j \in \IN$. Hence, the
    hitting time $T^{hit}$ for the random walk $\mathbf{S}^\IZ$  as defined in \cref{def_hitting} is exactly the
    termination time $T^{\PP}$. As we proved in \cref{i_i_d} that
    $\BE_{\start{x}}^\PP(Y_0)=\mu_\PP$ holds independent of $\start{x} \in \IZ$,  we can use Lemma
    \ref{drift_hitting_random_walk} for all $\start{x}$. So we get for all $\start{x} \in \IZ$:
    \begin{itemize}
    \item[$\bullet$] If $\mu_\PP\!>0$, then
      $\BP^\PP_{\start{x}}(T^{\PP}\!\!=\!\infty)\overset{\mathrm{Lemma}\;\ref{ert_unaffected}}{=}\IP^\PP_{\start{x}}(T^{\PP}\!\!=\!\infty)>0$,
      i.e., $\PP$ is not AST.  
\item[$\bullet$] 
  Note that as $\PP$ is non-trivial (i.e.,
      $p_0 \neq 1$), we have $\BP^\PP_{\start{x}}(Y_0^\IZ = 0) \neq
      1$.
      So if  $\mu_\PP=0$, then
     \cref{drift_hitting_random_walk} implies 
      $\BP^\PP_{\start{x}}(T^{\PP}=\infty)\overset{\mathrm{Lemma}\;\ref{ert_unaffected}}{=}\IP^\PP_{\start{x}}(T^{\PP}=\infty)=0$
      but $\BE^\PP_{\start{x}}(T^{\PP})
      \overset{\mathrm{Lemma}\;\ref{ert_unaffected}}{=}\IE^\PP_{\start{x}}(T^{\PP})=\infty$,
      i.e., $\PP$ is AST but not PAST.  
      \item[$\bullet$]  If $\mu_\PP<0$, then
        $\BE^\PP_{\start{x}}(T^{\PP})\overset{\mathrm{Lemma}\;\ref{ert_unaffected}}{=}\IE^\PP_{\start{x}}(T^{\PP})<\infty$,
        i.e., $\PP$ is PAST. \hfill \qed
    \end{itemize}
\end{proof}

\begin{example}[Termination of Variations of $\PP^\unisup_{race}$]\label{tortoise_and_hare_decidability}
{\sl
  \noindent{}We showed already in Sect.\ \ref{sec:decidability_of_termination} that the drift of the program $\PP^\unisup_{race}$ in \cref{exmpl:tortoise_and_hare_reduced}
  is
  $-\tfrac{3}{2}<0$.
  So by
  \cref{termination_decidable} this program is PAST,  i.e.,  the hare is expected
    to overtake the
  tortoise in a finite number of iterations.

  Now consider the modified program $\PP$:
   \begin{tabbing}
     \hspace*{4.5cm}\= \hspace*{.2cm}\=\kill
    \>\mbox{\rm \texttt{while}} $(x > 0) \; \{$\\
    \>\>$x=x+1$ \hspace*{.2cm}\=$[\tfrac{6}{11}];$\\[0.1cm]
    \>\>$x=x$\>$[\tfrac{3}{11}];$\\[0.1cm]
    \>\>$x = x-1$\>$[0];$\\[0.1cm]
    \>\>$x=x-2$\>$[\tfrac{1}{11}];$\\[0.1cm]
    \>\>$x = x-3$\>$[0];$\\[0.1cm]
    \>\>$x = x-4$\>$[\tfrac{1}{11}];$\\
    \>$\}$
  \end{tabbing} 
  The distance still increases with probability $\tfrac{6}{11}$ but it decreases by at most $4$.
  Its drift is $\mu_{\PP}=1 \cdot \tfrac{6}{11} +
0 \cdot \tfrac{3}{11} +
(-2)\cdot \tfrac{1}{11} + (-4)\cdot \tfrac{1}{11} =0$. Hence, on average
the distance $x$ between the tortoise and the hare remains unchanged after each loop
iteration. By \cref{termination_decidable} this program is AST but not PAST.
Hence, the hare wins with probability $1$, but the expected number of required loop iterations
is infinite.

Finally, we change the probabilities to obtain the program $\PP'$:
  
  \begin{tabbing}
  \hspace*{4.5cm}\= \hspace*{.2cm}\=\kill
  \>\mbox{\rm \texttt{while}} $(x > 0) \; \{$\\
  \>\>$x=x+1$ \hspace*{.2cm}\=$[\tfrac{6}{11}];$\\[0.1cm]
  \>\>$x=x$\>$[\tfrac{3}{11}];$\\[0.1cm]
  \>\>$x = x-1$\>$[\tfrac{1}{22}];$\\[0.1cm]
  \>\>$x=x-2$\>$[\tfrac{1}{22}];$\\[0.1cm]
  \>\>$x = x-3$\>$[\tfrac{1}{22}];$\\[0.1cm]
  \>\>$x = x-4$\>$[\tfrac{1}{22}];$\\
  \>$\}$
\end{tabbing} 

  \noindent{}Its drift is $\mu_{\PP'} =1\cdot \tfrac{6}{11}+0\cdot \tfrac{3}{11}+
  \tfrac{1}{22} \cdot\sum\nolimits_{-4 \leq j \leq -1} j =\tfrac{1}{11}>0$.
  Thus, $\PP'$ is not AST by \cref{termination_decidable}. So
  there is a positive probability that the hare never catches up with the tortoise
and the race takes forever.}
\end{example}

\corohelp*
\begin{proof}
If $\PP$ has direct termination (i.e., $p' \neq 0$), then $\PP$ and $\PP^\unisup$ are PAST by
\cref{PAST}.
Otherwise, by
  \cref{transformation_preserves_behavior} we can reduce the
  termination of $\PP$ to the termi\-nation  of $\PP^\unisup$ on inputs which are in the
  image of $\uni_\PP$. Note 
 that the termination behavior of $\PP^\unisup$ is the same for all $x > 0$.
Hence, to show that $\PP$ is (P)AST iff $\PP^\unisup$ (P)AST, we
prove that  $\uni_\PP$'s image  also includes positive values.
To see this,
  note that
  $\vec{a}\neq \vec{0}$ implies $\vec{a}\bullet\vec{a} > 0$.
Hence, for any natural number $u > \tfrac{b}{\vec{a}\bullet\vec{a}}$ 
we obtain $\uni_{\PP}(u \cdot
\vec{a})= u \cdot \vec{a}\bullet \vec{a} - b
> \tfrac{b}{\vec{a}\bullet\vec{a}} \cdot  \vec{a}\bullet \vec{a} - b = 0$. 
\qed
\end{proof}

\subsection{Proofs for \cref{sec:computability_of_asymptotic_expected_runtimes}}\label{app:computability_of_asymptotic_expected_runtimes}
We now show that for CP programs $\PP$ without direct termination, one can not only
decide termination, but the construction for the proof of \cref{termination_decidable}
also directly yields
asymptotically exact bounds on
their expected runtime. More precisely, we show that $rt^\PP_{\start{x}}$
is asymptotically linear whenever
$\PP$ is PAST
(and we even provide  actual upper and
lower bounds).
To prove this result, we
use \emph{Wald's Lemma} from probability theory. Again, we first consider random walk
programs and then use the reduction of \cref{sec:Restriction_to_Random_Walk_Programs} to
lift 
our result to arbitrary CP programs.

Recall that if a stochastic process $\mathbf{Y}=(Y_j)_{j \in \IN}$  on a probability space $(\Omega, \F{F}, \IP)$
is i.i.d., then $\IE(Y_0) = \IE(Y_j)$ for all $j \in \IN$. Thus, we  obtain
\[\IE\left(\sum_{0 \leq j \leq c-1} Y_j\right) = \sum_{0 \leq j\leq c-1} \IE(Y_j)
= c \cdot \IE(Y_0) \quad \text{for any constant
$c \in \IN$.}\]
By Wald's Lemma, a similar statement
even holds if instead of the constant $c$ we use a random variable $T$, provided that $T$
is independent from the stochastic process $\mathbf{Y}$.
We use a consequence of Wald's Lemma where $T$ does not need to be independent from the
whole process $\mathbf{Y}$ but 
for every $j$,
the random variable $Y_j$ is independent of whether $T$ is greater or equal to $j+1$.
The required independence can be expressed formally by demanding that
$Y_j$ must be independent
of  $\II_{\{T \geq j+1\}}: \Omega \to \{0,1 \}$, where  $\II_{\{T \geq j+1\}}(\pi)=1$ if
$T(\pi) \geq j+1$ and
$\II_{\{T \geq j+1\}}(\pi) =0$ otherwise.
Then, 
to compute $\IE\left(\sum\nolimits_{0 \leq j \leq T-1} Y_j\right)$, by Wald's Lemma one
can apply $\IE$ to both $T$ and $Y_n$ separately, i.e.,
one can compute $\IE(T) \cdot \IE(Y_0)$.

\begin{restatable}[Consequence of Wald's Lemma, \protect{cf.\ \cite[Lemma
        10.2(9)]{probabilityGrimmett}}]{lemma}{lemmawald}\label{walds_lemma}
  \hspace*{-.27cm} Let $\mathbf{Y}=(Y_j)_{j \in \IN}$ be a stochastic process on a probability
  space $(\Omega, \F{F}, \IP)$ which is i.i.d.\
  and let $T: \Omega \to \overline{\IN}$ be a random
  variable.
 Define the random variable $(\sum \nolimits_{0 \leq j \leq T-1} \!Y_j)\!:\!\Omega\!\to\!\IR, \pi \mapsto \!\sum \nolimits_{0 \leq j \leq T(\pi)-1} \!Y_j(\pi)$.
  If $\IE(Y_0)\!<\!\infty$, $\IE(T)\!< \infty$, and
the random variables $Y_j$ and $\II_{\{T \geq j+1\}}$ are independent for all $j\!\in\!\IN$,
 then
  \[\IE\left(\sum \limits_{0 \leq j \leq T-1} Y_j\right)=\IE(T)\cdot \IE(Y_0). \]   
\end{restatable}
\begin{proof}
  In \cite[Thm.~17.7]{probabilityBauer} it is shown that
  \begin{equation}
    \label{Wald1} \IE\left(\sum \limits_{0 \leq j \leq  T-1} Y_j\right) < \infty,
    \end{equation}
  i.e., the expected value of $\sum\nolimits_{0 \leq j \leq  T-1} Y_j$ exists.
  The proof of \cref{walds_lemma} is similar to the proof
  of  
  \cite[Lemma (9) in Sect.\ 10.2]{probabilityGrimmett}, but
    it is done under different preconditions.
  \[
  \begin{array}{ll@{\hspace*{-4cm}}l}
    &\IE\left(\sum\limits_{0 \leq j \leq T-1}Y_j\right) &\hspace*{9cm}                                        
  \end{array}\]
\pagebreak  
   \[
   \begin{array}{ll@{\hspace*{-4.5cm}}l}   
     =& \IE\left(\sum\limits_{0 \leq j}Y_j\cdot \mathbb{I}_{\{T \geq j+1\}}\right)& \\
     =&
    \sum\limits_{0 \leq j}\IE\left(Y_j\cdot \mathbb{I}_{\{T \geq j+1\}}\right) &
    \text{by the existence \cref{Wald1}, i.e.,}\\[-.3cm]
    &&\text{since this expected value is $< \infty$}\\
              =&
    \sum\limits_{0 \leq j}\IE\left(Y_j\right) \cdot \IE\left(\mathbb{I}_{\{T \geq
      j+1\}}\right) & \text{by independence of $Y_j$ and $\mathbb{I}_{\{T \geq
        j+1\}}$}\\
     =&
    \sum\limits_{0 \leq j}\IE\left(Y_0\right) \cdot \IE\left(\mathbb{I}_{\{T \geq
      j+1\}}\right) & \text{as $\IE(Y_j) = \IE(Y_0)$, since $\mathbf{Y}$ is i.i.d.}\\
                                                =& \IE\left(Y_0\right) \cdot\sum\limits_{0
                                                  \leq j} \IE\left(\mathbb{I}_{\{T \geq
                                                  j+1\}}\right) & \\ 
                                                =& \IE\left(Y_0\right)
                                                \cdot\sum\limits_{0 \leq j}
\left( 0 \cdot \IP\left(\mathbb{I}_{\{T \geq  j+1\}} = 0\right) + 1 \cdot
\IP\left(\mathbb{I}_{\{T \geq  j+1\}} = 1\right)\right)\\
      =& \IE\left(Y_0\right)
                                                \cdot\sum\limits_{0 \leq j}
\IP\left(T \geq  j+1\right)\\
=& \IE\left(Y_0\right) \cdot \sum\limits_{0 \leq j}\; \sum\limits_{j+1 \leq u}
    \IP\left(T = u\right)\\
=& \IE\left(Y_0\right) \cdot \sum\limits_{1 \leq u}\;\sum\limits_{0 \leq j \leq u-1}
    \IP\left(T = u\right)\\
 =& \IE\left(Y_0\right) \cdot \sum\limits_{1 \leq u} u \cdot  \IP\left(T = u\right)\\
   =& \IE\left(Y_0\right) \cdot \IE\left(T\right)\hspace*{8.1cm} \qed \hspace*{-8.1cm}
  \end{array}\]
\end{proof}

\noindent{}In
 our setting, we consider the stochastic process $\mathbf{Y}^\IZ$ from \cref{i_i_d}
  and the termination time $T^{\PP}$. When
regarding $\IP^\PP_{\start{x}}$, $Y_j$ (i.e., the difference between the
$(j+1)$-th and the $j$-th element of a run)
is clearly not independent of the question whether 
the run already terminated in (or before) the $j$-th element.
The reason is that under the probability measure $\IP^\PP_{\start{x}}$, 
the elements of a run do not change anymore after
termination. However, \cref{independence_terminationtime_differences} shows that when
regarding $\BP^\PP_{\start{x}}$ instead,  the independence requirement
of \cref{walds_lemma} is fulfilled.

\begin{restatable}[Independence of $Y_j^\IZ$ and $\II_{\{T^{\PP} \geq j+1\}}$]{lemma}{lemmaindependence}\label{independence_terminationtime_differences}
  Let $\mathbf{Y}^\IZ = (Y_j^\IZ)_{j \in \IN}$ be the stochastic process from \cref{i_i_d}.
Then for any random walk program $\PP$ without direct termination, any
$\start{x}\in \IZ$,
and  any $j \in \IN$, the random variables $Y_j^\IZ$ and $\II_{\{T^{\PP} \geq j+1\}}$ are
independent w.r.t.\ the probability measure $\BP^\PP_{\start{x}}$. 
\end{restatable}
\begin{proof}
  We show that for any $x,y \in \IZ$, we have
  \[\BP^\PP_{\start{x}}\left(Y_j^\IZ =x,\; \II_{\{T^{\PP} \geq j+1\}} = y\right) \; = \;
  \BP^\PP_{\start{x}}\left(Y_j=x\right) \cdot \BP^\PP_{\start{x}}\left(\II_{\{T^{\PP} \geq j+1\}} = y\right).\]
Note that the left- and the right-hand side are both zero whenever $y \notin
\{0,1\}$. Thus,
it is enough to show the claim for $y=0$ and $y=1$.

\medskip

  \noindent \underline{Case 1: $y=0$}
  \[
  \begin{array}{ll@{\hspace*{-3.5cm}}r}
    &\BP^\PP_{\start{x}}\left(Y_j=x,\; \II_{\{T^{\PP} \geq j+1\}} = 0\right)&\hspace*{10cm}
    \end{array}\]
 \pagebreak
 \[\begin{array}{ll@{\hspace*{-3.5cm}}r}
  =&\BP^\PP_{\start{x}}\left(\biguplus\limits_{0 \leq u \leq j}\biguplus \limits_{\pi =
      \run{z_0, \ldots,
      z_{j+1}}
      \in \IZ_{>0}^u \times \IZ_{\leq 0} \times \IZ^{j-u+1},\atop z_{j+1}-z_j = x}
    \Cyl^\IZ(\pi)\right)&\\
                                                          =&\sum\limits_{0 \leq u \leq j}\sum
                                                          \limits_{\pi = \run{z_0, \ldots,
      z_{j+1}}
      \in \IZ_{>0}^u \times \IZ_{\leq 0} \times \IZ^{j-u+1},\atop z_{j+1}-z_j = x}
                                                          \BP^\PP_{\start{x}}\left(\Cyl^\IZ(\pi)\right)&\text{as
                                                            $\BP^\PP_{\start{x}}$ is a prob.\ measure}\\ 
                                                          =&\sum\limits_{0 \leq u \leq j}\sum
                                                          \limits_{\pi =  \run{z_0, \ldots,
      z_{j+1}}
      \in \IZ_{>0}^u \times \IZ_{\leq 0} \times \IZ^{j-u+1},\atop z_{j+1}-z_j = x}
                                                          q^\PP_{\start{x}}(\pi)&\text{by \cref{new_measures}}\\ 
                                                          =&\sum\limits_{0 \leq u \leq
                                                            j}\sum \limits_{\pi = \run{z_0, \ldots,
      z_{j+1}}
      \in \IZ_{>0}^u \times \IZ_{\leq 0} \times \IZ^{j-u+1},\atop z_{j+1}-z_j = x}
                                                          \delta_{\start{x},z_0}\cdot\prod\limits_{0
                                                            \leq v \leq j}
                                                          p_{z_{v+1}-z_v}&\text{by
                                                            \cref{new_measures}}\\
                                                          =&\left(\sum\limits_{0 \leq u
                                                            \leq j}\sum \limits_{\pi =
                                                            \run{z_0, \ldots,
      z_{j}}
      \in \IZ_{>0}^u \times \IZ_{\leq 0} \times \IZ^{j-u}}
                                                          \delta_{\start{x},z_0}\cdot\prod\limits_{0
                                                            \leq v \leq j-1}
                                                          p_{z_{v+1}-z_v}\right)\cdot
                                                          p_x\\                                                    =&\BP^\PP_{\start{x}}\left(Y_j=x\right)\cdot
                                                          \left(\sum\limits_{0 \leq u \leq
                                                            j}\sum
                                                          \limits_{\pi = \run{z_0, \ldots,
                                                            z_{j}} \in \IZ_{>0}^u \times
                                                            \IZ_{\leq 0} \times \IZ^{j-u}}
                                                          \delta_{\start{x},z_0}
                                                          \cdot\prod\limits_{0 \leq v \leq
                                                            j-1} p_{z_{v+1}-z_v}\right)\\
    =&\BP^\PP_{\start{x}}\left(Y_j=x\right)\cdot
                                                          \left(\sum\limits_{0 \leq u \leq
                                                            j}\sum
                                                          \limits_{\pi = \run{z_0, \ldots,
                                                            z_{j}} \in \IZ_{>0}^u \times
                                                            \IZ_{\leq 0} \times \IZ^{j-u}}
                                                          q^\PP_{\start{x}}(\pi)\right)\\
  =&\BP^\PP_{\start{x}}\left(Y_j=x\right)\cdot
                                                          \left(\sum\limits_{0 \leq u \leq
                                                            j}\sum
                                                          \limits_{\pi = \run{ z_0, \ldots,
                                                            z_{j}} \in \IZ_{>0}^u \times
                                                            \IZ_{\leq 0} \times \IZ^{j-u}} \hspace*{-1cm}
                                                        \BP^\PP_{\start{x}}\left(\Cyl^\IZ(\pi)\right)
                                                        \right)&\text{by
                                                            \cref{new_measures}}\\                                                          
                                                          =&\BP^\PP_{\start{x}}\left(Y_j=x\right)\cdot
                                                        \BP^\PP_{\start{x}}\left(\biguplus\limits_{0
                                                          \leq u \leq j}\biguplus
                                                        \limits_{\pi = \run{z_0, \ldots,
                                                            z_{j}} \in \IZ_{>0}^u \times
                                                            \IZ_{\leq 0} \times
                                                            \IZ^{j-u}} \hspace*{-1cm} \Cyl^\IZ(\pi)\right)&\text{prob.\ measure}\\ 
                                                          =&\BP^\PP_{\start{x}}\left(Y_j=x\right)\cdot \BP^\PP_{\start{x}}\left(\II_{\{T^{\PP} \geq j+1\}} = 0\right)
 \end{array}\]

 \bigskip
 
  \noindent \underline{Case 2: $y=1$}
 \[ \begin{array}{ll@{\hspace*{-5cm}}r}    
   &\BP^\PP_{\start{x}}\left(Y_j=x, \; \II_{\{T^{\PP} \geq j+1\}} = 1\right) &\hspace*{11cm}\\
   =&\BP^\PP_{\start{x}}\left(\biguplus \limits_{\pi = \run{z_0, \ldots, z_{j+1}} \in \IZ_{>0}^{j+1}\times\IZ,\atop z_{j+1}-z_j = x} \Cyl^\IZ(\pi)\right)&\\
                                                            =&\sum \limits_{\pi = \run{z_0,
                                                              \ldots, z_{j+1}} \in
                                                              \IZ_{>0}^{j+1}\times\IZ,\atop
                                                              z_{j+1}-z_j =
                                                              x}\BP^\PP_{\start{x}}\left(
                                                            \Cyl^\IZ(\pi)\right)&\text{as $\BP^\PP_{\start{x}}$ is a prob.\ measure}\\
                                                            =&\sum \limits_{\pi = \run{z_0,
                                                              \ldots, z_{j+1}} \in
                                                              \IZ_{>0}^{j+1}\times\IZ,\atop
                                                              z_{j+1}-z_j = x}q^\PP_{\start{x}}(\pi)&\text{by
                                                            \cref{new_measures}}  
                                                        \end{array}\]
\pagebreak 
\[ \begin{array}{ll@{\hspace*{-.4cm}}r}
  
                                                            =&\sum \limits_{\pi = \run{z_0,
                                                              \ldots, z_{j+1}} \in
                                                              \IZ_{>0}^{j+1}\times\IZ,\atop
                                                              z_{j+1}-z_j =
                                                              x}\delta_{\start{x},z_0}\cdot\prod\limits_{0
                                                              \leq v \leq j} p_{z_{v+1}-z_v}&\text{by   \cref{new_measures}}\\
       =&\left(\sum \limits_{\pi =
                                                              \run{z_0,
                                                              \ldots, z_{j}} \in
                                                              \IZ_{>0}^{j+1}}\delta_{\start{x},z_0}\cdot\prod\limits_{0
                                                              \leq v \leq j-1}
                                                            p_{z_{v+1}-z_v}\right) \cdot
                                                            p_x\\
                                                              =&\BP^\PP_{\start{x}}\left(Y_j=x\right)
                                                            \cdot \left(\sum \limits_{\pi
                                                              = \run{z_0, \ldots, z_{j}}
                                                               \in
                                                              \IZ_{>0}^{j+1}}\delta_{\start{x},z_0}\cdot\prod\limits_{v=0}^{j-1}
                                                            p_{z_{v+1}-z_v}\right)&\\
   =&\BP^\PP_{\start{x}}\left(Y_j=x\right)
                                                            \cdot \left(\sum \limits_{\pi
                                                              = \run{z_0, \ldots, z_{j}}
                                                              \in
                                                              \IZ_{>0}^{j+1}}
                                                            q^\PP_{\start{x}}(\pi)\right)&\text{by
                                                            \cref{new_measures}}
                                                            \\                                                      =&\BP^\PP_{\start{x}}\left(Y_j=x\right)
                                                            \cdot \left(\sum \limits_{\pi
                                                              = \run{z_0, \ldots, z_{j}}
                                                             \in
                                                              \IZ_{>0}^{j+1}}\BP^\PP_{\start{x}}\left(\Cyl^\IZ(\pi)\right)\right)&\text{by
                                                            \cref{new_measures}}\\
                                                            =& \BP^\PP_{\start{x}}\left(Y_j=x\right) \cdot
                                                         \BP^\PP_{\start{x}}\left(\biguplus
                                                         \limits_{\pi = \run{z_0, \ldots,
                                                           z_{j}} \in
                                                           \IZ_{>0}^{j+1}}\Cyl^\IZ(\pi)\right)&\text{as $\BP^\PP_{\start{x}}$ is a prob.\ measure}\\ 
                                                            =&\BP^\PP_{\start{x}}\left(Y_j=x\right)
                                                            \cdot
                                                            \BP^\PP_{\start{x}}\left(\II_{\{T^{\PP} \geq j+1\}} = 1\right)
                                                            \hspace*{6.05cm} \qed  \hspace*{-6.05cm}
  \end{array}\]
\end{proof}

\noindent{}Now we can use \cref{walds_lemma} to infer linear upper and lower bounds for the expected
runtime if the random walk program $\PP$ is PAST (i.e., if  $\mu_\PP < 0$).

\begin{restatable}[Bounds on the Expected Runtime of Random Walk Programs]{theorem}{theoremboundsontheruntime}\label{app:Bounds on the Runtime of Programs}
 Let $\PP$ be a random walk program as in
   Def.\ \hyperlink{def-univariate-program}{\arabic{def-univariate-program}} without direct
  termination where 
$\mu_\PP<0$. Then
$rt_{\start{x}}^\PP = 0$ for $\start{x} \leq 0$ and 
  for $\start{x} > 0$,  we have
  \[ -\tfrac{1}{\mu_\PP} \cdot \start{x} \quad \leq \quad rt_{\start{x}}^\PP \quad
  \leq  \quad -\tfrac{1}{\mu_\PP} \cdot \start{x} + \tfrac{1-k}{\mu_\PP}.
\]
So for $\start{x}>0$,  $\PP$'s expected runtime is asymptotically linear, i.e., $rt_{\start{x}}^\PP \in \Theta(\start{x})$.
\end{restatable}

\begin{proof}
All prerequisites are satisfied to apply Wald's Lemma (\cref{walds_lemma}) for the
stochastic process  $\mathbf{Y}^\IZ$ on the probability space $(\IZ^{\omega}, \F{F}^\IZ,
\BP^\PP_{\start{x}})$ and the termination time $T^{\PP}$:  
By \cref{i_i_d},  $\mathbf{Y}^\IZ$ is i.i.d.\
w.r.t.\  $(\IZ^{\omega}, \F{F}^\IZ,
\BP^\PP_{\start{x}})$
and  $\BE^\IZ_{\start{x}}(Y_0^\IZ) =
\mu_\PP < \infty$.
  Since $\mu_\PP < 0$, \cref{termination_decidable} yields
  that $\PP$ is PAST and hence $rt^\PP_{\start{x}} = \expecP{\start{x}}{T^{\PP}}< \infty$.
  By \cref{ert_unaffected} this implies  $\BE^\PP_{\start{x}}\left(T^{\PP}\right)=\expecP{\start{x}}{T^{\PP}}<
  \infty$.
  Furthermore,  $Y_j^\IZ$ and $\II_{\{T^{\PP} \geq j+1\}}$ are independent by
   \cref{independence_terminationtime_differences}. Thus, 
 \cref{walds_lemma} yields
  \begin{equation}
    \label{Wald Consequence}
\BE^\PP_{\start{x}}\left(\sum_{0 \leq j \leq T^{\PP}-1}
Y_j^\IZ\right)\; = \;\BE_{\start{x}}^\PP(T^{\PP})\cdot \BE_{\start{x}}^\PP(Y_0^\IZ). 
  \end{equation}

  \noindent{}Let the random variable $X_{T^{\PP}}:\Omega\!\to\!\IZ$ map every run $\pi$ to the
  first non-positive value in $\pi$, i.e., to the value of the program variable when $\PP$
terminates, or 0 otherwise. So
  $X_{T^{\PP}}(\pi) = X_{T^{\PP}(\pi)}(\pi)$ if $T^{\PP}(\pi)\!<\!\infty$
and $X_{T^{\PP}}(\pi) = 0$ if $T^{\PP}(\pi)\!=\!\infty$.

To infer linear bounds on the expected value of the termination time
    $\BE^\PP_{\start{x}}(T^{\PP})$ resp.\ $\expecP{\start{x}}{T^{\PP}}$,
    we first infer bounds on $\BE^\PP_{\start{x}}(X_{T^{\PP}})$.
     Clearly, we have $X_{T^{\PP}}(\pi) \leq 0$ for every $\pi \in \Omega$ by the definition of
     the termination time and of $X_{T^{\PP}}$.
Hence, this implies $\BE^\PP_{\start{x}}(X_{T^{\PP}}) \leq 0$, i.e., $0$ is an upper bound for
$\BE^\PP_{\start{x}}(X_{T^{\PP}})$.

To infer a lower bound for $\BE^\PP_{\start{x}}(X_{T^{\PP}})$,
note that if $\start{x} > 0$, then for every run $\pi = \run{z_0, \ldots, z_{j-1}, z_j,
\ldots}$ where $\BP^\PP_{\start{x}}(\Cyl^\IZ(\pi)) = q^\PP_{\start{x}}(\pi) > 0$ and $z_j$ is
the first non-positive value in $\pi$, we have $j \geq 1$ and $z_j$ is at most $k$ smaller
than $z_{j-1}$. Thus, $z_{j-1} \geq 1$ implies $z_j \geq z_{j-1}-k \geq 1-k$. Hence, for
all these runs we have $X_{T^{\PP}}(\pi) = z_j \geq 1-k$.
Moreover, for runs $\pi$ without non-positive values, we also have 
$X_{T^{\PP}}(\pi)  \geq 1-k$,
since $X_{T^{\PP}}(\pi) = 0$ and since $\mu_\PP < 0$ implies $k \geq 1$. Thus, we obtain
$\BE^\PP_{\start{x}}(X_{T^{\PP}}) \geq 1-k$ whenever $\start{x} > 0$.

So to summarize, we get the following upper and lower bounds for
$\BE^\PP_{\start{x}}(X_{T^{\PP}})$ if $\start{x} > 0$:
    \begin{equation}
      \label{value at program end}
      1-k \quad \leq \quad \BE^\PP_{\start{x}}(X_{T^{\PP}})\quad \leq \quad 0
    \end{equation}

    Recall that for every $j \geq 0$ we have
 $X_{j}^\IZ = X_0^\IZ+\sum \nolimits_{0 \leq u \leq j-1} Y_u^\IZ$.
Hence, we also have 
    $X_{T^{\PP}} = X_0^\IZ+\sum \nolimits_{0 \leq u \leq T^{\PP}-1} Y_u^\IZ$.  This implies:
 \[ \begin{array}{rcll}
   \BE^\PP_{\start{x}}(X_{T^{\PP}}) &=&\BE^\PP_{\start{x}}(X_0^\IZ) + \BE^\PP_{\start{x}}\left(\sum_{0
     \leq u \leq T^{\PP}-1} Y_u^\IZ\right)&\\
   &=&\start{x} + \BE^\PP_{\start{x}}(T^{\PP}) \cdot \BE^\PP_{\start{x}}(Y_0^\IZ) &\text{by \cref{Wald Consequence}}\\
   &=&\start{x} + \BE^\PP_{\start{x}}(T^{\PP})\cdot \mu_\PP &\text{by \cref{i_i_d}}\\
    &=&\start{x} + \expecP{\start{x}}{T^{\PP}}\cdot \mu_\PP &\text{by \cref{ert_unaffected}.} 
 \end{array}\]

 \noindent
 Hence, by \cref{value at program end} we obtain
  $-\tfrac{1}{\mu_\PP} \cdot \start{x} \; \leq \; \expecP{\start{x}}{T^{\PP}} \; \leq
  \;  -\tfrac{1}{\mu_\PP} \cdot \start{x} + \tfrac{1-k}{\mu_\PP}$  for any $\start{x} > 0$.
This implies the theorem, since  $rt^\PP_{\start{x}} = \expecP{\start{x}}{T^{\PP}}$.
 \qed 
\end{proof}

\thmboundsgeneralized*
\begin{proof}
  The result directly follows from \cref{transformation_preserves_behavior,app:Bounds on the Runtime of Programs}.
\end{proof}

\section{Proofs for \cref{sec:exact}}\label{app:exact}
\lemmanumberroots*
\begin{proof} 
  \hypertarget{sec:exactProofs}{We}
     use Rouch\'e's Theorem: For a univariate polynomial
     $a_v \cdot x^v + \ldots + a_1 \cdot x +
  a_0$, if there is a number $w \in \IR_{>0}$ and an index $u\in \IN$ with $0 \leq u \leq v$
  such that
\begin{equation}|a_u| \cdot w^u \; >\;\sum_{0 \leq j \leq v,\; j\not=u}|a_j| \cdot w^j,\label{ineq}
	\end{equation}
then the polynomial has exactly $u$ (possibly complex) roots  (counted with multiplicity)
of absolute value less than  $w$.

We now apply Rouch\'e's Theorem to the characteristic polynomial
and proceed by case analysis. First, we consider the case where $p'>0$. Here, we choose $w=1$ and $u=k$.
Then \cref{ineq} becomes 
\[|p_0 -1| \; > \; \sum_{-k \leq  j \leq m, \; j\neq 0}|p_j|.\]
As $|p_0 - 1| = 1 - p_0$ and $|p_j| = p_j$ for all $j$, this is equivalent to
\[1 \; > \; \sum_{-k \leq j \leq m}p_j \; = \; 1-p' \]
which is true since $p'>0$.
So by Rouch\'e's Theorem, the characteristic polynomial $\chi_\PP$ has $k$ roots $\lambda$ with
$|\lambda| < 1$.

However, we would like to conclude that there are no more than $k$ roots $\lambda$ with
$|\lambda| \leq 1$. Thus, we still need to show that $\chi_\PP$ has no root $\lambda$ with
$|\lambda|=1$.
Clearly, $0 = \chi_\PP(\lambda)$ is equivalent to $0 =\sum\nolimits_{-k \leq j \leq m} p_j \cdot
\lambda^{k+j}-\lambda^k$.
If $|\lambda| = 1$ were true, then
$1 = \sum\nolimits_{-k \leq j \leq m} p_j \cdot
\lambda^{j}$ and
\[  1 \; = \; |1| \; \le \; \sum\limits_{-k \leq j\leq m}|p_j|\cdot |\lambda|^j \; = \; \sum\limits_{-k \leq j\leq m} p_j \; = \;
1-p'\]
by using $|p_j| = p_j$. However, this is  a contradiction to $p'>0$.

Now we consider the case where 
	$p'=0$ and thus $\sum\nolimits_{-k \leq j \leq m} p_j=1$.
Our goal is to show that for all small enough $\epsilon > 0$, the inequality \cref{ineq}
holds if we set $w=1+\epsilon$ and $u=k$.	Then \cref{ineq} becomes 
	\[
	|p_0 -1| \cdot w^k \; >\; \sum_{-k \leq j \leq m, \; j\neq 0}|p_j| \cdot w^{j+k}.\]
        As $|p_0 - 1| = 1 - p_0$, $w = 1 + \epsilon$, and $|p_j| = p_j$ for all $j$, this is equivalent to
\[ 	(1-p_0) \cdot (1+\epsilon)^k \; > \; \sum_{-k \leq j \leq m,\;j\neq 0} p_j \cdot
(1+\epsilon)^{j+k}.
\]
Note that\footnote{This notation means that $(1+\epsilon)^j=1+j \cdot \epsilon+
  f(\epsilon)$ for a function $f$ with $f(x) \in \OO(x^2)$. Here, $k$, $m$, and the $p_j$ are
  considered to be constants, i.e.,
we write $\OO(\epsilon^2)$ instead of $(1 - p_0) \cdot \OO(\epsilon^2)$ or $\sum\nolimits_{-k \leq
  j \leq m,\; j\neq 0} p_j\cdot \OO(\epsilon^2)$.}
$(1+\epsilon)^j=1+j \cdot \epsilon+\OO(\epsilon^2)$ for any $j \geq 0$. Hence, we obtain
  \[\begin{array}{ll}
  &(1-p_0)+k \cdot (1-p_0) \cdot \epsilon+\OO(\epsilon^2)\\
  >&\left(\sum_{-k \leq j \leq m, \; j\neq 0} \hspace*{-.2cm} p_j \right)  +  \left(\sum_{-k \leq j \leq
    m,\; j\neq 0} \hspace*{-.2cm} p_j \cdot (j+k) \cdot
    \epsilon\right)  +\OO(\epsilon^2).
  \end{array}
  \]
  By using 
  $\sum\nolimits_{-k \leq j \leq m} p_j=1$, this simplifies to
  \[ k \cdot (1-p_0) \cdot \epsilon+\OO(\epsilon^2) \; >\;
\left( \sum_{-k \leq j \leq m,\; j\neq 0} p_j \cdot (j+k) \cdot
\epsilon \right)+\OO(\epsilon^2).\]
When dividing by $\epsilon > 0$, we get
\[                        
k \cdot (1-p_0)+\OO(\epsilon) \;> \; \left(\sum_{-k \leq j \leq m, \; j\neq 0} p_j \cdot
(j+k) \right)+\OO(\epsilon).\]
To satisfy this, it is sufficient to have 
\[                        
k \cdot (1-p_0) \;> \; \left(\sum_{-k \leq j \leq m, \; j\neq 0} p_j \cdot
(j+k) \right)+\OO(\epsilon).\]
This is equivalent to
\[\begin{array}{rcl}	k &>&\sum_{-k \leq j \leq m}p_j \cdot (j+k)+\OO(\epsilon)\\
	&=&\sum_{-k \leq j \leq m}p_j \cdot j+k+\OO(\epsilon)\\
		&=&\mu_\PP+k+\OO(\epsilon).
			\end{array}\]
Since $\mu_\PP<0$ as $\PP$ is PAST (cf.\ \cref{termination_decidable}), this is true for
all sufficiently small $\epsilon$. Hence, there are exactly $k$ roots of absolute value less than $1+\epsilon$, where $\epsilon$ is sufficiently small, so in particular $k$ roots of absolute value $\le 1$.
\qed 
\end{proof}

\lemmadisregarding*
\begin{proof}
To
 encode the requirement on the $a_{j,u}$,
we  modify
\eqref{constraint1} into a new constraint \eqref{constraint01} which ensures $a_{j,u} = 0$ whenever
$|\lambda_j| > 1$.
More precisely, this new constraint \eqref{constraint01} is a recurrence equation such
that  the
characteristic polynomial $\chi_0$ of its homogeneous part has all the roots of $\chi_\PP$ except those whose absolute
value is greater than 1, i.e., $\chi_0(\lambda) = \prod\nolimits_{1 \leq j \leq c, \; |\lambda_j| \leq
  1} (\lambda-\lambda_j)^{v_j}$.
Thus, we can define the coefficients $q_j \in \IC$ by
	\[\chi_0(\lambda) \quad =\quad \prod_{1 \leq j \leq c,\;  |\lambda_j| \leq 1} (\lambda-\lambda_j)^{v_j}\\
		\quad =\quad \lambda^k-\sum_{-k \leq j \leq -1}q_j \cdot \lambda^{k+j}.\]
        Note that the degree of the polynomial $\chi_0$ is indeed $k$, because by \cref{Number of Roots Lemma} we have
        $\sum\nolimits_{1 \leq j \leq c,\;  |\lambda_j| \leq 1} v_j = k$.

Moreover, the constant add-on of the new recurrence equation is
constructed in such a way that the particular solutions $C_{const}$ resp.\ $C_{lin} \cdot
x$ of \eqref{constraint1} are also solutions of the inhomogeneous recurrence equation. Thus,
 let $D_{const} =
C_{const} \cdot
        \left(1-\sum\nolimits_{-k \leq j \leq -1}q_j\right)$ and 
        $D_{lin} 
        =-C_{lin} \cdot \sum\nolimits_{-k \leq j \leq -1}j \cdot q_j$.
        Instead of \eqref{constraint1}, we now consider the constraint
        \begin{equation}
     		\label{constraint01} f(x) \; = \; \sum_{-k \leq j \leq -1} q_j \cdot f(x+j)+D
                \quad \text{for
                  all $x > 0$,}
\end{equation}
        where we choose $D = D_{const}$ if $p'>0$ and $D = D_{lin}$ if $p'=0$.
      We show the following two claims:
      	\begin{enumerate}[(a)]
	\item There is exactly one function $f:\IZ \to \IC$ which satisfies
         \cref{constraint2} and \cref{constraint01}.
	\item A function $f:\IZ \to \IC$  satisfies  \cref{constraint01} iff
 $f$  satisfies  \cref{constraint1}  (thus, it has
 the form \eqref{sol}) where $a_{j,u} = 0$ whenever
$|\lambda_j| > 1$.
          \end{enumerate}
        These two claims imply the statement of the lemma. To see this, note
        that by (a) there exists a function which satisfies \cref{constraint2} and \cref{constraint01}
and by (b) this function also satisfies \cref{constraint1} and it has $a_{j,u} = 0$ whenever
$|\lambda_j| > 1$. This function is unique, because if there were two different functions $f_1$ and
$f_2$ that satisfy  \cref{constraint2} and \cref{constraint1} and have $a_{j,u} = 0$ whenever
$|\lambda_j| > 1$, then by (b) these two functions would also both satisfy
\cref{constraint01}. But this would be a contradiction to the uniqueness stated in (a).

We now prove the claims (a) and (b). For (a),
note that 
the recurrence equation \eqref{constraint01} is formulated in such a way that $f(x)$
only depends on the values of $f$ on the smaller values $x-1, \ldots, x-k$ (i.e., it is a
recurrence of order $k$). By
 the constraint \eqref{constraint2}, the initial value of $f$ on
negative values is uniquely determined (i.e., $f(0) = f(-1) = \ldots = f(-k+1) =
0$). Hence, by induction on $x$, one can easily prove that there is a single unique function $f:\IZ \to \IC$ that satisfies both
\eqref{constraint2} and \eqref{constraint01}.

For the claim (b), we only have to show that $C_{const}$ is a solution of the inhomogeneous
recurrence equation \cref{constraint01} if $p' > 0$ and $C_{lin} \cdot x$ is a solution of
\cref{constraint01} if $p' = 0$.
Once this is shown, it is clear that all solutions of  \cref{constraint01} result from
adding the particular solution $C_{const}$ resp.\ $C_{lin} \cdot x$  of the
inhomogeneous equation to a solution of the homogeneous variant of \cref{constraint01}
(where $D$ is replaced by 0). Any solution of this homogeneous variant can be represented
as a linear combination of the  solutions $\lambda_j^x \cdot x^u$
where $|\lambda_j| \leq 1$ and $u \in \{0, \ldots, v_j-1\}$. That these are linearly independent solutions
of the homogeneous variant of \cref{constraint01} follows from the fact that $\chi_0$ is
the corresponding characteristic polynomial. Thus,  the solutions of
\cref{constraint01} are all functions of the form 
\eqref{sol} where $a_{j,u} = 0$ whenever
$|\lambda_j| > 1$, which proves (b).

It remains to show that  $C_{const}$  resp.\  $C_{lin} \cdot x$ are particular solutions
of the  inhomogeneous
recurrence equation \cref{constraint01}. If $p' > 0$, then the definition of $D_{const}$
indeed implies
$C_{const} =C_{const} \cdot \sum\nolimits_{-k \leq j \leq -1}q_j +D_{const}$. If $p' = 0$, then we have to show
\begin{equation}
  \label{inhomSolutionConstraint01} C_{lin} \cdot x \; = \;C_{lin} \cdot \sum_{-k \leq j \leq -1}q_j
  \cdot (x+j) + D_{lin}.
  \end{equation}
Since 1 is a root of $\chi_\PP$ (i.e., one of the $\lambda_j$ with $|\lambda_j| \leq
1$ is $\lambda_j = 1$), 1 is also a root of $\chi_0$. So we have $0 = \chi_0(1) = 1 -
\sum\nolimits_{-k \leq j \leq -1}q_j$, which implies $\sum\nolimits_{-k \leq j \leq -1}q_j = 1$.
So \eqref{inhomSolutionConstraint01} is equivalent to
\[ \begin{array}{rcl}
C_{lin} \cdot x &=& 
C_{lin} \cdot (x \cdot \sum_{-k \leq j \leq -1}q_j
+  \sum_{-k \leq j \leq -1}j \cdot q_j) + D_{lin}\\
&=&
C_{lin} \cdot (x 
+  \sum_{-k \leq j \leq -1}j \cdot q_j) + D_{lin}.
\end{array}\]
This holds due to the definition of
$D_{lin}$.
\qed
\end{proof}

\theoremsolution*
\begin{proof}
By  \cref{correctness_of_ert}, the expected runtime $rt^\PP_x$
is the least fixpoint of the  expected runtime transformer
$\C{L}^\PP$, i.e.,
 the smallest function
$f(x):\IZ \to \overline{\IR_{\geq 0}}$
which satisfies \cref{constraint1OLD}, or equivalently,
 the smallest function
which satisfies \cref{constraint2} and \cref{constraint1}.

Since $f$ satisfies \cref{constraint1}, it is a function of the form \cref{sol},
i.e., there exist coefficients $a_{j,u} \in \IC$ such that for all $x > -k$ we have
\[ 
f(x) \;=\; C(x) \, + \, \sum_{1 \leq j \leq c} \;\; \sum_{0\leq u \leq v_j-1}a_{j,u} \cdot
\lambda_j^x \cdot x^u.
\]

\noindent
If we had $a_{j,u}\not=0$ for a coefficient where
$|\lambda_j| > 1$, then
$f(x)$ would not be bounded by a constant (if $p' > 0$) resp.\ by a linear function
(if $p' = 0$).
  Thus, this would contradict
\Cref{PAST}   (if $p' > 0$) resp.\  \Cref{bounds_runtime_constant_probability_programs} (if $p' = 0$).

By \cref{reduction} there is a single unique function $f: \IZ \to \IC$ which satisfies both
\cref{constraint2} and \cref{constraint1}
and has $a_{j,u}=0$  whenever
$|\lambda_j| > 1$. So this function must be the expected runtime (and hence, it maps any integer
to a non-negative real number). Due to \cref{constraint1} the function must be of the form
\cref{sol} for all $x > -k$ but at the same time it also has to satisfy $f(x) = 
0$ for all $x \leq 0$ due to \cref{constraint2}. Therefore, it must satisfy the linear
equations \eqref{initial}.
On the other hand,  the linear equations \eqref{initial}
cannot have more than one solution because otherwise this would yield
two different functions 
that satisfy both
\cref{constraint2} and \cref{constraint1}
and have $a_{j,u}=0$  whenever
$|\lambda_j| > 1$, in contradiction to \cref{reduction}.

If $k=0$, then $p' > 0$ as $\PP$ is PAST.  \cref{Number of Roots Lemma}
implies that $\chi_\PP$ has no root with $|\lambda| \leq 1$ and thus, $rt^\PP_x
= C_{const}   + \sum\nolimits_{1 \leq j \leq c, \; |\lambda_j|\le 1} \ldots =  C_{const}$
for $x > 0$.
\qed
\end{proof}

\coroexactruntime*
\begin{proof}
  If $\PP$ is trivial, then its expected runtime is obvious. Otherwise, by \cref{coro1} one can
  decide if
$\PP$ is PAST and in that case, $\PP^\unisup$ is PAST as well. 
  For any CP program $\PP$, we have
    $rt_{\vec{x}}^\PP = rt_{\uni_{\PP}(\vec{x})}^{\PP^\unisup}$ due to
  \cref{transformation_preserves_behavior}.
  As $rt_{\uni_{\PP}(\vec{x})}^{\PP^\unisup}$ can be computed exactly by \cref{solution},
  this also holds for   $rt_{\vec{x}}^\PP$. \qed
\end{proof}

\noindent{}
As mentioned in \cref{sec:exact},
 \cref{solution,coro3} imply that for any $\startvec{x} \in \IZ^r$,
 the expected runtime $rt^\PP_{\startvec{x}}$ of a CP
program $\PP$ that is PAST and has only \emph{rational} probabilities $p_{\vec{c}_1},\ldots,p_{\vec{c}_n},p' \in \mathbb{Q}$
is always an algebraic number.
This is due to the fact that $rt^\PP_{\startvec{x}}$ can be represented as a
linear combination of algebraic numbers (the roots of the characteristic polynomial $\chi_{\PP^\unisup}$). The
coefficients of this linear combination are the solution of a linear equation system \cref{initial} over algebraic
numbers and hence, they are algebraic numbers themselves. 
Therefore, one could also compute a closed form for the
exact expected runtime $rt^\PP_{\vec{x}}$ using a representation with algebraic numbers
instead of numerical approximations.

As also discussed in \cref{sec:exact}, while the
exact computation of the expected runtime of a random walk program
$\PP$ according to \cref{solution} may yield a
representation of $rt^\PP_x$ with possibly complex number, one can easily obtain a
more intuitive representation of $rt^\PP_x$ that uses real numbers only.

As stated before, for any coefficients $a_{j,u}, a_{j,u}' \in \IC$
with $j \in \{s+1,\ldots,s+t\}$ and
$u \in \{0,\ldots,
v_{j}-1\}$ there exist coefficients $b_{j,u}$ and
$b_{j,u}'$ such that
\[a_{j,u} \cdot \lambda_{j}^x
 + a_{j,u}' \cdot \overline{\lambda_{j}}^x
 \quad = \quad b_{j,u} \cdot \mathrm{Re}(\lambda_{j}^x)
 + b_{j,u}' \cdot \mathrm{Im}(\lambda_{j}^x)
\]
holds for all $x \in \IZ$. More precisely, $b_{j,u} = a_{j,u} + a_{j,u}'$
and $b_{j,u}' = (a_{j,u} - a_{j,u}') \cdot i$.
So  any linear combination of the functions $\lambda_j^x \cdot
x^u$ and $\overline{\lambda_j}^x \cdot
x^u$ can be replaced by a linear combination of the 
functions $\mathrm{Re}(\lambda_j^x) \cdot x^u$ and $\mathrm{Im}(\lambda_j^x)
\cdot x^u$.
In this way, one obtains $k+m$ linearly independent \emph{real} solutions of the
corresponding homogeneous recurrence equation.
Hence, by \cref{solution} we now get the representation of the expected runtime in
\eqref{final}:
\[
  \mbox{\small \hspace*{-.2cm}$rt^\PP_x\!=\!\left\{\begin{array}{lll}
 C(x) &+ \hspace*{-.5cm}
\sum_{1 \leq j \leq s, \; |\lambda_j|\le 1} \; \; \;\sum_{0\leq u \leq v_j-1}a_{j,u} \cdot
\lambda_j^x \cdot x^u\\
&+  \hspace*{-.5cm}
\sum_{s+1 \leq j \leq s+t, \; |\lambda_{j}|\le 1} \; \sum_{0 \leq u \leq
  v_{j}-1} \hspace*{-.3cm} \left(b_{j,u}\!\cdot\!\mathrm{Re}(\lambda_{j}^x)  + 
b_{j,u}'\!\cdot\!\mathrm{Im}(\lambda_{j}^x)\right) \cdot x^u, & \text{for $x > 0$}\\
0, && \text{for $x \leq 0$}
\end{array}\hspace*{-.2cm}
\right.$}\]

\noindent{}Since $rt^\PP_x$ is real-valued,
$\lambda_j^x \in \IR$ for $j \in \{1,\ldots,s\}$, and $\mathrm{Re}(\lambda_{j}^x),
\mathrm{Im}(\lambda_{j}^x) \in \IR$ for   $j \in \{
s+1, \ldots, s+t \}$,  all $a_{j,u}$ for $j \in \{1,\ldots,s\}$ and all $b_{j,u}, b_{j,u}'$ for $j \in \{
s+1, \ldots, s+t \}$ are real numbers.
As $b_{j,u} = a_{j,u} + a_{j,u}'$, this means that $a_{j,u}'$ is the conjugate
 of $a_{j,u}$, i.e., $a_{j,u}' = \overline{a_{j,u}}$ and thus, $b_{j,u} = 2
\cdot \mathrm{Re}(a_{j,u})$ and $b_{j,u}' = -2
\cdot \mathrm{Im}(a_{j,u})$.

As mentioned, to compute   $\mathrm{Re}(\lambda_{j}^x)$ and $\mathrm{Im}(\lambda_{j}^x)$,
we consider the polar representation of the non-real roots
$\lambda_{j}$, i.e., for $j \in \{s+1,\ldots,s+t\}$ let $\lambda_{j} = 
w_j\cdot e^{\theta_j \cdot i}$
with $w_j\in\mathbb R_{>0}$ and $\theta_j\in (0,2\pi)$.
Then $\lambda_{j}^x=w_j^x \cdot e^{\theta_j \cdot i \cdot x}$,
and $\mathrm{Re}(\lambda_{j}^x) = w^x_j \cdot \cos(\theta_j \cdot x)$ and
$\mathrm{Im}(\lambda_{j}^x) =w^x_j \cdot \sin(\theta_j \cdot x)$.

Note that in \cref{algorithm_exact}, one could
        also already use the representation in \cref{final}
        with $\mathrm{Re}(\lambda_{j}^x) = w^x_j \cdot \cos(\theta_j\cdot x)$ and
    $\mathrm{Im}(\lambda_{j}^x) =w^x_j \cdot \sin(\theta_j \cdot x)$ here. Then one would only
        have to solve a system of linear equations over the reals and can compute $b_{j,u}$
        and $b_{j,u}'$ directly. 

}
\end{document}